\documentclass{article}
\usepackage{graphicx}
\usepackage[]{inputenc}
\usepackage[T1]{fontenc}
\usepackage[pdfencoding=auto, psdextra]{hyperref}
\usepackage{bookmark}

\usepackage{color}

\usepackage{fullpage}

\usepackage{mathtools}
\usepackage{empheq}
\usepackage{amssymb}
\usepackage[toc,page]{appendix}

\usepackage{soul}
\usepackage[english]{babel}
\usepackage{amsthm}
\usepackage{float}
\usepackage{wrapfig}
\usepackage{caption}
\usepackage{graphicx}
\usepackage{pgfplots}
\usepackage{subcaption}
\usepackage{tikz,tikz-3dplot}
\tdplotsetmaincoords{80}{45}
\tdplotsetrotatedcoords{-90}{180}{-90}

\usetikzlibrary{intersections,fadings,decorations.pathreplacing}

  \pgfplotsset{
        compat=1.8}

\newtheorem{theorem}{Theorem}[section]
\newtheorem{lemma}[theorem]{Lemma}

\theoremstyle{definition}

\newcommand{\LL}{\ensuremath{\mathcal{L}}}

\newcommand{\NN}{\ensuremath{\mathcal{N}}}
\newcommand{\DD}{\ensuremath{\mathcal{D}}}
\newcommand{\fD}{\ensuremath{\mathfrak{D}}}

\def \be{\begin{equation}}
\def \ee{\end{equation}}
\def \ba{\begin{eqnarray}}
\def \ea{\end{eqnarray}}

\title{The 0th and 2nd Laws of Black Hole Mechanics in Einstein-Maxwell-Scalar Effective Field Theory}

\author{Iain Davies\\
{\small Department of Applied Mathematics and Theoretical Physics, University of Cambridge,} \\ {\small Wilberforce Road, Cambridge CB3 0WA, United Kingdom} \\
{\small id318@cam.ac.uk}
}

\begin{document}

\maketitle

\begin{abstract}

There has been recent progress in extending the 0th and 2nd laws of black hole mechanics to gravitational effective field theories (EFTs).  We generalize these results to a much larger class of EFTs describing gravity coupled to electromagnetism and a real scalar field. We also show that the 0th law holds for the EFT of gravity coupled to electromagnetism and a charged scalar field.

\end{abstract}

\tableofcontents

\section{Introduction}

The Laws of Black Hole Mechanics are a set of theorems determining the classical properties of black holes. Their striking resemblance to the Laws of Thermodynamics leads to an interpretation of black holes as thermodynamic objects, which is made concrete through the mechanism of Hawking radiation.

The original proofs of the laws \cite{Hawking:1973} \cite{Bardeen:1973} require that the theory of gravity is the 2-derivative Einstein-Hilbert action, with a theory of matter, such as Maxwell theory or a minimally coupled scalar field, that satisfies suitable energy conditions. However, we know that 2-derivative Einstein-Maxwell theory cannot be the complete description of gravity and electromagnetism on all scales as it is not a UV complete theory. Generically we expect any low energy limit of a UV theory of gravity to come with higher derivative corrections, which will invalidate the standard proofs of the Laws of Black Hole Mechanics. Since we do not expect these corrections to change the physical interpretation of black holes as thermodynamic objects, this is a problem.

There have been a variety of attempts to reconcile the laws in higher derivative theories of gravity, some of which are reviewed by Sarkar in 2019 in \cite{Sarkar:2019}. In particular, Wald proved in \cite{Wald:1993nt} that a modified, but still geometric definition, of black hole entropy could be used to prove the "equilibrium state" version of the 1st Law in any diffeomorphism-invariant theory of gravity with arbitrary matter fields. However, this definition of the entropy fails to satisfy a 2nd Law, and he had to assume that the 0th Law holds via the assumption of a bifurcate Killing horizon.

Recently however, there were several developments in the proving the 0th, 1st and 2nd laws in the setting of effective field theory (EFT). This setting requires assuming two things: (a) the Lagrangian is a series of terms with increasing derivatives coming with coefficients that scale in appropriate powers of some UV length scale $l$, and (b) any time or length scale $L$ associated with the solution satisfies $L\gg l$. The first assumption is physically reasonable if we view our theory as a low energy limit of some UV complete theory of gravity. The second assumption means that higher derivative terms are less important, and so our solution remains in the regime of validity of the EFT.  

Let us briefly review the spate of recent results.

\textbf{0th Law}: Bhattacharyya et al \cite{Bhat:2022} proved that the 0th Law holds for any diffeomorphism-invariant EFT of gravity without matter. The 0th Law states that the surface gravity, $\kappa$, of a stationary black hole is constant across the horizon\footnote{Note that in order to define the surface gravity, the horizon of the black hole must be a Killing horizon. Hawking proved this is always the case for 2-derivative GR in his rigidity theorem \cite{Hawking:1972}, under the assumption of analyticity. Recent work by Hollands et al \cite{Hollands:2022ajj} has proved the rigidity theorem also holds in the EFT of gravity with no matter, also under the assumption of analyticity.}. The proof uses Gaussian Null Co-ordinates and the concept of "boost weight" to show that derivatives of $\kappa$ tangent to the horizon are proportional to a component of the equations of motion, which is set to 0. 

\textbf{1st Law}: Biswas, Dhivakara and Kundu \cite{Biswas:2022} generalized the Wald entropy \cite{Wald:1993nt} to prove the "physical process" version of the 1st Law for an arbitrary higher derivative diffeomorphism-invariant, gauge-independent theory of gravity, electromagnetism and a real scalar field. The physical process version of the 1st Law concerns a stationary black hole that is perturbed by some matter before settling down to a new stationary configuration. It relates the change in entropy $\delta S$ to the mass $\delta M$, angular momentum $\delta J$ and charge $\delta Q$ of the matter perturbation:
\be
    \frac{\kappa}{2 \pi} \delta S = \delta M - \Omega_H \delta J - \Phi_{bh} \delta Q
\ee
where $\Omega_H$ is the angular velocity of the horizon and $\Phi_{bh}$ is the electrostatic potential.

\textbf{2nd Law}: Hollands, Kovacs and Reall (HKR) \cite{Hollands:2022}, following on from work by Wall \cite{Wall:2015} and Bhattacharyya et al \cite{Bhattacharyya:2021jhr}, proved a version of the 2nd Law for any diffeomorphism-invariant EFT of gravity and a real scalar field. The 2nd Law states that the entropy of a dynamical black hole is non-decreasing in time, $\dot{S}\geq 0$. HKR consider a dynamical black hole settling down to an equilibrium stationary state, and that remains in the regime of validity of the EFT as described above. They were able to define an entropy which is non-decreasing to quadratic order in perturbations around the stationary state, up to $O(l^N)$ terms, where $l^N$ is the order up to which we know our EFT. Furthermore, this entropy reduces to the Wald entropy in equilibrium.

Finally, in the companion to this paper \cite{DaviesReall:2023}, Davies and Reall were able to show that a further extension of the HKR procedure can strengthen the 2nd Law result significantly by dropping its perturbative nature. They define an entropy which satisfies a \textit{non-perturbative} 2nd Law in vacuum gravity EFT, up to $O(l^N)$ terms. This entropy reduces to the Wald entropy in equilibrium, satisfies the 1st Law, and is purely geometrically defined for theories with up to 6 derivatives.

Taken together, these results mean we now have a much better understanding of the Laws of Black Hole Mechanics in EFT. However, the results we have for the 0th Law and 2nd Law are only applicable to gravity with minimal matter couplings, or to gravity with the simplest matter field, a scalar field. Here we ask, are these results robust to the addition of non-minimal couplings of some more complicated matter field? The only field other than the metric and scalar field for which we know the classical approximation may be valid is the Maxwell field, and hence this seems like an important addition to make.

In this paper we extend the aforementioned works by completing the story for the EFT of gravity, electromagnetism and a real (uncharged) scalar field. We prove a Generalized 0th Law holds exactly, and that the 2nd Law holds in the sense of Davies and Reall. Taken all together, this means there is now a 0th, 1st and 2nd Law for such theories. Along the way we will also show the 0th Law still holds even if the scalar field is charged, and discuss how the 2nd Law could be generalized in this case. This gives further evidence that these proofs are robust to more complicated matter models, and that even in higher derivative theories of gravity, black holes will still obey the laws.

The paper is broken down as follows. In Section 2, we define our Einstein-Maxwell-Scalar EFT. In Section 3, we define two distinct choices of Gaussian Null Co-ordinates (GNCs) and the notion of "boost weight". We will work in GNCs throughout the paper. In Section 4, we state the Generalized 0th Law and sketch its proof for Einstein-Maxwell-Scalar EFT. Section 5 contains the details of this proof. We also show how the proof can be modified if the scalar field is charged. In Section 6, we make precise the scenario in which we will prove the 2nd Law, and review the previous work on the matter. In Section 7, we prove the 2nd Law for Einstein-Maxwell-Scalar EFT, and discuss its generalization if the scalar is charged.

\section{Einstein-Maxwell-Scalar EFT}

We consider the EFT of gravity, electromagnetism and a real\footnote{The EFT of a charged, complex scalar field is discussed in Section \ref{ChargedScalar}.} scalar field $\phi$, which we shall refer to as Einstein-Maxwell-Scalar EFT. In EFT, the Lagrangian is a sum of terms ordered by their number of derivatives. We assume diffeomorphism invariance and electromagnetic gauge invariance\footnote{The assumption that the Lagrangian is invariant under an electromagnetic gauge transformation $A_{\alpha}\rightarrow A_{\alpha}+\nabla_{\alpha}{\chi}$ will rule out, for example, Chern-Simons terms, which aren't themselves gauge invariant but do produce gauge invariant equations of motion. See very recent work \cite{Deo:2023} for a linearized 2nd Law for Chern-Simons terms. }, so that the Lagrangian consists only of contractions of $R_{\alpha \beta \gamma \delta}$, $F_{\alpha \beta}$, $\phi$ and their covariant derivatives,

\begin{equation}
    \mathcal{L} = \mathcal{L}(g_{\alpha \beta}, R_{\alpha \beta \gamma \delta}, \nabla_{\alpha}R_{\alpha \beta \gamma \delta}, ... , F_{\alpha \beta}, \nabla_{\alpha}F_{\alpha \beta}, ... , \phi, \nabla_{\alpha}\phi, ...  )
\end{equation}

The most general Lagrangian of this form with up to 2 derivatives can be written as\footnote{We have included the 0-derivative term $V(\phi)$ in the 2-derivative Lagrangian $\LL_2$. Naively in EFT, we should expect $V(\phi)$ to come with a factor $1/l^2$. However, we assume $V(\phi)$ is comparable to the cosmological constant $\Lambda$, which is extremely small for somewhat mysterious reasons. More precisely, we assume $|V|\leq 1/L^2$, where $L$ is any typical length scale of the solution, and so $V(\phi)$ is of no larger scale than the 2-derivative terms.}

\begin{equation} \label{0and2-deriv}
    \mathcal{L}_{2} = R - V(\phi) -\frac{1}{2} \nabla_{\alpha}{\phi} \nabla^{\alpha}{\phi} - \frac{1}{4} c_1(\phi) F_{\alpha\beta} F^{\alpha\beta} + c_2(\phi) F_{\alpha \beta} F_{\gamma \delta} \epsilon^{\alpha \beta \gamma \delta}
\end{equation}
An arbitrary function of $\phi$ multiplying $R$ can be eliminated by redefining the metric, whilst an arbitrary function of $\phi$ multiplying $\nabla_{\alpha}{\phi} \nabla^{\alpha}{\phi}$ can be eliminated by redefining $\phi$ \cite{Weinberg:2008}. Here we have taken units with $16\pi G = 1$ and rescaled $F_{\alpha \beta}$ appropriately. The final term only appears in $d=4$, where the volume form $\epsilon_{\alpha \beta \gamma \delta}$ has 4 indices; in higher dimensions, it is taken that this term is not present. 

The only condition we put on the arbitrary functions $V(\phi), c_1(\phi)$ and $c_2(\phi)$ is that $c_1(\phi)>0$. This is a sufficient condition for the energy-momentum tensor of the leading order 2-derivative theory to satisfy the Null Energy Condition (NEC). For Einstein-Maxwell theory without a scalar field, $c_1=1$, so this positivity condition is also motivated on the grounds that we do not expect the scalar field to change the sign of $c_1$. If we were additionally to impose $V(\phi)\geq 0$ then the 2-derivative energy-momentum tensor would also satisfy the Dominant Energy Condition (DEC), which is the condition assumed in the original proof of the 0th Law by Bardeen, Carter and Hawking. However as noted in e.g. \cite{Dey:2021}, a weaker version called the Null Dominant Energy Condition\footnote{The Null Dominant Energy Condition states that the vector $W^{\alpha} = -T^{\alpha}_{\,\,\beta} V^{\beta}$ is future-directed casual for all future-directed null vectors $V^{\alpha}$.} is in fact sufficient for the proof, which is satisfied by our 2-derivative theory regardless of the sign of $V(\phi)$ (and also includes cases such as Anti de-Sitter that are excluded by the DEC). Indeed, in the following proofs we will require no condition on $V(\phi)$.

In the full EFT action, higher derivative terms come with a factor of some UV scale $l$ for each extra derivative:

\begin{equation} \label{action}
    S = \int \text{d}^d x \sqrt{-g} \left(\LL_{2} + \sum_{n=1}^{\infty} l^{n} \mathcal{L}_{n+2} \right)
\end{equation}
where $\mathcal{L}_{n+2}$ contains all terms with $n+2$ derivatives. 

The equations of motion for this action are $E_{\alpha \beta}=0, E_{\alpha}=0, E=0$, where
\begin{equation} \label{EoMs}
\begin{split}
    E_{\alpha\beta} \equiv \frac{1}{\sqrt{-g}} \frac{\delta S}{\delta g^{\alpha \beta}} = E^{(0)}_{\alpha \beta} + \sum_{n=1}^{\infty} l^{n} E^{(n)}_{\alpha \beta},  \quad & \quad  E_\alpha \equiv -\frac{1}{\sqrt{-g}} g_{\alpha \beta} \frac{\delta S}{\delta A_{\beta}} = E^{(0)}_{\alpha} + \sum_{n=1}^{\infty} l^{n} E^{(n)}_{\alpha}\\
    E \equiv \frac{1}{\sqrt{-g}} \frac{\delta S}{\delta \phi} = & \, E^{(0)} + \sum_{n=1}^{\infty} l^{n} E^{(n)}
\end{split}
\end{equation}
where $E^{(0)}_{\alpha \beta}, E^{(0)}_{\alpha}, E^{(0)}$ are the result of varying the 2-derivative terms from $\LL_{2}$, i.e.
\begin{equation} \label{0th order EoM 1}
    E^{(0)}_{\alpha\beta}= R_{\alpha \beta} - \frac{1}{2} \nabla_{\alpha} \phi \nabla_{\beta} \phi -\frac{1}{2} c_1(\phi) F_{\alpha \delta} F_{\beta}^{\,\,\,\,\delta} - \frac{1}{2} g_{\alpha \beta} \left( R - V(\phi) -\frac{1}{2} \nabla_{\gamma}{\phi} \nabla^{\gamma}{\phi} - \frac{1}{4} c_1(\phi) F_{\gamma \delta} F^{\gamma\delta} \right)
\end{equation}
\begin{equation}
    E^{(0)}_\alpha = \nabla^{\beta}{\Big[c_1(\phi)F_{\alpha\beta} - 4 c_2(\phi) F^{\gamma \delta} \epsilon_{\alpha \beta \gamma \delta}\Big]}
\end{equation}
\begin{equation} \label{0th order EoM 2}
    E^{(0)} = \nabla^{\alpha}{\nabla_{\alpha}{\phi}}-V'(\phi)-\frac{1}{4} c'_1(\phi) F_{\alpha \beta} F^{\alpha \beta} + c'_2(\phi) F_{\alpha \beta} F_{\gamma \delta} \epsilon^{\alpha \beta \gamma \delta}
\end{equation}

\section{Gaussian Null Coordinates}
We will be concerned with quantities on the event horizon $\NN$ of a black hole, which is a null hypersurface. We assume $\mathcal{N}$ is smooth and has generators that extend to infinite affine parameter to the future. The smoothness assumption will always be true in the stationary setting of the 0th Law, but is not generally true when considering dynamical black holes, as in the 2nd Law. However, it seems a reasonable assumption for the situation of a black hole settling down to equilibrium, as envisioned in \cite{Wall:2015}, \cite{Bhattacharyya:2021jhr} and \cite{Hollands:2022}.

To describe quantities near $\NN$, we will use two appropriate choices of Gaussian Null Coordinates (GNCs). The first applies to both the stationary and dynamical setting, whilst the second will only be used in the stationary case to prove the 0th Law.

\subsection{Affinely Parameterized GNCs} \label{APGNCs}

Here we use the same notation as \cite{Hollands:2022} and \cite{Davies:2023}. Assume all generators intersect a spacelike cross-section $C$ exactly once, and take $x^A$ to be a co-dimension 2 co-ordinate chart on $C$. Let the null geodesic generators have affine parameter $v$ and future directed tangent vector $l^\alpha$ such that $l = \partial_{v}$ and $v=0$ on $C$. We can transport $C$ along the null geodesic generators a parameter distance $v$ to obtain a foliation $C(v)$ of $\mathcal{N}$. Finally, we uniquely define the null vector field $n^\alpha$ by $n \cdot (\partial/\partial x^A) =0$ and $n \cdot l = 1$. The co-ordinates $(r, v, x^A)$ are then assigned to the point affine parameter distance $r$ along the null geodesic starting at the point on $\mathcal{N}$ with coordinates $(v,x^A)$ and with tangent $n^\alpha$ there. The metric in these GNCs is given by 
\begin{equation}
    g = 2 \text{d}v \text{d}r - r^2 \alpha(r,v,x^C) \text{d}v^2 -2 r\beta_{A}(r,v,x^C) \text{d}v \text{d}x^A + \mu_{A B}(r,v,x^C) \text{d}x^A \text{d}x^B, \quad l = \partial_{v}, \quad n=\partial_{r}
\end{equation}
This choice of coordinates will be referred to as \textit{affinely parameterized GNCs}. $\mathcal{N}$ is the surface $r=0$, and $C$ is the surface $r=v=0$. The inverse of $\mu_{A B}$ is denoted by $\mu^{A B}$, and we raise and lower $A, B, C, ...$ indices with $\mu^{A B}$ and $\mu_{A B}$. We denote the induced volume form on $C(v)$ by $\epsilon_{A_1 ... A_{d-2}} = \epsilon_{r v A_1 ... A_{d-2}}$ where $d$ is the dimension of the spacetime. The covariant derivative on $C(v)$ with respect to $\mu_{A B}$ is denoted by $D_{A}$. We also define
\begin{equation}
    K_{A B} \equiv \frac{1}{2} \partial_{v}{\mu_{A B}}, \quad \bar{K}_{A B} \equiv \frac{1}{2} \partial_{r}{ \mu_{A B} }, \quad K \equiv K^{A}\,_{A}, \quad \bar{K} \equiv \bar{K}^{A}\,_{A}
\end{equation}
$K_{AB}$ describes the expansion and shear of the horizon generators. $\bar{K}_{AB}$ describes the expansion and shear of the ingoing null geodesics orthogonal to a horizon cut $C(v)$.

Affinely parameterized GNCs are not unique: we are free to change the affine parameter on each generator of $\NN$ by $v'=v/a(x^A)$ with arbitrary $a(x^A)>0$. This will lead to a change $(v,r, x^A)\rightarrow (v',r', x'^A)$ with $v' = v/a(x^A) + O(r)$, $r'=a(x^A)r+O(r^2)$, $x'^A = x^A + O(r)$ near the horizon. Details of how this transformation changes the quantities above are given in \cite{Hollands:2022}. The remaining freedom in our affinely parameterized GNCs is to change our coordinate chart $x^A$ on $C$, however all calculations in this paper are manifestly covariant in $A, B,...$ indices and so this freedom will not change any of the expressions.

\subsubsection{Boost Weight}\label{BW}
An important concept in this set of GNCs is the \textit{boost weight} of a quantity. Suppose we take $a$ to be constant and consider the rescaling $v' = v / a, r'=a r$, which preserves the form of the GNCs above. If a quantity $T$ transforms as $T'=a^b T$, then $T$ is said to have boost weight $b$. See \cite{Hollands:2022} for a full definition. Some important facts are stated here:
\begin{itemize}
    \item A tensor component $T^{\mu_1 ... \mu_n}_{\beta_1 ... \beta_m}$ has boost weight given by the sum of $+1$ for each $v$ subscript and each $r$ superscript and $-1$ for each $r$ subscript and $v$ superscript. $A, B, ...$ indices contribute 0. E.g. $T^{A}_{v v r B}$ has boost weight $+1$.
    \item $\alpha$, $\beta_{A}$, and $\mu_{A B}$ have boost weight 0. $K_{A B}$ and $\bar{K}_{A B}$ have boost weight $+1$ and $-1$ respectively.
    \item  If $T$ has boost weight $b$, then $D_{A_1}{... D_{A_n}{ \partial_{v}^{p} \partial_{r}^q T } }$ has boost weight $b+p-q$.
    \item If $X_{i}$ has boost weight $b_i$ and $T= \prod_{i} X_{i}$, then $T$ has boost weight $b=\sum_{i} b_{i}$.
\end{itemize}

In Lemma 2.1 of \cite{Hollands:2022}, it is proved that boost weight is independent of the choice of affinely parameterized GNCs on $\NN$. More precisely, a quantity of certain boost weight in $(r, v, x^A)$ GNCs on $\NN$ can be written as the sum of terms of the same boost weight in $(r', v', x'^A)$ GNCs on $\NN$, where $v'=v/a(x^A)$ on $\NN$.

\subsection{Killing Vector GNCs} \label{KVGNCs}

In the stationary setting of the 0th Law we will also use another choice of GNCs, hereafter referred to as \textit{Killing vector GNCs}. 

In standard 2-derivative GR with a wide range of matter models, it can be proved that the future event horizon $\NN$ of a stationary analytic black hole spacetime is a Killing horizon whose normal is some Killing vector $\xi$ \cite{Hawking:1973}. Recent work by Hollands et al. \cite{Hollands:2022ajj} has extended this result to arbitrary higher derivative effective field theories of gravity with no matter fields present. Here we assume this result still holds for our Einstein-Maxwell-Scalar EFT, and that we can drop the analyticity assumption.

Therefore we can take a similar construction to the above, except with the null geodesic generators having non-affinely parameterized, future-directed tangent vectors $\xi=\partial_{\tau}$. In the notation of \cite{Bhat:2022}, this leads to coordinates $(\rho, \tau, x^A)$ with metric
\begin{equation}
    g = 2 \text{d}\tau \text{d}\rho - \rho X(\rho,x^C) \text{d}\tau^2 +2 \rho \omega_{A}(\rho,x^C) \text{d}\tau \text{d}x^A + h_{A B}(\rho,x^C) \text{d}x^A \text{d}x^B, \quad \xi = \partial_{\tau}, \quad \chi=\partial_{\rho}
\end{equation}
$\mathcal{N}$ is the surface $\rho=0$, and $C$ is the surface $\rho=\tau=0$. In these co-ordinates we raise $A, B, C,...$ indices with $h^{A B}$ and $h_{A B}$ and denote the induced volume form on $C(\tau)$ by $\varepsilon_{A_1 ... A_{d-2}} = \epsilon_{\rho \tau A_1 ... A_{d-2}}$. The covariant derivative on $C(\tau)$ with respect to $h_{A B}$ is denoted by $\DD_{A}$. 

The differences between the two GNCs are two-fold. Firstly, since $\xi=\partial_{\tau}$ is a Killing vector, the unknown metric coefficients $X, \omega_A$ and $h_{A B}$ are independent of $\tau$. Secondly, the fact that $\tau$ is not necessarily an affine parameter means the coefficient of $\text{d}\tau^2$ only comes with a factor of $\rho$, whereas $\text{d}v^2$ comes with a factor of $r^2$ in the affinely parameterized GNCs.

The relationship between these two forms of GNCs is crucial to prove the 0th Law for this theory, following the method of \cite{Bhat:2022}.

\section{The Generalized 0th Law}
We proceed to prove a Generalized 0th Law of Black Hole Mechanics for this theory. The 0th Law concerns stationary black hole solutions $(g_{\alpha \beta}, F_{\alpha \beta}, \phi)$ to the equations of motion above. 

\subsection{Assumptions}

We make the following assumptions in order to prove the Generalized 0th Law:

1. The rigidity theorem of \cite{Hollands:2022ajj} can be extended to this theory, i.e. that the future event horizon $\NN$ of the black hole is a Killing horizon with Killing vector $\xi$.

2. The matter fields are invariant under this Killing vector, i.e. 
\begin{equation}
    \mathcal{L}_{\xi} F = 0, \quad\quad \mathcal{L}_{\xi} \phi = 0
\end{equation}
In Killing vector GNCs, these imply that $\partial_{\tau} F_{\mu \nu}=0$ and $\partial_{\tau} \phi=0$.

3. The BH solution is analytic in $l$, i.e. we can write
    \begin{equation}
       \begin{split}
           g_{\alpha \beta} &= g_{\alpha \beta}^{(0)} + l g_{\alpha \beta}^{(1)} + l^2 g_{\alpha \beta}^{(2)}+ \ldots\\
           F_{\alpha \beta} &= F_{\alpha \beta}^{(0)} + l F_{\alpha \beta}^{(1)} + l^2 F_{\alpha \beta}^{(2)}+ \ldots\\
           \phi &= \phi^{(0)} + l \phi^{(1)} + l^2 \phi^{(2)}+ \ldots
       \end{split}
    \end{equation}
    
    In particular, we can write the Killing vector GNC metric components as series in $l$: 
\be
\begin{split}
    X&=X^{(0)}+lX^{(1)} + l^2 X^{(2)}+...\\
    \omega_A&=\omega_A^{(0)}+l\omega_A^{(1)} + l^2 \omega_A^{(2)}+...\\
    h_{A B}&=h_{A B}^{(0)}+lh_{A B}^{(1)} + l^2 h_{A B}^{(2)}+...
\end{split}
\ee
    
4. Any spacelike cut, $C$, of the horizon is compact and simply connected. The second assumption implies every closed 1-form on $C$ is exact. These assumptions hold e.g. if $C$ has spherical $S^{d-2}$ topology with $d\geq 4$ but not e.g. if $C$ has the topology of a black ring $S^1\times S^{d-3}$.

\subsection{Statement of Generalized 0th Law}

The surface gravity of the horizon $\NN$ of a stationary black hole is defined by

\begin{equation}
    \xi^\beta \nabla_\beta \xi_\alpha \Big|_\NN = \kappa \xi_\alpha
\end{equation}

The 0th Law of Black Hole Mechanics is the statement that $\kappa$ is constant on $\NN$. One can compute both sides of this equation in the Killing vector GNCs of Section (\ref{KVGNCs}) and find

\begin{equation}
    \kappa = \frac{1}{2} X(\rho, x^C)\Big|_{\rho=0}
\end{equation}

From this we see that $\kappa$ is clearly independent of $\tau$. Therefore, to prove the 0th Law we must show that

\begin{equation}
    \partial_{A} X(\rho, x^C)\Big|_{\rho=0} = 0
\end{equation}

When electromagnetic fields are included in a black hole theory, the 0th Law is usually generalized to include a statement about their behaviour on the horizon. Our \textit{Generalized 0th Law} formulation is 

\begin{equation} \label{0thLawCond}
    \partial_{A} X(\rho, x^C)\Big|_{\rho=0} = 0, \quad \text{and} \quad F_{\tau A}(\rho, x^C)\Big|_{\rho=0} = 0
\end{equation}

The interpretation of the second condition can be seen as follows. By Cartan's formula, $\LL_\xi F = \text{d} (\iota_\xi F ) + \iota_\xi \text{d}F$. $\text{d}F=0$ because $F$ is a Maxwell field. We have also assumed $\LL_\xi F = 0$ above. Hence $\text{d} (\iota_\xi F ) = 0$, and so at least locally, $\iota_\xi F = \text{d} \Phi$ for some scalar $\Phi$. This scalar is the \textit{electric potential} from the definition of the 1st Law. The condition $F_{\tau A}\big|_{\rho=0} = 0$ is then equivalent to $\partial_A \Phi\big|_{\rho=0}=0$, which says that the electric potential is constant on the horizon.

In the course of the proof, we will see that the two conditions in (\ref{0thLawCond}) are not independent. In fact, we will need to show $F_{\tau A}\Big|_{\rho=0} = 0$ in order to prove $\partial_{A} X\Big|_{\rho=0} = 0$.

\subsection{Plan of the Proof}

In \cite{Bhat:2022}, Bhattacharyya et al prove the 0th Law for gravitational EFTs without matter. Here we generalize their method to apply to our Einstein-Maxwell-Scalar EFT. The gravitational parts go through largely unchanged, whilst additional steps are needed to deal with the Maxwell and scalar fields. Here, we sketch the main ideas of the proof.

Let $\Phi_I$ denote the collection of fields $(g_{\mu \nu}, F_{\mu \nu}, \phi)$. We have assumed we can write this as a series in $l$: $\Phi_I= \Phi^{(0)}_I+l\Phi^{(1)}_I+l^2 \Phi^{(2)}_I+...$ . Let $E_I[\Phi_J]$ denote the collection of variations of the action $S$ defined in (\ref{EoMs}). The equations of motion are

\be \label{EoMs by l}
    E_I[\Phi_J] \equiv E^{(0)}_{I}[\Phi_J]+\sum_{n=1}^{\infty}l^n E^{(n)}_{I}[\Phi_J]=0
\ee

where $l^n E^{(n)}_{I}[\Phi_J]$ comes from varying $l^n \LL_{n+2}$. This equation must hold to each order in $l$ individually. The order $l^0$ part is simply

\be
    E_{I}^{(0)}[\Phi_J^{(0)}] = 0
\ee

We will show that in Killing vector GNCs on the horizon, $E_{\tau A}^{(0)}$ and $E_{\tau \tau}^{(0)}$ evaluate to

\be
\begin{split} \label{EtauA and Etautau}
       E_{\tau A}^{(0)}[\Phi_{J}]\Big|_{\rho=0} =& -\frac{1}{2}\partial_{A}X - \frac{1}{2} c_1(\phi)\left( F_{A B} h^{B C} - F_{\tau \rho} \delta_{A}^{C} \right) F_{\tau C}\\
       E_{\tau \tau}^{(0)}[\Phi_{J}]\Big|_{\rho=0} =& - \frac{1}{2} c_1(\phi) F_{\tau A} F_{\tau B} h^{A B}
\end{split}
\ee
From the first component, we see that $E_{\tau A}^{(0)}[\Phi_{J}^{(0)}]\Big|_{\rho=0} = 0$ implies $\partial_{A} X^{(0)}\Big|_{\rho=0} = 0$ if $F^{(0)}_{\tau A}\Big|_{\rho=0} = 0$. But from the second component, $E_{\tau \tau}^{(0)}[\Phi_{J}^{(0)}]\Big|_{\rho=0} = 0$ implies $F^{(0)}_{\tau A}\Big|_{\rho=0} = 0$ because $h^{(0) A B}$ is positive definite and we assumed $c_1>0$. Thus the Generalized 0th Law holds to zeroth order in $l$.

The proof will then proceed by induction. We will assume that $\partial_{A} X^{(n)}\Big|_{\rho=0} = 0$ and $F^{(n)}_{\tau A}\Big|_{\rho=0} = 0$ for $n< k$, and then prove that $\partial_{A} X^{(k)}\Big|_{\rho=0} = 0$ and $F^{(k)}_{\tau A}\Big|_{\rho=0} = 0$.

To do this, we will consider the order $l^k$ part of two components of (\ref{EoMs by l}). In a similar fashion to Bhattacharyya et al, we will show that $E_{\tau A}[\Phi_J]$ greatly simplifies on the horizon, regardless of the higher order terms in the EFT. In particular, its order $l^k$ part is of the form
\be \label{EtauA}
    \text{At order} \,\, l^k, \quad E_{\tau A}[\Phi_J]\Big|_{\rho=0} = -\frac{1}{2} l^k \partial_{A} X^{(k)} + l^k M_{A}^{\,\,\,\,C} F^{(k)}_{\tau C} = 0
\ee
where $M_{A}^{\,\,\,\,C}$ is a function only of the lowest order fields $\Phi_I^{(0)}$. From this we can see that the two statements in the Generalized 0th Law are not independent: if we can show that $F^{(k)}_{\tau A}\Big|_{\rho=0} = 0$, then we immediately have $\partial_{A} X^{(k)}\Big|_{\rho=0} = 0$. 

To do this final step we must look at another component of the equation of motion. We will show that the order $l^k$ part of $E_{\tau}[\Phi_J]\Big|_{\rho=0}$ can be brought into the form
\be \label{Etau}
    \text{At order} \,\, l^k, \quad E_{\tau}[\Phi_J]\Big|_{\rho=0} = l^k \DD^{(0)}_A \left[ N^{A B} F_{\tau B}^{(k)} \right] = 0
\ee
where $N^{A B}$ is a function only of the lowest order fields $\Phi_I^{(0)}$, and $\DD^{(0)}_A$ is the covariant derivative with respect to $h_{A B}^{(0)}$. We will show that this equation has no non-trivial solutions for $F^{(k)}_{\tau A}\Big|_{\rho=0}$ if every closed 1-form on $C(\tau)$ is exact. This condition follows from our assumptions on the topology of $C(\tau)$, and thus the Generalized 0th Law is proved in this case.

In order to simplify $E_{\tau A}[\Phi_J]\Big|_{\rho=0}$ and $E_{\tau}[\Phi_J]\Big|_{\rho=0}$ to the forms in (\ref{EtauA}) and (\ref{Etau}), we will need to prove a crucial fact: the Generalized 0th Law implies\footnote{In case it is confusing why we will assume the Generalized 0th Law whilst in the middle of proving it, we will be using this result up to and including order $l^{k-1}$ to complete an inductive loop at order $l^k$.} that all positive boost weight, gauge-independent quantities vanish on the horizon $\NN$. In \cite{Bhat:2022}, it is shown that a relation between Killing vector GNCs and affinely parameterized GNCs can be used to show the 0th Law implies that positive boost weight quantities built only out of metric components vanish on $\NN$. In Section \ref{PBWQH} we will show this relation can also be applied to positive boost weight quantities built out of the Maxwell field $F_{\mu \nu}$ and the scalar field $\phi$.

\section{Proof of the Generalized 0th Law}
\subsection{The Base Case: The Generalized 0th Law for 2-Derivative Einstein-Maxwell-Scalar Theory}\label{base case}

The first step in our proof will be to show that the Generalized 0th Law holds at lowest order in $l$, i.e. that
\begin{equation}
    \partial_{A} X^{(0)}\Big|_{\rho=0} = 0, \quad \text{and} \quad F_{\tau A}^{(0)}\Big|_{\rho=0} = 0
\end{equation}
This is equivalent to proving the Generalized 0th Law for the 2-derivative Einstein-Maxwell-Scalar Lagrangian $\LL_{2}$ given in (\ref{0and2-deriv}). Bardeen, Carter and Hawking's original proof of the 0th Law would achieve this, because the 2-derivative theory satisfies the Null Dominant Energy Condition as defined above. Here we give a reformulation of the proof that motivates many of the steps used in the later proof for full Einstein-Maxwell-Scalar EFT.

We proceed by studying $E^{(0)}_I[\Phi_J]$, which is the part of the equation of motion arising from $\LL_{2}$. Rewriting here for convenience, the $(\alpha \beta)$ component is
\be
    E^{(0)}_{\alpha\beta}[\Phi_J]= R_{\alpha \beta} - \frac{1}{2} \nabla_{\alpha} \phi \nabla_{\beta} \phi -\frac{1}{2} c_1(\phi) F_{\alpha \delta} F_{\beta}^{\,\,\,\,\delta} - \frac{1}{2} g_{\alpha \beta} \left( R - V(\phi) -\frac{1}{2} \nabla_{\gamma}{\phi} \nabla^{\gamma}{\phi} - \frac{1}{4} c_1(\phi) F_{\gamma \delta} F^{\gamma\delta} \right)
\ee
We will evaluate two components in Killing vector GNCs on the horizon. Firstly we will evaluate $E^{(0)}_{\tau A}[\Phi_J]\Big|_{\rho=0}$. The Ricci component $R_{\tau A}\Big|_{\rho=0}$ is evaluated in \cite{Bhat:2022}:
\be
    R_{\tau A}\Big|_{\rho=0} = -\frac{1}{2} \partial_{A} X\Big|_{\rho=0}
\ee
The second term, $\nabla_{\tau}\phi \nabla_{A}\phi$ vanishes because we assumed the scalar field is invariant under the Killing vector $\xi$, which implied $ \partial_{\tau} \phi =0$. The third term $c_1(\phi) F_{\tau \gamma} F_{A \delta} g^{\gamma \delta}$ simplifies on the horizon where $g^{\gamma \delta}$ is particularly straightforward:
\be
\begin{split} 
    c_1(\phi) F_{\tau \gamma} F_{A \delta} g^{\gamma \delta}\Big|_{\rho=0} =& c_1(\phi) \left( F_{\tau \rho} F_{A \tau} + F_{\tau C} F_{A B} h^{B C} \right)\\
    =& c_1(\phi) \left( F_{A B} h^{B C} - F_{\tau \rho} \delta^{C}_{A} \right) F_{\tau C}
\end{split}
\ee
The final bracketed term also vanishes on the horizon because the prefactor $g_{\tau A}\Big|_{\rho=0} = 0$. This leaves us with the first equation from (\ref{EtauA and Etautau}), and as discussed, we can substitute this into the order $l^0$ part of the equation of motion $E_{I}^{(0)}[\Phi_J^{(0)}] = 0$ to get that $\partial_{A}X^{(0)}\big|_{\rho=0} = 0$ if $F^{(0)}_{\tau C}\big|_{\rho=0} = 0$.

In pursuit of proving $F^{(0)}_{\tau C}\big|_{\rho=0} = 0$, let us now evaluate another component on the horizon, $E_{\tau \tau}^{(0)}[\Phi_{J}]\Big|_{\rho=0}$. The first term is $R_{\tau \tau}\Big|_{\rho=0} = R_{\mu \nu} \xi^\mu \xi^\nu \Big|_{\NN}$ which vanishes by Raychaudhuri's equation. The second term, $\nabla_{\tau}\phi \nabla_{\tau}\phi$ once again vanishes by $\partial_{\tau} \phi =0$. The third term is
\be
\begin{split}
    c_1(\phi) F_{\tau \gamma} F_{\tau \delta} g^{\gamma \delta}\Big|_{\rho=0} =& c_1(\phi) F_{\tau A} F_{\tau B} h^{A B}
\end{split}
\ee
The bracketed term vanishes again since $g_{\tau \tau}\big|_{\rho=0} = 0$. Therefore we retrieve the second equation from (\ref{EtauA and Etautau}), which we can substitute into the $l^0$ part of the equation of motion to get
\be \label{F0squared}
    c_1(\phi^{(0)}) F^{(0)}_{\tau A} F^{(0)}_{\tau B} h^{(0) A B} = 0
\ee
Now, $h_{A B}$ is the induced metric on the spacelike cut $C(\tau)$, therefore it is positive definite. This implies $h^{(0)}_{A B}$ is also positive definite since $h^{(0)}_{A B} = h_{A B}\big|_{l=0}$. Furthermore, we assumed $c_1>0$. Therefore (\ref{F0squared}) implies $F^{(0)}_{\tau A}\big|_{\rho=0}=0$, and so we have proved the Generalized 0th Law to order $l^0$.

\subsection{The Inductive Step}

We prove the Generalized 0th Law to all orders in $l$ by induction. Let us assume $\partial_{A} X^{(n)}\Big|_{\rho=0} = 0$ and $F^{(n)}_{\tau A}\Big|_{\rho=0} = 0$ for $n< k$. We now aim to prove that $\partial_{A} X^{(k)}\Big|_{\rho=0} = 0$ and $F^{(k)}_{\tau A}\Big|_{\rho=0} = 0$.

To do this, let us consider the order $l^k$ part of the full equation of motion (\ref{EoMs by l}). This is not simply $l^k E^{(k)}_{I}[\Phi_J]$ because $\Phi_J$ is itself a series in $l$. However, it will certainly not depend on $E^{(n)}_{I}$ or $\Phi_J^{(n)}$ for $n>k$ because they come with too high powers of $l$. Let us introduce the notation
\be
    f^{[n]} = \sum_{m=0}^{n}l^m f^{(n)}
\ee
Then the order $l^k$ part of (\ref{EoMs by l}) will be a subset of the terms in $E_{I}^{[k]}[\Phi_J^{[k]}]$. Furthermore, since $\Phi_J^{(k)}$ already comes with a factor $l^k$, the only place $\Phi_J^{(k)}$ can appear is in $E^{(0)}_I[\Phi_J^{[k]}]$. In particular, it will appear as $E^{(0)}_I[\Phi_J^{(0)}+l^k \Phi_J^{(k)}]$ linearized around the background $\Phi_J^{(0)}$. Therefore we can write the order $l^k$ part of (\ref{EoMs by l}) as
\be \label{orderlk}
    \text{At order} \,\, l^k, \quad E_{I}[\Phi_J] = E_{I}^{[k]}[\Phi_J^{[k-1]}] + l^k \left( \Phi_J^{(k)} \frac{\delta E^{(0)}_I}{\delta \Phi_J}[\Phi_J^{(0)}] +  \partial_\mu \Phi_J^{(k)} \frac{\delta E^{(0)}_I}{\delta (\partial_\mu \Phi_J)}[\Phi_J^{(0)}] + ... \right) 
\ee
where it is given that we only take order $l^k$ terms in the first term, and the bracketed term is $E^{(0)}_I[\Phi_J^{(0)}+l^k \Phi_J^{(k)}]$ linearized around $\Phi^{(0)}_J$. Setting this to zero allows us to solve for the fields $\Phi_J$ order-by-order in $l$: once we have solved for $\Phi_J^{(0)}$ we can study (\ref{orderlk}) at order $l$ to solve for $\Phi_J^{(1)}$, then at order $l^2$ to solve for $\Phi_J^{(2)}$ and so on.

It is difficult to study (\ref{orderlk}) however, because we do not know the form of $E^{[k]}_I$. It comes from the variation of the higher derivative parts of our EFT Lagrangian, which in theory could take a variety of forms. The only part we do know the form of is $E^{(0)}_I$. Therefore we would like to find a scenario where the dependence on the unknown $E^{[k]}_I$ vanishes.

It turns out that on the horizon, certain components of $E_{I}^{[k]}[\Phi_J^{[k-1]}]$ do necessarily vanish by our inductive hypothesis. In particular, $E_{\tau A}^{[k]}[\Phi_J^{[k-1]}]\big|_{\NN}=0$ and $E_{\tau}^{[k]}[\Phi_J^{[k-1]}]\big|_{\NN}=0$. The proof of these are left to the next two sections. In short, it follows from the fact that they are proportional to positive boost weight components in affinely parameterized GNCs. It will be shown in Section \ref{PBWQH} that positive boost weight quantities vanish on the horizon if the Generalized 0th Law holds. But by our inductive hypothesis, the Generalized 0th Law holds for the fields $\Phi_J^{[k-1]}$, which are all that $E_{\tau A}^{[k]}[\Phi_J^{[k-1]}]$ and $E_{\tau}^{[k]}[\Phi_J^{[k-1]}]$ depend on. 

Let us assume for now that these components do indeed vanish on the horizon. Then at order $l^k$, $E_{\tau A}[\Phi_J]\big|_{\rho=0}$ and $E_{\tau}[\Phi_J]\big|_{\rho=0}$ are simply given by $E^{(0)}_{\tau A}[\Phi_J^{(0)}+l^k \Phi_J^{(k)}]\big|_{\rho=0}$ and $E^{(0)}_{\tau}[\Phi_J^{(0)}+l^k \Phi_J^{(k)}]\big|_{\rho=0}$ linearized around $\Phi_I^{(0)}$.

During the proof of the base case, we calculated
\be
    E_{\tau A}^{(0)}[\Phi_{J}]\Big|_{\rho=0} = -\frac{1}{2}\partial_{A}X - \frac{1}{2} c_1(\phi)\left( F_{A B} h^{B C} - F_{\tau \rho} \delta_{A}^{C} \right) F_{\tau C}
\ee
We can now replace the fields $X=X^{(0)}+l^k X^{(k)}$, $F_{\mu \nu}=F^{(0)}_{\mu \nu}+l^k F^{(k)}_{\mu \nu}$ etc., linearize around the order $l^0$ fields $X^{(0)}$, $F_{\mu \nu}^{(0)}$ etc., and use that $F^{(0)}_{\tau C}\big|_{\rho=0}=0$ to get that
\be
    \text{At order} \,\, l^k, \quad E_{\tau A}[\Phi_J]\Big|_{\rho=0} = -\frac{1}{2} l^k \partial_{A} X^{(k)} - \frac{1}{2} c_1(\phi^{(0)})\left( F^{(0)}_{A B} h^{(0) B C} - F^{(0)}_{\tau \rho} \delta_{A}^{C} \right) l^k F^{(k)}_{\tau C} = 0
\ee

Therefore,
\be
    \partial_{A} X^{(k)}\Big|_{\rho=0} = M_{A}^{\,\,\,\,C} F^{(k)}_{\tau C}\Big|_{\rho=0}
\ee
where $M_{A}^{\,\,\,\,C} = - c_1(\phi^{(0)})\left( F^{(0)}_{A B} h^{(0) B C} - F^{(0)}_{\tau \rho} \delta_{A}^{C} \right)$ is a function only of the lowest order fields $\Phi_I^{(0)}$. Once again we see that $\partial_{A} X^{(k)}\big|_{\rho=0} = 0$ if $F^{(k)}_{\tau C}\big|_{\rho=0} = 0$. 

We now turn to $E^{(0)}_{\tau}[\Phi_J]$ in order to prove $F^{(k)}_{\tau C}\big|_{\rho=0} = 0$. Rewriting from above, it is given by
\be
    E^{(0)}_\tau[\Phi_J] = \nabla^{\beta}{\Big[c_1(\phi)F_{\tau\beta} - 4 c_2(\phi) F^{\gamma \delta} \epsilon_{\tau \beta \gamma \delta}\Big]}
\ee
We can again evaluate this on the horizon in Killing vector GNCs. The calculation involves evaluating Christoffel symbols, and is given in Appendix \ref{EtauEvaluation} as an example of using Killing vector GNCs. The result is
\be
    E^{(0)}_\tau[\Phi_J]\Big|_{\rho=0} = h^{A B} \DD_{A}\Big[c_1(\phi)F_{\tau B} - 8 c_2(\phi) \epsilon_{B}^{\,\,\,\,C}F_{\tau C}\Big]
\ee
where $\DD_A$ is the covariant derivative with respect to $h_{A B}$. We again linearize all the fields around the background $\Phi_{J}^{(0)}$ and use $F^{(0)}_{\tau C}\big|_{\rho=0}=0$ to get
\be \label{Etauhorizon}
    \text{At order} \,\, l^k, \quad E_{\tau}[\Phi_J]\Big|_{\rho=0} = l^k h^{(0) A B} \DD^{(0)}_A \left[ c_1(\phi^{(0)}) F^{(k)}_{\tau B} - 8 c_2(\phi^{(0)}) \epsilon_{\,\,B}^{(0)\,C} F^{(k)}_{\tau C} \right] = 0
\ee
where $\DD^{(0)}$ is the covariant derivative with respect to $h^{(0)}_{A B}$. For each $\tau$, this is a differential equation for $F^{(k)}_{\tau C}\big|_{\rho=0}$ on the spacelike 2-slice $C(\tau)$. We will show that it has no non-trivial solutions.

Note that since $F_{\alpha \beta}$ is a Maxwell field, it satisfies $\partial_{[\alpha} F_{\beta \gamma]} = 0$. We can combine this with $\partial_{\tau} F_{\alpha \beta} = 0$ to get
\be
\partial_{[\tau} F_{A B]} = 0 \implies \partial_A F_{\tau B} - \partial_B F_{\tau A} = 0
\ee
This must hold to all orders in $l$, so
\be \label{Closed}
\partial_A F^{(k)}_{\tau B} - \partial_B F^{(k)}_{\tau A} = 0
\ee
Now, we can view $F^{(k)}_{\tau A}\big|_{\rho=0}$ as a 1-form on the submanifold $C(\tau)$. Call this 1-form $V_A=F^{(k)}_{\tau A}\big|_{\rho=0}$. Then equation (\ref{Closed}) becomes
\be
    \text{d}V=0
\ee
where $\text{d}$ is the exterior derivative on $C(\tau)$. We assumed that every closed 1-form on $C(\tau)$ is exact, and hence there exists some function $f$ on the whole of $C(\tau)$ such that 
\be
V = \text{d}f
\ee
$f$ is essentially just the order $l^k$ part of the electric potential. The crucial thing is that we know we can define $f$ on the whole of $C(\tau)$, whereas we could not necessarily define the electric potential globally.

We can substitute this into (\ref{Etauhorizon}) to get
\be
h^{(0) A B} \DD^{(0)}_A \left[ c_1(\phi^{(0)}) \DD^{(0)}_B f - 8 c_2(\phi^{(0)}) \epsilon_{\,\,B}^{(0)\,C} \DD^{(0)}_C f \right]=0
\ee
We now integrate this against $\sqrt{h^{(0)}} f$ over $C(\tau)$:
\be
\int_{C(\tau)} \text{d}^{d-2}x \sqrt{h^{(0)}} f h^{(0) A B} \DD^{(0)}_A \left[ c_1(\phi^{(0)}) \DD^{(0)}_B f - 8 c_2(\phi^{(0)}) \epsilon_{\,\,B}^{(0)\,C} \DD^{(0)}_C f \right]=0
\ee
Apply the divergence theorem to get
\be
-\int_{C(\tau)} \text{d}^{d-2}x \sqrt{h^{(0)}}\left[ c_1(\phi^{(0)}) h^{(0) A B} \DD^{(0)}_A f \DD^{(0)}_B f - 8 c_2(\phi^{(0)}) \epsilon^{(0) A B} \DD^{(0)}_A f \DD^{(0)}_B f \right]=0
\ee
where the boundary term vanished because we assumed $C(\tau)$ was compact. The second term in the brackets is 0 by the antisymmetry of $\epsilon^{(0) A B}$, which just leaves the first term. Since $h^{(0)}_{A B}$ is postive definite and $c_1>0$, for the integral to be 0 we must have $\DD^{(0)}_A f = 0$, or equivalently
\be
   F^{(k)}_{\tau A}\Big|_{\rho=0} = 0 
\ee
Therefore the inductive step is proven. 

Now all that remains is to fill the gap in our proof and show that $E_{\tau A}^{[k]}[\Phi_J^{[k-1]}]\big|_{\NN}=0$ and $E_{\tau}^{[k]}[\Phi_J^{[k-1]}]\big|_{\NN}=0$. To do this we will need to prove a statement about positive boost weight quantities on the horizon.

\subsection{Positive Boost Weight Quantities on the Horizon} \label{PBWQH}

Following the method of \cite{Bhat:2022}, we will prove that the Generalized 0th Law implies that all positive boost weight, electromagnetic gauge-independent quantities vanish on the horizon. Somewhat counter-intuitively, this will allow us to complete the inductive step and prove the Generalized 0th Law itself.

The boost weight of a quantity (defined in Section \ref{BW}) is determined in the Affinely Parametrized GNCs of Section \ref{APGNCs}. The most basic electromagnetic gauge-independent quantities we can make from the metric and matter fields in these coordinates are of the form $\partial_{A_1}{... \partial_{A_n}{ \partial_{r}^{p} \partial_{v}^q \varphi } }$ with $\varphi \in \{\alpha, \beta_{A}, \mu_{A B}, \mu^{A B}, F_{v r}, F_{A B}, F_{v A}, F_{r A}, \phi \}$. Call such terms \textit{building blocks}. On the horizon $r=0$, all quantities in our theory can be expanded out as expressions in building blocks\footnote{There is no explicit appearance of the coordinates $(v,x^A)$ because they do not appear explicitly in the metric.}, e.g. $\nabla_{v} F_{r A}\Big|_{\rho=0} = \partial_{v}{F_{r A}} - \frac{1}{2}F_{A B} \beta_{C} \mu^{B C}-\frac{1}{2}F_{r B} \partial_v \mu_{A C} \mu^{B C} - \frac{1}{2}F_{r v} \beta_{A}$.

The boost weights of these building blocks are

\begin{itemize}
    \item $F_{v A}$ has boost weight $+1$.
    \item $\alpha$, $\beta_{A}$, $\mu_{A B}$, $\mu^{A B}$, $F_{v r}$, $F_{A B}$, and $\phi$ have boost weight 0.
    \item $F_{r A}$ has boost weight $-1$.
    \item $\partial_v$ derivatives each add $+1$ to the boost weight, $\partial_r$ derivatives each add $-1$, and $\partial_A$ derivatives add $0$.
\end{itemize}

Therefore positive boost weight building blocks are of the form 
\begin{equation} \label{PBB}
    \partial_{A_1}{... \partial_{A_n}{ \partial_{r}^{p} \partial_{v}^{q} \varphi } }\,\, \text{with}\,\, \varphi \in \{F_{v A}, \partial_{v}\alpha, \partial_{v}\beta_{A}, \partial_{v}\mu_{A B}, \partial_{v}\mu^{A B}, \partial_{v}F_{v r}, \partial_{v}F_{A B}, \partial_{v}\phi, \partial_{v}^2 F_{r A} \} \,\, \text{and} \,\, q\geq p
\end{equation}

The terms in the expansion of a positive boost weight quantity on the horizon must all have at least one factor of the positive boost weight building blocks listed above. Therefore, if we can show that all positive boost weight building blocks vanish on the horizon, then we have shown that all positive boost weight quantities vanish on the horizon.

To do this we shall employ a relation between affinely parameterized GNCs and Killing vector GNCs.

Let us assume the Generalized 0th Law holds. Then $F_{\tau A}\Big|_{\rho=0}= 0$ and $X\Big|_{\rho=0}=2 \kappa$ with $\kappa$ constant. By smoothness, 
\begin{equation}
    X(\rho, x^C) = 2 \kappa + \rho f(\rho, x^C), \quad \quad F_{\tau A}(\rho, x^C) = \rho f_{A}(\rho, x^C)
\end{equation}
where $f(\rho, x^C)$ and $f_{A}(\rho, x^C)$ are regular on the horizon. Then, in Killing vector GNCs, $(g_{\alpha \beta}, F_{\alpha \beta}, \phi)$ are given by (with $x^C$ dependence suppressed)

\begin{equation}
    \begin{split}
        g &= 2 \text{d}\tau \text{d}\rho - \left[2 \kappa \rho + \rho^2 f(\rho)\right] \text{d}\tau^2 +2 \rho \omega_{A}(\rho) \text{d}x^A \text{d}\tau + h_{A B}(\rho) \text{d}x^A \text{d}x^B\\
        F &= F_{\tau \rho}(\rho) \text{d}\tau \wedge \text{d}\rho + F_{\rho A}(\rho) \text{d}\rho \wedge \text{d}x^A + \rho f_{A}(\rho) \text{d}\tau \wedge \text{d}x^A + F_{A B}(\rho) \text{d}x^A \wedge \text{d}x^B\\
        \phi &= \phi(\rho)
    \end{split}
\end{equation}

We now make the coordinate transformation\footnote{Note this is slightly different from the choice in \cite{Bhat:2022} in that we have $(\kappa v+1)$ where they have $\kappa v$. We have added the $1$ so that $v=0$ corresponds to $\tau=0$, and also to put it in such a form that if the black hole is extremal, i.e. $\kappa=0$, then the transformation is the identity $\rho=r, \tau=v$.}

\begin{equation}\label{transform}
    \rho = r( \kappa v+1), \quad \quad \tau = \frac{1}{\kappa} \log\left(\kappa v+1\right).
\end{equation}

with the $x^C$ co-ordinates unchanged. In these new co-ordinates, the horizon is $r=0$ and $C$ is $v=r=0$. The transformation also has the effect of putting the metric in affinely parameterized form:  

\begin{equation} \label{affine param form}
    \begin{split}
        g = & 2 \text{d}v \text{d}r - r^2 f\big(r( \kappa v+1)\big) \text{d}v^2 + 2 r \omega_{A}\big(r( \kappa v+1)\big) \text{d}v \text{d}x^A + h_{A B}\big(r( \kappa v+1)\big) \text{d}x^A \text{d}x^B\\
        F = &F_{\tau \rho}\big(r( \kappa v+1)\big) \text{d}v \wedge \text{d}r + (\kappa v+1)F_{\rho A}\big(r( \kappa v+1)\big) \text{d}r \wedge \text{d}x^A +\\
        & r\Big[ \kappa F_{\rho A}\big(r( \kappa v+1)\big) +f_{A}\big(r( \kappa v+1)\big)\Big] \text{d}v \wedge \text{d}x^A + F_{A B}\big(r( \kappa v+1)\big) \text{d}x^A \wedge \text{d}x^B\\
        \phi = &\phi\big(r(\kappa v+1)\big)
    \end{split}
\end{equation}

Thus the $(r,v,x^C)$ are a choice of affinely parameterized GNCs with (again suppressing $x^C$ dependence)

\begin{equation} \label{APGNCsSymmetry}
    \begin{split}
        \alpha(r,v)= f\big(r(\kappa v+1)\big), \quad & \quad \beta_A(r, v)= -\omega_{A}\big(r(\kappa v+1)\big),\\
        \mu_{A B}(r,v)=h_{A B}\big(r(\kappa v+1)\big), \quad & \quad \mu^{A B}(r,v)=h^{A B}\big(r(\kappa v+1)\big),\\
        F_{v r}(r,v) = F_{\tau \rho}\big(r(\kappa v+1)\big), \quad & \quad F_{A B}(r,v) = F_{A B}\big(r(\kappa v+1)\big),\\
        F_{r A}\big(r, v\big) = (\kappa v+1)&F_{\rho A}\big(r(\kappa v+1)\big),\\
        F_{v A}\big(r, v\big) = r\Big[ \kappa F_{\rho A}&\big(r(\kappa v+1)\big) +f_{A}\big(r(\kappa v+1)\big)\Big],\\
        \phi\big(r, v\big) = &\phi\big(r(\kappa v+1)\big)
    \end{split}
\end{equation}

The importance of this is that the $v$-dependence of these quantities is severely restricted by the $\tau$-independence of the original Killing vector GNC quantities. The zero boost weight quantities $\alpha$, $\beta_A$, $\mu_{A B}$, $\mu^{A B}$, $F_{v r}$, $F_{A B}$ and $\phi$ depend on $v$ strictly through the combination $rv$. Therefore, taking a $\partial_v$ of these quantities will produce an overall factor of $r$, which vanishes on the horizon. The positive boost weight quantity $F_{v A}$ already has a pre-factor of $r$, and also depends on $v$ strictly through $rv$. Finally,

\begin{equation}
    \partial_{v}^2 F_{r A} = r \left[2\kappa \partial_{\rho} F_{\rho A}\big(r(\kappa v+1)\big) + r \kappa (\kappa v+1)\partial^2_{\rho}F_{\rho A}\big(r(\kappa v+1)\big)\right]
\end{equation}

Thus we can write all of the quantities $\varphi \in \{F_{v A}, \partial_{v}\alpha, \partial_{v}\beta_{A}, \partial_{v}\mu_{A B}, \partial_{v}\mu^{A B}, \partial_{v}F_{v r}, \partial_{v}F_{A B}, \partial_{v}\phi, \partial_{v}^2 F_{r A} \}$ from (\ref{PBB}) in the form

\begin{equation}
    \varphi = r f_{\varphi}\big(r(\kappa v+1)\big)
\end{equation}

Taking a $(\partial_{r}\partial_{v})$ derivative preserves this form, as does taking $\partial_A$ derivatives. Thus, every positive boost weight building block satisfies 

\begin{equation}
\begin{split}
        \partial_{A_1}{... \partial_{A_n}{ \partial_{r}^{p} \partial_{v}^{q} \varphi } } &= \partial^{q-p}_v \left[ \partial_{A_1}{... \partial_{A_n}{ (\partial_{r}\partial_{v})^{p} \varphi } }\right]\\
        &= \partial^{q-p}_v \left[ r f_{\partial_{A_1}{... \partial_{A_n}{ (\partial_{r}\partial_{v})^{p} \varphi } }}\big(r(\kappa v+1)\big) \right]\\
        &\propto r^{1+q-p}
\end{split}
\end{equation}

with $q\geq p$. Therefore all positive boost weight building blocks vanish on the horizon $r=0$. 

This proves that all positive boost weight quantities vanish on the horizon in the choice of affinely parameterized GNCs given by the transformation (\ref{transform}). However, as discussed in Section \ref{APGNCs}, there are infinitely many choices of such GNCs, all related by $v'=v/a(x^A)$ on $\NN$ for some arbitrary function $a(x^A)>0$. To prove that positive boost weight quantities vanish on the horizon in all choices of affinely parameterized GNCs, we use Lemma 2.1 of \cite{Hollands:2022}, which states that on the horizon a quantity of certain boost weight in $(r, v, x^A)$ GNCs can be written as the sum of terms of the same boost weight in $(r', v', x^A)$ GNCs. This means that if all positive boost weight quantities vanish on the horizon in one choice of affinely parameterized GNCs, then they vanish in all choices of affinely parameterized GNCs.

\subsection{Completion of The Inductive Step}\label{CompleteInduct}

We will now use the statement about positive boost weight quantities on the horizon to complete our inductive step by proving $E_{\tau A}^{[k]}[\Phi_J^{[k-1]}]\big|_{\NN}=0$ and $E_{\tau}^{[k]}[\Phi_J^{[k-1]}]\big|_{\NN}=0$.

Our inductive hypothesis was that $\partial_{A} X^{(n)}\Big|_{\rho=0} = 0$ and $F^{(n)}_{\tau A}\Big|_{\rho=0} = 0$ for $n< k$. Therefore the fields $\Phi^{[k-1]}_I\equiv \Phi_I^{(0)} + l \Phi_I^{(1)} + ... + l^{k-1} \Phi_I^{(k-1)}$ satisfy the Generalized 0th Law. In particular, $X^{[k-1]}\big|_{\rho=0} = 2\kappa^{[k-1]}$ is constant, and we can make the coordinate transformation
\begin{equation}\label{transforminduct}
    \rho = r( \kappa^{[k-1]} v+1), \quad \quad \tau = \frac{1}{\kappa^{[k-1]}} \log\left(\kappa^{[k-1]} v+1\right).
\end{equation}
to bring the fields $\Phi^{[k-1]}_I = (g^{[k-1]}_{\mu \nu}, F^{[k-1]}_{\mu \nu}, \phi^{[k-1]})$ into the affinely parameterized form of (\ref{affine param form}). Then, by the above proof, any positive boost weight quantity made out of $\Phi^{[k-1]}_I$ will vanish on the horizon in these coordinates. In particular, 
\be
E_{v A}^{[k]}[\Phi_J^{[k-1]}]\Big|_{r=0}=0 \,\text{ and }\, E_{v}^{[k]}[\Phi_J^{[k-1]}]\Big|_{r=0}=0
\ee
because they have boost weight $+1$. But we also know how $E_{\mu \nu}^{[k]}[\Phi_J^{[k-1]}]$ and $E_{\mu}^{[k]}[\Phi_J^{[k-1]}]$ transform under the change of co-ordinates (\ref{transforminduct}) because they are tensors. The inverse co-ordinate transformation is
\begin{equation}\label{inversetransforminduct}
    r = \rho e^{-\kappa^{[k-1]} \tau} \quad \quad v = \frac{1}{\kappa^{[k-1]}} \left( e^{\kappa^{[k-1]}\tau} -1\right).
\end{equation}
So,
\be
\begin{split}
    E_{\tau A}^{[k]}[\Phi_J^{[k-1]}]\big|_{\rho=0}=&\frac{\partial \Tilde{x}^{\mu}}{\partial \tau}\frac{\partial \Tilde{x}^{\nu}}{\partial x^A} E_{\mu \nu}^{[k]}[\Phi_J^{[k-1]}]\big|_{r=0}\\
    =&e^{\kappa^{[k-1]}\tau} E_{v A}^{[k]}[\Phi_J^{[k-1]}]\big|_{r=0}\\
    =&0
\end{split}
\ee
Similarly,
\be
\begin{split}
    E_{\tau}^{[k]}[\Phi_J^{[k-1]}]\big|_{\rho=0}=&\frac{\partial \Tilde{x}^{\mu}}{\partial \tau} E_{\mu}^{[k]}[\Phi_J^{[k-1]}]\big|_{r=0}\\
    =&e^{\kappa^{[k-1]}\tau} E_{v}^{[k]}[\Phi_J^{[k-1]}]\big|_{r=0}\\
    =&0
\end{split}
\ee

This completes the proof of the Generalized 0th Law. 

\subsection{The Generalized 0th Law for a Charged Scalar Field}\label{ChargedScalar}

This proof of the Generalized 0th Law can be modified to apply to the EFT of gravity, electromagnetism and a \textit{charged} scalar field. In this scenario, we assume we have a global gauge potential $A_\mu$ with $F=\text{d}A$. The dynamical fields are $\Phi_I = (g_{\mu \nu}, A_\mu, \phi)$, the scalar field $\phi$ is complex with some charge $\lambda$, and $A_\mu$ and $\phi$ transform under an electromagnetic gauge transformation as
\be \label{gaugetransform}
A_\mu \rightarrow \Tilde{A}_\mu = A_\mu + \partial_\mu \chi, \quad \quad \phi \rightarrow \Tilde{\phi} = e^{i \lambda \chi} \phi
\ee
with $\chi$ an arbitrary real-valued function. We generalize our leading order Lagrangian to 
\begin{equation} \label{0and2-derivcharged}
    \LL_{2} = R - V(|\phi|^2) - g^{\alpha \beta} \left(\mathfrak{D}_{\alpha}{\phi}\right)^{*} \mathfrak{D}_{\beta}{\phi} - \frac{1}{4} c_1(|\phi|^2) F_{\alpha\beta} F^{\alpha\beta} + c_2(|\phi|^2) F_{\alpha \beta} F_{\gamma \delta} \epsilon^{\alpha \beta \gamma \delta}
\end{equation}
where $\mathfrak{D}_\alpha$ is the gauge covariant derivative $\fD_\alpha = \nabla_\alpha-i \lambda A_\mu$. $\LL_{2}$ is invariant under the gauge transform (\ref{gaugetransform}). Since the charge $\lambda$ adds a new scale into the theory, the EFT series is now a joint series in derivatives and powers of $\lambda$. We assume $\lambda \leq 1/L$ where $L$ is a typical length scale of the solution, so that $\lambda$ is comparable to a 1-derivative term. This is reasonable if we want the classical approximation to be valid. The EFT Lagrangian is
\be
\LL = \LL_{2} + \sum_{n=1}^{\infty} l^n \LL_{n+2}
\ee
where the $\LL_{n+2}$ contains all gauge-independent terms with $n+2$ derivatives or powers of $\lambda$. 

The equations of motion are 
\be
E_I[\Phi_J] \equiv E^{(0)}_I[\Phi_J] + \sum_{n=1}^{\infty}l^n E^{(n)}_I[\Phi_J]
\ee
where the parts arising from the leading order theory, $E^{(0)}_I[\Phi_J]$, are 
\begin{equation} \label{0th order EoM 1 charged}
    E^{(0)}_{\alpha\beta}= R_{\alpha \beta} - \left(\fD_{(\alpha}\phi\right)^{*} \fD_{\beta)} \phi -\frac{1}{2} c_1(|\phi|^2) F_{\alpha \delta} F_{\beta}^{\,\,\,\,\delta} - \frac{1}{2} g_{\alpha \beta} \left( R - V(|\phi|^2) - g^{\gamma \delta}\left(\fD_{\gamma}{\phi}\right)^{*} \fD_{\delta}{\phi} - \frac{1}{4} c_1(|\phi|^2) F_{\gamma \delta} F^{\gamma\delta} \right)
\end{equation}
\begin{equation}
    E^{(0)}_\alpha = i\lambda\left[\phi^* \fD_\alpha \phi - \phi \left(\fD_{\alpha}\phi\right)^{*} \right] + \nabla^{\beta}{\Big[c_1(|\phi|^2)F_{\alpha\beta} - 4 c_2(|\phi|^2) F^{\gamma \delta} \epsilon_{\alpha \beta \gamma \delta}\Big]}
\end{equation}
\begin{equation} \label{0th order EoM 2 charged}
    E^{(0)} = g^{\alpha\beta}\fD_{\alpha}{\fD_{\beta}{\phi}}-\phi V'(|\phi|^2)-\frac{1}{4} \phi c'_1(|\phi|^2) F_{\alpha \beta} F^{\alpha \beta} + \phi c'_2(|\phi|^2) F_{\alpha \beta} F_{\gamma \delta} \epsilon^{\alpha \beta \gamma \delta}
\end{equation}

We again assume that we have a stationary black hole solution to these equations, with a Killing horizon $\NN$. What needs more subtlety is our assumption of the invariance of the matter fields on the symmetry corresponding to the Killing vector $\xi = \frac{\partial}{\partial \tau}$. The assumption that $\phi$ and $A_\mu$ are independent of $\tau$ is no longer appropriate here because the conditions $\partial_\tau \phi = 0$ and $\partial_\tau A_\mu=0$ are not invariant under an electromagnetic gauge transformation. Instead, we need to modify our notion of a symmetry of the system.

For the metric, a Killing vector symmetry corresponds to invariance under the diffeomorphism $\tau \rightarrow \tau+s$ for all $s$, i.e. $g_{\mu \nu}(\tau + s) = g_{\mu \nu}(\tau)$. This diffeomorphism can be viewed as a one-parameter coordinate gauge transformation of the metric, labelled by $s$. A complete gauge transformation of our matter fields $\phi$ and $A_\mu$ would be a \textit{combined} diffeomorphism and electromagnetic gauge transformation. Hence we assume the notion of symmetry for $\phi$ and $A_\mu$ is their invariance under a one-parameter combined diffeomorphism and electromagnetic gauge transformation. More precisely, given any fixed gauge of $A_\mu$ and $\phi$, we assume there exists some one-parameter family of functions $\theta_s$ with $\theta_0 = 0$ such that
\be
    A_{\mu}(\tau+s)+ \partial_\mu \theta_s = A_{\mu}(\tau), \quad \quad e^{i\lambda \theta_s}\phi(\tau+s) = \phi(\tau)
\ee
for all $s$. We can take the derivative of this with respect to $s$ and set $s=0$ to obtain the conditions
\be \label{symmetryconditions}
    \partial_\tau A_\mu = -\partial_\mu \Theta, \quad \quad \partial_\tau \phi = -i \lambda \Theta \phi
\ee
where $\Theta = \frac{\text{d} \theta_s}{\text{d} s }\Big|_{s=0}$.
These are the conditions\footnote{These conditions can be proved to be equivalent to the conditions assumed in (2.9) of \cite{An:2023} (a recent paper discussing stationary black hole solutions with charged scalar hair). The formulation above avoids the need to define the phase of $\phi$ however.} which constrain the $\tau$-dependence of $A_\mu$ and $\phi$. Note the first condition implies $\partial_\tau F_{\mu \nu} = 0$, which was the condition we assumed in the real scalar field case.

Let us now make an electromagnetic gauge transformation of the form (\ref{gaugetransform}). Using (\ref{symmetryconditions}), the $\tau$-dependence of $\Tilde{A}_\mu$ and $\Tilde{\phi}$ can be found to be
\be
    \partial_\tau \Tilde{A}_\mu = - \partial_\mu ( \Theta - \partial_\tau \chi), \quad \quad \partial_\tau \Tilde{\phi} = - i \lambda ( \Theta - \partial_\tau \chi ) \Tilde{\phi}
\ee
From this we see that conditions (\ref{symmetryconditions}) are preserved under an electromagnetic gauge transformation, so long as we relabel $\Tilde{\Theta} = \Theta - \partial_\tau \chi$. In particular, we can take $\chi = \int^\tau \Theta(\tau') \text{d}\tau'$ to find a gauge in which 
\be
    \partial_\tau \Tilde{A}_\mu = 0, \quad \quad \partial_\tau \Tilde{\phi} = 0
\ee
We will drop the tildes and work in this gauge to prove the Generalized 0th Law. It must then hold in all gauges because the statements $\partial_{A} X\big|_{\rho=0} = 0$ and $F_{\tau A}\big|_{\rho=0} = 0$ are gauge-independent. Note that $F_{\tau A} = - \partial_A A_\tau$ in this gauge. To prove $F_{\tau A}\big|_{\rho=0} = 0$, we will actually show that $A_{\tau}\big|_{\rho=0} = 0$ in this gauge.

The proof follows in a similar fashion to that of Einstein-Maxwell-Scalar EFT above, with modifications to deal with the fact that $A_\mu$ can now appear outside of the gauge-independent combination $F_{\mu \nu}$. The relevant parts of the equation of motion will be $E^{(0)}_{\tau \tau}[\Phi_I]\big|_{\rho=0}$, $E^{(0)}_{\tau A}[\Phi_I]\big|_{\rho=0}$ and $E^{(0)}_{\tau}[\Phi_I]\big|_{\rho=0}$ as before. In this gauge they can be found to be:
\begin{align}
    E^{(0)}_{\tau \tau}[\Phi_I]\big|_{\rho=0} &= - \lambda^2 A_\tau^2 |\phi|^2 -\frac{1}{2} c_1(|\phi|^2) F_{\tau A} F_{\tau B} h^{A B}\\
    E^{(0)}_{\tau A}[\Phi_I]\big|_{\rho=0} &= -\frac{1}{2} \partial_A X - \frac{1}{2} i \lambda A_\tau \left(\phi^{*} \partial_A \phi - \phi \partial_A \phi^* - 2i\lambda A_A |\phi|^2\right) - \frac{1}{2} c_1(|\phi|^2)\left( F_{A B} h^{B C} - F_{\tau \rho} \delta_{A}^{C} \right) F_{\tau C}\\
    E^{(0)}_{\tau}[\Phi_I]\big|_{\rho=0} &= 2\lambda^2 A_\tau |\phi|^2 + h^{A B} \DD_{A}\Big[c_1(|\phi|^2)F_{\tau B} - 8 c_2(|\phi|^2) \epsilon_{B}^{\,\,\,\,C}F_{\tau C}\Big] \label{Etaucharge}
\end{align}

For simplicity, we will assume $\phi^{(0)}$ is not identically 0 on $\NN$ in the following proof. The case $\phi^{(0)}\big|_{\rho=0} \equiv 0$ adds a variety of technical difficulties that stray from the main argument, and is dealt with in Appendix \ref{0thLawChargedProof}. 

Let us first look at the equations of motion at order $l^0$. From $E^{(0)}_{\tau \tau}[\Phi^{(0)}_I]\big|_{\rho=0} = 0$, we get $A^{(0)}_\tau \phi^{(0)} \big|_{\rho = 0} = 0$ and $F^{(0)}_{\tau A}\big|_{\rho = 0} = 0$. But since $F^{(0)}_{\tau A} = - \partial_A A^{(0)}_\tau$, the latter condition means $A^{(0)}_\tau$ is constant on the horizon ($\partial_\tau A^{(0)}_\tau = 0$ in this gauge). Therefore we can extract $A^{(0)}_\tau\big|_{\rho = 0} = 0$ from the former condition because we are assuming $\phi^{(0)}$ is not identically 0 on the horizon. Plugging $A^{(0)}_\tau\big|_{\rho = 0} = 0$ into $E^{(0)}_{\tau A}[\Phi^{(0)}_I]\big|_{\rho=0} =0$, we get $\partial_A X^{(0)}\big|_{\rho = 0} = 0$, and so the Generalized 0th Law is proved at order $l^0$.

We now take our inductive hypothesis to be $\partial_A X^{(n)}\big|_{\rho = 0} = 0$ and $A^{(n)}_\tau \big|_{\rho = 0} = 0$ for $n<k$. We again decompose the order $l^k$ part of $E_I[\Phi_J] = 0$ into $E^{[k]}_I[\Phi^{[k-1]}_J]$ plus $E^{(0)}_I[\Phi_J^{(0)}+ l^k \Phi^{(k)}_J]$ linearized around $\Phi^{(0)}_J$.

$E^{[k]}_{\tau A}[\Phi^{[k-1]}_J]$ and $E^{[k]}_{\tau}[\Phi^{[k-1]}_J]$ can again be shown to vanish on the horizon by the statement that positive boost weight quantities vanish on the horizon if $\partial_A X\big|_{\rho=0} = 0$ and $A_\tau\big|_{\rho=0} = 0$ in this gauge. The proof of this statement follows exactly as in Section \ref{PBWQH}, except we must additionally show that positive boost weight quantities made from $A_\mu$, such as $A_v, \partial_v A_A$ and $\partial_{v v} A_r$, vanish on the horizon. To do this we again use the relation between Killing vector GNCs and affinely parameterized GNCs.

Using $\partial_\tau A_\mu =0$, we can write $A_\mu$ in Killing vector GNCs as (suppressing $x^C$ dependence):
\be \label{AKilling}
    A= A_\tau(\rho) \text{d}\tau + A_\rho(\rho) \text{d}\rho + A_A(\rho) \text{d}x^A
\ee
We transform to affinely parameterized GNCs, $\rho = r (\kappa v+1)$, $\tau = \frac{1}{\kappa}\log \left(\kappa v + 1\right)$, to get
\begin{equation} \label{Atransform}
    \begin{split}
        A_v(r,v)= \frac{1}{\kappa v+1} A_{\tau}\big(r\left(\kappa v+ 1\right)&\big)+ \kappa r A_\rho\big(r\left(\kappa v+1\right)\big),\\
        A_r(r,v)=(\kappa v+1) A_\rho\big(r\left(\kappa v+1\right)\big), \quad & \quad A_A(r,v)=A_A\big(r\left(\kappa v+1\right)\big)
    \end{split}
\end{equation}

Positive boost weight quantities involving $A_\mu$ are given by $\partial_{A_1}{... \partial_{A_n}{ \partial_{r}^{p} \partial_{v}^{q} \varphi } }$ with $\varphi \in \{ A_v, \partial_v A_A, \partial^2_v A_r\}$ and $q\geq p$. It is easy to show that $\partial_v A_A$ and $\partial^2_v A_r$ have the functional form $r f_{\varphi}\big(r(\kappa v + 1)\big)$ which, as shown in Section \ref{PBWQH}, means their positive boost weight derivatives vanish on the horizon. $\partial_r\partial_v A_v$ also has this functional form, so if $q\geq p \geq 1$ then  $\partial_{A_1}{... \partial_{A_n}{ \partial_{r}^{p} \partial_{v}^{q} A_v } }$ vanishes on the horizon. This leaves only terms with $p=0$, however one can show
\be \label{dvAv}
    \partial^q_v A_v\big|_{r=0} = \frac{(-\kappa)^q}{(\kappa v+1)^{q+1}} A_{\tau}\big|_{\rho=0}
\ee
Therefore, if $A_\tau\big|_{\rho=0} = 0$ then $\partial_{A_1}... \partial_{A_n} \partial_{v}^{q} A_v$ also vanishes on the horizon.

Therefore we can return to our inductive step and look at $E^{(0)}_I[\Phi_J^{(0)}+ l^k \Phi^{(k)}_J]\big|_{\rho=0}$ linearized around $\Phi^{(0)}_J$ for $I=(\tau A)$ and $I=\tau$. For $E^{(0)}_{\tau A}[\Phi_I]\big|_{\rho=0}$, use $A_\tau^{(0)}\big|_{\rho=0}=0$ to obtain
\begin{multline}
    \text{At order} \,\, l^k, \quad E_{\tau A}[\Phi_J]\Big|_{\rho=0} = -\frac{1}{2} l^k \partial_{A} X^{(k)} - \frac{1}{2} i l^k \lambda A^{(k)}_\tau \left(\phi^{(0)*} \partial_A \phi^{(0)} - \phi^{(0)} \partial_A \phi^{(0)*} - 2i\lambda A^{(0)}_A |\phi^{(0)}|^2\right) -\\
    \frac{1}{2} c_1(|\phi^{(0)}|^2)\left( F^{(0)}_{A B} h^{(0) B C} - F^{(0)}_{\tau \rho} \delta_{A}^{C} \right) l^k F^{(k)}_{\tau C} = 0
\end{multline}
which implies that $\partial_{A} X^{(k)}\big|_{\rho=0} = 0$ if $A^{(k)}_{\tau}\big|_{\rho=0} = 0$.

Finally, for $E^{(0)}_{\tau}[\Phi_I]\big|_{\rho=0}$ use $A_\tau^{(0)}\big|_{\rho=0}=0$ to get
\be
    \text{At order} \,\, l^k, \quad E_{\tau}[\Phi_J]\Big|_{\rho=0} = 2\lambda^2 l^k |\phi^{(0)}|^2 A^{(k)}_\tau + l^k h^{(0) A B} \DD^{(0)}_A \left[ c_1(|\phi^{(0)}|^2) F^{(k)}_{\tau B} - 8 c_2(|\phi^{(0)}|^2) \epsilon_{\,\,B}^{(0)\,C} F^{(k)}_{\tau C} \right] = 0
\ee
Plugging in $F^{(k)} = - \partial_A A^{(k)}_\tau$ gives
\be
2\lambda^2 |\phi^{(0)}|^2 A^{(k)}_\tau - h^{(0) A B} \DD^{(0)}_A \left[ c_1(|\phi^{(0)}|^2) \DD^{(0)}_B A^{(k)}_\tau - 8 c_2(|\phi^{(0)}|^2) \epsilon_{\,\,B}^{(0)\,C} \DD^{(0)}_C A^{(k)}_\tau \right]=0
\ee
Integrate this against $\sqrt{h^{(0)}}A_\tau^{(k)}$ over $C(\tau)$, integrate by parts, and again use the antisymmetry of $\epsilon^{(0) B C}$ to obtain
\be
\int_{C(\tau)} \text{d}^{d-2}x \sqrt{h^{(0)}}\left[ 2\lambda^2 l^k |\phi^{(0)}|^2 (A^{(k)}_\tau)^2+  c_1(|\phi^{(0)}|^2) h^{(0) A B} \left(\partial_A A^{(k)}_\tau \right) \left( \partial_B A^{(k)}_\tau \right) \right]=0
\ee
Both terms in the integrand are manifestly non-negative, and hence $A^{(k)}_\tau\big|_{\rho=0} = 0$, once again using our assumption that $\phi^{(0)}$ is not identically 0 on the horizon\footnote{$\partial_{\tau}\phi = 0$ in this gauge, so if $\phi^{(0)}$ is not identically 0 on the horizon then it is also not identically 0 on any individual spatial cross section $C(\tau)$.}. This completes the induction.

Therefore, $\partial_A X\big|_{\rho=0} = 0$ and $F_{\tau A}\big|_{\rho=0} = 0$ in this gauge. But since these statements are gauge-independent, the Generalized 0th Law holds in all gauges for this charged scalar field EFT.

\section{The Second Law}

The Second Law of Black Hole Mechanics is the statement that the entropy of dynamical (i.e. non-stationary) black hole solutions is non-decreasing in time. 

In standard 2-derivative GR coupled to matter satisfying the Null Energy Condition, it can be proved that the area $A(v)$ of a spacelike cross section of the horizon is always non-decreasing in $v$. This supports a natural interpretation of the entropy of a black hole as proportional to its area. However, when we include higher derivative terms in the metric or matter fields, $A(v)$ is no longer necessarily non-decreasing. Therefore we need a generalization of the definition of entropy in order to satisfy a Second Law. Whilst there has been no answer that applies to all situations, a fruitful avenue has been to study dynamical black holes that are close to being stationary in the regime of validity of EFT. 

Here we extend the recent work of \cite{DaviesReall:2023}, \cite{Hollands:2022} and \cite{Biswas:2022} to define an entropy that satisfies the Second Law up to an arbitrarily high order $l^N$ in our EFT for Einstein-Maxwell-Scalar theory.

\subsection{Perturbations Around Stationary Black Holes}

We consider the scenario of a black hole settling down to a stationary equilibrium. As such, we assume our dynamical black hole solution is close to some stationary black hole solution $\Phi^{st}_I$, and consider perturbations around it.

To make the order of perturbation precise, we use the statement proved in Section \ref{PBWQH}: positive boost weight quantities vanish on the horizon of a stationary black hole solution. Our construction of the entropy, $S(v)$, will consist of manipulating affinely parameterized GNC quantities evaluated on the horizon, and so the number of factors of positive boost weight quantities determines the order of perturbation. For example, $K_{A B}$ has boost weight $+1$ and so a term such as $K_{A B} K^{A B}$ is quadratic order. 

To zeroth order, our entropy $S(v)$ will be the Wald entropy of the stationary solution $\Phi^{st}_I$. This is constant in time, and so $\dot{S}(v) = 0$ to zeroth order. 

To linear order, our entropy will be the one defined by Biswas, Dhivakara and Kundu (BDK) in \cite{Biswas:2022}, where it is proved to be constant at linear order. Therefore, linearized around $\Phi^{st}_I$, we have $\delta \dot{S}(v) = 0$.

This paper extends the BDK entropy by adding terms quadratic in positive boost weight quantities in a similar fashion to how Hollands et al \cite{Hollands:2022} and then Davies and Reall \cite{DaviesReall:2023} extended the Iyer-Wald-Wall entropy \cite{Wall:2015} \cite{Bhattacharyya:2021jhr}. We aim to show that such an entropy satisfies the Second Law \textit{non-pertubatively}, i.e. $\dot{S}(v)$ is non-negative to all orders in perturbations around a black hole. However, this can only be done in the regime of validity of EFT.

\subsection{Regime of Validity of EFT}

We shall not be interested in arbitrary black hole solutions of our Einstein-Maxwell-Scalar EFT. In general, there will be pathological solutions that do not obey the 2nd Law. 

Instead, we shall consider only black hole solutions that lie within the \textit{regime of validity of the EFT}. This is defined in \cite{Hollands:2022} as follows. We assume we
have a 1-parameter family of dynamical black hole solutions labelled by a lengthscale $L$ (e.g. the size of the black hole or some other dynamical length/time scale) such that $\NN$ is the event
horizon for all members of the family (this is a gauge choice). We assume there exist affinely parameterized GNCs defined near $\NN$ such that any quantity constructed from $n$ derivatives of $\{\alpha$, $\beta_A$, $\mu_{AB}$, $\phi\}$ or $n-1$ derivatives of $F_{\mu \nu}$ is bounded
by $C_n$/$L^n$
for some constant $C_n$, and that $|V| L^2\leq 1$. Then the solution lies within the regime of
validity of EFT for sufficiently small $l/L$. This definition captures the notion of a solution “varying
over a length scale L” with L large compared to the UV scale l.

Note that we are no longer assuming the black hole solution is analytic in $l$, as we did in the proof of the 0th Law. This is not an applicable assumption in a dynamical situation because treating the solution as an expansion in $l$ can typically lead to secular growth. See Footnote 1 of \cite{Davies:2023} for an example of such a situation.

\subsection{Order of the EFT}

Our EFT action (\ref{action}) is made up of potentially infinitely many higher derivative terms. In practice, we will only know finitely many of the coefficients of these terms, and so there will be some $N$ for which we know all the terms with $N+1$ or fewer derivatives. In this case we only fully know part of the equations of motion, which satisfies
\be
E_I^{[N-1]} \equiv E_I^{(0)} + \sum_{n=1}^{N-1}l^{n}E^{(n)}_I = O(l^N)
\ee
Since we only know our theory up to some accuracy of order $l^N$, it is reasonable to expect our Second Law to only be provable up to order $l^N$ terms. This is indeed what we will show, i.e.
\be \label{2ndLawSense}
    \dot{S}(v) \geq -O(l^N)
\ee
where the RHS of the inequality signifies that $\dot{S}(v)$ might be negative but only by an $O(l^N)$ amount. This means that the better we know our EFT, the closer we can construct an entropy satisfying a complete Second Law. The entropy $S(v)$ will contain terms of up to $N-2$ derivatives.

\subsection{Review of Recent Progress on the 2nd Law in Vacuum Gravity EFT}

Before we jump into proving the 2nd Law for our Einstein-Maxwell-Scalar EFT, we shall briefly review the recent progress in Einstein-Scalar EFTs that we are building on. Firstly, the Iyer-Wald-Wall entropy, which satisfies the 2nd Law to linear order in perturbations around a stationary black hole. Secondly, the extension made by Hollands, Kovacs and Reall (HKR) to define an entropy that satisfies the 2nd Law to quadratic order in perturbations, up to order $l^N$ terms. Finally, the very recent work by Davies and Reall in the companion to this paper that adds extra terms to the HKR entropy which result in the 2nd Law being satisfied non-perturbatively, up to order $l^N$ terms.

The Iyer-Wald-Wall entropy was devised by Wall \cite{Wall:2015} as an improvement on the entropy defined by Iyer and Wald \cite{Iyer:1994ys}. It was formalised by Bhattacharyya et al. \cite{Bhattacharyya:2021jhr}. It applies to any theory of gravity and a scalar field with diffeomorphism invariant Lagrangian (under no EFT assumption). The approach is to use affinely parameterized GNCs and study the $E_{v v}$ equation of motion. They prove that it can always be manipulated into the following form on the horizon:
\begin{equation} \label{IWWdef}
    -E_{v v}\Big|_{\mathcal{N}} = \partial_{v}\left[\frac{1}{\sqrt{\mu}} \partial_{v}\left(\sqrt{\mu} s^{v}_{IWW}\right) + D_{A}{ s^A }\right] + ...
\end{equation}
where the ellipsis denotes terms that are quadratic or higher order in positive boost weight quantities and hence quadratic order in perturbations around a stationary black hole. $(s^{v}_{IWW}, s^A)$ is denoted the Iyer-Wald-Wall entropy current. They are only defined uniquely up to first order in positive boost weight quantities, as any higher order terms can be absorbed into the ellipsis. As proved in \cite{Hollands:2022}, the higher order terms can be fixed so that $s^v_{IWW}$ is invariant under a change of GNCs.

The Iyer-Wald-Wall entropy of the horizon cross-section $C(v)$ is then defined as
\begin{equation}
    S_{IWW}(v) = 4 \pi \int_{C(v)} \text{d}^{d-2}x \sqrt{\mu} \,s^{v}_{IWW} 
\end{equation}
Taking the $v$-derivative of this gives
\begin{equation} \label{vderiv}
\begin{split}
    \dot{S}_{IWW} =& 4 \pi\int_{C(v)} \text{d}^{d-2}x \sqrt{\mu}\left[\frac{1}{\sqrt{\mu}} \partial_{v}(\sqrt{\mu} s^{v}_{IWW}) + D_{A}{s^A}\right] \\
    =& - 4 \pi \int_{C(v)} \text{d}^{d-2}x \sqrt{\mu} \int_{v}^{\infty} \text{d}v' \, \partial_{v} \left[ \frac{1}{\sqrt{\mu}} \partial_{v}(\sqrt{\mu} s^{v}_{IWW}) + D_{A}{s^A} \right](v',x)
\end{split}
 \end{equation}
where in the first line we trivially added the total derivative $\sqrt{\mu} D_{A}{s^A}$ to the integrand, and in the second line we assumed the black hole settles to the stationary black hole solution, so positive boost weight quantities vanish on the horizon at late times. The integrand can then be swapped for terms that are quadratic or higher in positive boost weight quantities using (\ref{IWWdef}) and the equation of motion $E_{v v}=0$. Thus $\dot{S}_{IWW}$ is quadratic order in perturbations around a stationary black hole, and so $\dot{S}_{IWW}\big|_{\Phi^{st}_I} = 0$ and the first variation $\delta \dot{S}_{IWW} = 0$. Therefore $S_{IWW}$ satisfies the 2nd Law to linear order. Even stronger than that, its change in time vanishes to linear order rather than being non-negative.

To see a possible increase in the entropy, we must go to quadratic order, which is what the extension by Hollands, Kovacs and Reall achieves in \cite{Hollands:2022}. They show that if the theory and solution lie in the regime of EFT the ellipsis in (\ref{IWWdef}) can be manipulated into the following form:
\begin{multline}\label{HKRForm}
    -E_{v v}\Big|_{\mathcal{N}} = \partial_{v}\left[\frac{1}{\sqrt{\mu}} \partial_{v}\left(\sqrt{\mu} s^{v}_{IWW}\right) + D_{A}{ s^A }\right] +\\
    \partial_{v}\left[\frac{1}{\sqrt{\mu}} \partial_{v}\left(\sqrt{\mu} \varsigma^{v}\right)\right] + \left(K_{A B}+X_{A B}\right) \left(K^{A B}+X^{A B}\right)+ \frac{1}{2} \left( \partial_v \phi + X\right)^2 + D_{A}{Y^{A}} + O(l^{N})
\end{multline}
$X^{A B}$ and $X$ are linear in positive boost weight quantities, and $Y^A$ and the $O(l^N)$ terms are quadratic. To do this they go "on-shell", meaning they use the equations of motion to swap out various terms. The entropy density is then defined by $s^v_{HKR} = s^v_{IWW} + \zeta^v$, and the Hollands-Kovacs-Reall entropy is given by
\begin{equation} 
    S_{HKR}(v) = 4 \pi \int_{C(v)} \text{d}^{d-2}x \sqrt{\mu} \,s^{v}_{HKR} 
\end{equation}
Just as in (\ref{vderiv}), we can take the $v$-derivative of this and substitute in (\ref{HKRForm}) to get
\begin{equation} \label{dotSHKR}
        \dot{S}_{HKR}(v)= 4 \pi \int_{C(v)} \text{d}^{d-2} x \sqrt{\mu} \int_{v}^{\infty} \text{d}v' \left[W^2 + D_{A}{Y^{A}} + O(l^{N}) \right](v',x)
\end{equation}
where $W^2= \left(K_{A B} + X_{A B}\right) \left(K^{A B} + X^{A B}\right) + \frac{1}{2} \left( \partial_v \phi + X\right)^2$. Since $W^2$, $Y^A$ and the $O(l^N)$ terms are quadratic in positive boost weight, they vanish on the horizon along with their first variations, so once again we have $\dot{S}_{HKR}|_{\Phi^{st}_I} = 0$ and $\delta \dot{S}_{HKR} = 0$.

Turning to the second variation, $\delta^2 W^2 = (\delta W)^2$ is a positive definite form so must be non-negative. The second term is 
\begin{equation}
    \int_{C(v)} \text{d}^{d-2} x \sqrt{\mu}\bigg|_{\Phi^{st}_I} \int_{v}^{\infty} \text{d}v' D_{A}\bigg|_{\Phi^{st}_I} \delta^2 Y^A
\end{equation}
The induced metric $\mu_{A B}$ is independent of $v$ on the horizon for the stationary solution $\Phi^{st}_I$ (see (\ref{APGNCsSymmetry}) with $r=0$). Therefore, we can exchange the order of integrations and see the integrand is a total derivative on $C(v)$. Hence this integral vanishes and so $\dot{S}_{HKR}$ is non-negative to quadratic order, modulo $O(l^{N})$ terms. Thus it satisfies the 2nd Law to quadratic order in the sense of $\delta^2 \dot{S}_{HKR} \geq -O(l^N)$.

Finally, the recent work by Davies and Reall \cite{DaviesReall:2023} showed that for vacuum gravity EFTs (i.e. with no scalar field $\phi$) we can manipulate the terms in the RHS of (\ref{dotSHKR}) further into the form
\begin{multline}
        \dot{S}_{HKR}(v)= -\frac{d}{dv} \left(4\pi \int_{C(v)} \text{d}^{d-2}x \sqrt{\mu(v)} \sigma^{v}(v) \right) +\\
        4 \pi \int_{C(v)} \text{d}^{d-2} x \sqrt{\mu} \int_{v}^{\infty} \text{d}v' \left[(K_{A B} + Z_{A B})(K^{A B} + Z^{A B}) + O(l^{N}) \right](v,v',x)
\end{multline}
where $Z_{A B}(v,v')$ is made up of so-called "bilocal" quantities, meaning they depend on both on $v$ and the integration variable $v'$. The final integral is a positive definite form up to $O(l^N)$ terms, and hence the entropy defined by 
\begin{equation}
    S(v) = 4 \pi \int_{C(v)} \text{d}^{d-2}x \sqrt{\mu} \,s^{v} 
\end{equation}
with $s^v=s^v_{HKR}+\sigma^v$ satisfies $\dot{S}(v)\geq -O(l^N)$, i.e. it satisfies the 2nd Law non-perturbatively up to $O(l^N)$ terms.

We shall now show how the above methods can be extended to define an entropy for Einstein-Maxwell-Scalar EFT (with real, uncharged scalar field) that satisfies the 2nd Law in the same sense.

\section{Proof of the Second Law}

\subsection{The Desired Generalizations}

Throughout the proof, we work exclusively in affinely parameterized GNCs. Rewriting here for convenience,
\begin{equation}
    g = 2 \text{d}v \text{d}r - r^2 \alpha(r,v,x^C) \text{d}v^2 -2 r\beta_{A}(r,v,x^C) \text{d}v \text{d}x^A + \mu_{A B}(r,v,x^C) \text{d}x^A \text{d}x^B
\end{equation}
We raise and lower $A,B,C,...$ indices with $\mu_{A B}$ and denote the covariant derivative with respect to $\mu_{A B}$ by $D_{A}$. As well as $K_{A B} \equiv \frac{1}{2} \partial_{v}{\mu_{A B}}, \bar{K}_{A B} \equiv \frac{1}{2} \partial_{r}{ \mu_{A B} }$ defined previously, it will be useful to define
\begin{equation} \label{Maxwellquants}
    K_{A} \equiv F_{v A}, \quad \bar{K}_{A} \equiv F_{r A}, \quad \psi = F_{v r}
\end{equation}
$K_A$ has boost weight $+1$, $\bar{K_A}$ has boost weight $-1$ and $\psi$ has boost weight $0$.

The Iyer-Wald-Wall entropy has already been generalized to Einstein-Maxwell-Scalar EFT with real scalar field by Biswas et al. in \cite{Biswas:2022}, as discussed in Section \ref{BDKEntropy}.

The main body of the proof consists of generalizing the HKR entropy by studying the $E_{v v}$ component of the equations of motion in affinely parameterized GNCs on the horizon. Our generalization of equation (\ref{HKRForm}) is as follows. We will show that on-shell (i.e. by using the known part of the equations of motion $E^{[N-1]}_I = O(l^N)$) we can bring $E_{v v}\big|_{\NN}$ into the form
\begin{multline} \label{EvvForm}
    - E_{v v}\Big|_{\mathcal{N}} = \partial_{v}\left[\frac{1}{\sqrt{\mu}} \partial_{v}\left(\sqrt{\mu} s^{v}_{HKR}\right) + D_{A}{ s^A }\right] + \left(K_{A B} + X_{A B}\right) \left(K^{A B} + X^{A B}\right)\\
    + \frac{1}{2} c_1(\phi) \left( K_{A} + X_{A} \right)\left(K^{A} + X^{A} \right) + \frac{1}{2} \left( \partial_v \phi + X\right)^2 + D_A Y^A + O(l^N)
\end{multline}
where $X = \sum_{n=1}^{N-1} l^n X^{(n)}, X_A = \sum_{n=1}^{N-1} l^n X_A^{(n)}, X_{A B} = \sum_{n=1}^{N-1} l^n X^{(n)}_{A B}$ (boost weights +1) are linear or higher in positive boost weight quantities, and $Y^A = \sum_{n=1}^{N-1} l^n Y^{(n) A}$ (boost weight +2) and the $O(l^N)$ terms are quadratic. $s^v_{HKR} = \sum_{n=0}^\infty l^n s^{(n) v}_{HKR}$ has boost weight 0 and $s^A=\sum_{n=1}^\infty l^n s^{(n) A}$ has boost weight +1. They will be invariant upon change of electromagnetic gauge.

The generalization of the HKR entropy of the spacelike cross section $C(v)$ is then defined to be
\be
    S_{HKR}(v) = 4\pi \int_{C(v)}\text{d}^{d-2}x \sqrt{\mu} s^{v}_{HKR}
\ee

The proof that $\delta^2 \dot{S} \geq -O(l^N)$ follows in the same way as for $S_{HKR}$ detailed above, with the only change being $W^2= \left(K_{A B} + X_{A B}\right) \left(K^{A B} + X^{A B}\right) + \frac{1}{2} \left( \partial_v \phi + X\right)^2 + \frac{1}{2} c_1(\phi) \left( K_{A} + X_{A} \right)\left(K^{A} + X^{A} \right)$. The additional term is still a positive definite form and so the same method holds.

The algorithm to write $E_{v v}$ in the form (\ref{EvvForm}) is very similar to the algorithm devised by Hollands, Kovacs and Reall \cite{Hollands:2022} (and further detailed in \cite{Davies:2023}) for Einstein-Scalar EFT. We will emphasise where we need to extend the HKR algorithm to apply to our Einstein-Maxwell-Scalar EFT.

Finally, we generalize the entropy defined by Davies and Reall by proving we can write $\dot{S}_{HKR}(v)$ as
\begin{multline} \label{HKRGeneralization}
        \dot{S}_{HKR}(v)= -\frac{d}{dv} \left(4\pi \int_{C(v)} \text{d}^{d-2}x \sqrt{\mu(v)} \sigma^{v}(v) \right)\\
        +4 \pi \int_{C(v)} \text{d}^{d-2} x \sqrt{\mu} \int_{v}^{\infty} \text{d}v' \Big[(K_{A B} + Z_{A B})(K^{A B} + Z^{A B}) + \frac{1}{2} c_1(\phi) \left( K_{A} + Z_{A} \right)\left(K^{A} + Z^{A} \right)\\
        + \frac{1}{2} \left( \partial_v \phi + Z\right)^2 + O(l^{N}) \Big](v,v',x)
\end{multline}
for bilocal $Z_{A B}, Z_{A}$ and $Z$. Thus we can define an entropy
\be
    S(v) = 4\pi \int_{C(v)}\text{d}^{d-2}x \sqrt{\mu} s^v
\ee
with $s^v=s^{v}_{HKR}+\sigma^v$ that satisfies $\dot{S}_{HKR}(v)\geq -O(l^N)$ as desired.

\subsection{Leading Order Einstein-Maxwell-Scalar Theory}

Let us look at how this works for the leading order Einstein-Maxwell-Scalar terms arising from $\LL_{2}$. The leading order part of the $(\alpha \beta)$ equation of motion is
\begin{equation}
    E^{(0)}_{\alpha\beta}= R_{\alpha \beta} - \frac{1}{2} \nabla_{\alpha} \phi \nabla_{\beta} \phi -\frac{1}{2} c_1(\phi) F_{\alpha \delta} F_{\beta}^{\,\,\,\,\delta} - \frac{1}{2} g_{\alpha \beta} \left( R - V(\phi) -\frac{1}{2} \nabla_{\gamma}{\phi} \nabla^{\gamma}{\phi} - \frac{1}{4} c_1(\phi) F_{\gamma \delta} F^{\gamma\delta} \right)
\end{equation}
In affinely parameterized GNCs on the horizon, we have
\begin{equation}
    E^{(0)}_{v v}\big|_{\NN}= R_{v v} - \frac{1}{2} (\partial_{v} \phi)^2 -\frac{1}{2} c_1(\phi) K_A K^A
\end{equation}
Using $R_{v v}|_{\NN} = -\mu^{A B} \partial_{v} K_{A B} + K_{A B} K^{A B}$ and $\partial_{v} \sqrt{\mu} = \sqrt{\mu}\mu^{A B} K_{A B}$, we can write this as
\begin{equation} \label{Orderl^0}
    -E^{(0)}_{v v}\big|_{\NN}= \partial_{v}\left[\frac{1}{\sqrt{\mu}} \partial_{v}\left(\sqrt{\mu}\right) \right] + K_{A B} K^{A B} + \frac{1}{2} c_1(\phi) K_A K^A + \frac{1}{2} (\partial_{v} \phi)^2 
\end{equation}
This is of the form (\ref{EvvForm}) with $s^{(0)v}_{HKR} = 1$ and $s^{(0) A} = X^{(0)} = X^{(0)}_A = X^{(0)}_{A B} = Y_A^{(0)} = 0$. Because there is no total derivative term $D_A Y^{(0)A}$, there are no further manipulations needed to get to (\ref{HKRGeneralization}) with $\sigma^v=0$, $Z^{(0)}=X^{(0)}$, $Z_{A}^{(0)}=X_{A}^{(0)}$ and $Z_{A B}^{(0)}=X_{A B}^{(0)}$. Thus we have proved our theory satisfies the 2nd Law non-perturbatively at leading order $l^0$ (which of course can be proved by the usual proof of the 2nd Law on the 2-derivative theory).

We will ultimately work through the higher order terms order-by-order to mould them into the correct form. To get to that point however, we must start with the Biswas, Dhivakara and Kundu entropy.

\subsection{The Biswas, Dhivakara and Kundu Entropy}\label{BDKEntropy}

Our starting point is the generalization of the Iyer-Wald-Wall entropy by Biswas, Dhivakara and Kundu (BDK) defined in \cite{Biswas:2022}. They prove that for any theory of gravity, electromagnetism and a real (uncharged) scalar field with diffeomorphism-invariant and electromagnetic gauge-independent Lagrangian, the $E_{v v}$ component of the equations of motion can be brought into the following form on the horizon\footnote{In \cite{Biswas:2022}, they have an additional $T_{v v}$ in this defining equation, which is the part of the energy momentum tensor arising from the minimal coupling part of the matter sector Lagrangian. However they also show that $T_{v v}$ is quadratic in positive boost weight quantities if $T_{\mu \nu}$ satisfies the NEC, which is the case for our 2-derivative Einstein-Maxwell-Scalar theory and hence we can absorb $T_{v v}$ into the ellipsis.}:
\begin{equation} \label{IWW}
    -E_{v v}\Big|_{\mathcal{N}} = \partial_{v}\left[\frac{1}{\sqrt{\mu}} \partial_{v}\left(\sqrt{\mu} s^{v}_{BDK}\right) + D_{A}{ s^A }\right] + ...
\end{equation}
where the ellipsis denote terms at least quadratic in positive boost weight quantities. We will call the quantity $(s^{v}_{BDK}, s^A)$ the BDK entropy current. It is proved to only depend on the electromagnetic potential through $F_{\mu \nu}$, and is thus invariant upon a gauge transformation $A_\mu \rightarrow A_\mu + \partial_\mu \chi$. $s^A$ is a vector in $A, B,..$ indices on $C(v)$ whilst $s^v$ is a scalar. They are only defined uniquely up to linear order in positive boost weight quantities, as any higher order terms can be absorbed into the ellipsis. As discussed in Section \ref{gaugeinvariance}, we assume we can fix the higher order terms so that $s^v_{BDK}$ is invariant under a change of GNCs.

The BDK entropy is then defined by
\be
S_{BDK}(v) = 4 \pi \int_{C(v)}\text{d}^{d-2}x \sqrt{\mu} s^{v}_{BDK}
\ee
This can be proved to satisfy $\delta \dot{S}_{BDK} = 0$ in the same way as above, since the ellipsis is quadratic in positive boost weight quantities.

For our Einstein-Maxwell-Scalar EFT, we can calculate $s^v_{BDK}$ to all orders in $l$ and take (\ref{IWW}) as our starting point. We group all the remaining terms in the ellipsis and define
\begin{equation} \label{Qdef}
    H\equiv - E_{v v}\Big|_{\mathcal{N}} - \partial_{v}\left[\frac{1}{\sqrt{\mu}} \partial_{v}\left(\sqrt{\mu} s^{v}_{BDK}\right) + D_{A}{ s^A }\right]
\end{equation}

We will use the fact that $H$ is quadratic in positive boost weight terms to show we can manipulate it so that (\ref{IWW}) becomes (\ref{EvvForm}). The resulting generalization of the HKR entropy density $s^v_{HKR}$ will be 
\be
    s^v_{HKR} = s^v_{BDK} + \sum_{n=0}^{N-1}l^n \varsigma^{(n) v}
\ee
where the $\varsigma^{(n) v}$ are quadratic in positive boost weight quantities. We will not need to add any terms to $s^A$.

From (\ref{Orderl^0}), we can see that for the leading order theory $\LL_{2}$, the BDK entropy density is $s^{(0) v}_{BDK}=1$ and we need no correction, $\varsigma^{(0) v} = 0$. 

\subsection{Reducing to Allowed Terms}

To generalize the HKR entropy we study the possible quantities $H$ is made out of. $H$ comes from the equations of motion and is gauge invariant, so is made from the fields $g_{\mu \nu}, F_{\mu \nu}$, and $\phi$ and their derivatives. It is also evaluated in affinely parameterized GNCs on the horizon, and $H$ is a scalar with respect to $A,B,...$ indices. Therefore it is made from gauge invariant affinely parameterized GNC quantities of the metric and matter fields that are covariant in $A, B,...$ indices, namely
\be \label{covariant list}
    D^k{ \partial_{v}^{p} \partial_{r}^q \varphi } \,\,\, \text{for} \,\,\, \varphi \in \{\alpha, \beta_A, \mu_{A B}, R_{A B C D}[\mu], \epsilon_{A_1 ... A_{d-2}}, \phi, F_{A B}, K_A, \bar{K}_A, \psi\}
\ee
where $k,p,q\geq 0$ and we have suppressed the indices $D^k = D_{A_1}...D_{A_k}$. $K_A, \bar{K}_A$ and $\psi$ are defined in (\ref{Maxwellquants}). Section 3.3 of \cite{Davies:2023} gives commutation rules for commuting $D_A$ derivatives past $\partial_v$ and $\partial_r$ derivatives, which allow us to have all $D_A$ derivatives on the left. $R_{A B C D}[\mu]$ is the induced Riemann tensor with respect to $\mu_{A B}$. 

We will now show that we can reduce this set of possible objects that can appear on the horizon by using the equations of motion. It is worth emphasising, this reduction holds in an EFT sense, meaning \textit{it is only done up to $O(l^N)$ terms}. In the HKR procedure of \cite{Hollands:2022} for Einstein-Scalar EFT, they show how to reduce the metric and scalar field terms to the set $\mu_{A B}$, $\epsilon_{A_1 ... A_{d-2}}$, $D^k R_{A B C D}[\mu]$, $D^k \beta_{A}$, $D^k\partial_{v}^p K_{A B}$, $D^k\partial_{r}^p \bar{K}_{A B}$, $D^k\partial_{v}^p \phi$, $D^k\partial_{r}^p \phi$ with $p\geq 0$. This procedure still holds in our Einstein-Maxwell-Scalar EFT. To focus on where we need to generalize the HKR procedure, we only detail how to reduce the Maxwell terms.

We aim to reduce the set of Maxwell terms on the horizon to
\be
    D^k \psi,\,\, D^k F_{A B},\,\, D^k\partial_{v}^p K_{A},\,\, D^k\partial_{r}^q \bar{K}_{A}
\ee
To do this we must eliminate any $\partial_v$ and $\partial_r$ derivative of both $\psi$ and $F_{A B}$. We must also eliminate any $\partial_r$ derivative of $K_A$, and any $\partial_v$ derivative of $\bar{K}_A$. 

To begin we use the fact that
\be \label{Fantisym}
\partial_{\alpha}F_{\beta \gamma} + \partial_{\beta}F_{\gamma\alpha} + \partial_{\gamma}F_{\alpha \beta} = 0
\ee
which follows from $F=\text{d}A$. Taking $\alpha = v, \beta = A, \gamma = B$, we can rearrange this to\footnote{If we had explicitly picked a gauge $A_\mu$, then this relation would be trivially true and we would have fewer terms to eliminate. However, we would like to keep the entropy current manifestly gauge invariant, and hence we do not pick a gauge.}
\be
    \partial_v F_{A B} = D_A K_B - D_B K_A
\ee
Similarly, taking $\alpha = v, \beta = A, \gamma = B$ gives
\be
    \partial_r F_{A B} = D_A \bar{K}_B - D_B \bar{K}_A
\ee
These two relations allow us to eliminate all $\partial_v$ and $\partial_r$ derivatives of $F_{A B}$ in favour of other Maxwell and metric terms.

Furthermore, taking $\alpha = v, \beta = r, \gamma = A$, we get
\be \label{dvKA}
    \partial_v \bar{K}_{A} = \partial_r K_A - D_A \psi
\ee
which allows us to eliminate any $\partial_v$ derivative or mixed $\partial_v$ and $\partial_r$ derivative of $\bar{K}_A$.

To go further, we will have to use the equations of motion for the Maxwell field. In particular, we inspect the leading order part $E^{(0)}_{\alpha} = O(l)$:
\be
    \nabla^{\beta}{\Big[c_1(\phi)F_{\alpha\beta} - 4 c_2(\phi) F^{\gamma \delta} \epsilon_{\alpha \beta \gamma \delta}\Big]} = O(l)
\ee
where in theory we know all the terms on the right hand side up to $O(l^N)$. We can use $\epsilon_{\alpha \beta \gamma \delta} \nabla^{\beta} F^{\gamma \delta} = 0$ (which follows from (\ref{Fantisym})) to rewrite this as
\begin{equation} \label{Fswap}
    \nabla^{\beta}F_{\alpha\beta} = \frac{1}{c_1(\phi)}\left[ 4 c'_2(\phi) \nabla^{\beta}\phi F^{\gamma \delta} \epsilon_{\alpha \beta \gamma \delta} - c'_1(\phi) \nabla^{\beta}\phi F_{\alpha\beta} \right] + O(l)
\end{equation}
The order $l^0$ terms on the right hand only involve Maxwell terms that we are not trying to eliminate. Let us now evaluate the $v$ component of $\nabla^{\beta}F_{\alpha\beta}$ in affinely parameterized GNCs:
\be \label{Ev0}
    \nabla^{\beta}F_{v \beta} = \partial_{v} \psi + D^{A}{K_{A}} +\psi K + ...
\ee
where the ellipsis denotes terms that vanish on $\NN$. We can substitute this into (\ref{Fswap}) to get an expression for $\partial_v \psi$ on the horizon up to terms higher order in $l$:
\be \label{dvpsi}
    \partial_v \psi\big|_{\NN} = - D^{A}{K_{A}} - \psi K + \frac{1}{c_1(\phi)}\left[ 4c'_2(\phi)\epsilon^{A B}\left(2 D_{A} \phi K_B - \partial_v \phi F_{A B} \right) - c'_1(\phi)\left( \psi \partial_v \phi + K_A D^A \phi \right) \right] + O(l)
\ee
Therefore, wherever we find a $\partial_v \psi$ in $H$, we can swap it out order-by-order in $l$, pushing it to higher order with each step. Eventually it will only appear at $O(l^N)$, at which point it is not relevant to our analysis since we do not know the equations of motion at that order.

Similarly we can evaluate $\nabla^\beta F_{r \beta}$ in affinely parameterized GNCs:
\be \label{Er0}
    \nabla^{\beta}F_{r \beta} = -\partial_{r}{\psi} + D^{A}{\bar{K}_{A}} + \bar{K}^{A} \beta_{A}- \psi \bar{K} + ...
\ee
where, again, the terms in the ellipsis vanish on the horizon. We can substitute this into (\ref{Fswap}) to get an expression for $\partial_r \psi$ on the horizon:
\be \label{drpsi}
    \partial_r \psi\big|_{\NN} = D^{A}{\bar{K}_{A}} + \bar{K}^{A} \beta_{A}- \psi \bar{K} + \frac{1}{c_1(\phi)}\left[ c'_1(\phi)\left( \bar{K}_A D^A \phi - \psi \partial_r \phi  \right) - 4c'_2(\phi)\epsilon^{A B}\left( \partial_r \phi F_{A B} - 2 D_{A} \phi \bar{K}_B  \right)  \right] + O(l) 
\ee
This allows us to eliminate $\partial_r \psi$ up to $O(l^N)$ in a similar fashion.

We can take $\partial_v$ derivatives of  (\ref{dvpsi}) and (\ref{drpsi}) in order to eliminate $\partial^p_v \partial_r^q \psi$ for $p\geq 1$ and $q = 0,1$. However, we cannot naively take $\partial_r$ derivatives because these expressions are evaluated on the horizon $r=0$. Instead, we must take successive $\partial_r$ derivatives of (\ref{Fswap}) and (\ref{Er0}), and then evaluate them on $r=0$, possibly using substitution rules already calculated for lower order derivatives. This will involve taking care of the terms in the ellipsis in (\ref{drpsi}), which are given in full in Appendix \ref{dFAppendix}. However, these only ever involve lower order derivatives, for which we already have substitution rules and hence do not cause an issue. Therefore, we can eliminate all $\partial_v$ and $\partial_r$ derivatives of $\psi$ up to order $O(l^N)$.

This just leaves $\partial_r$ derivatives of $K_A$ to be eliminated, for which we look at $\nabla^\beta F_{A \beta}$:
\be \label{EA0}
    \nabla^{\beta}F_{A \beta} = -2\partial_{r}{K_{A}}+D_{A}{\psi} + D^{B}{F_{A B}}+2\bar{K}^{B} K_{A B} + 2K^{B} \bar{K}_{A B} - \psi \beta_{A} - \bar{K}_{A} K - K_{A} \bar{K} + F_{A B} \beta^{B} + ...
\ee
Substituting this into (\ref{Fswap}) gives us an expression which we can use to eliminate $\partial_r K_A$ on the horizon. Taking $\partial_r$ derivatives of (\ref{EA0}) again allows us to eliminate higher $\partial_r$ derivatives of $K_A$ because the terms in the ellipsis only involve lower order derivatives. This completes the reduction of Maxwell terms.

Combining the Maxwell terms with the metric and scalar field terms already reduced through the HKR procedure, we are left with a small set of "allowed terms":
\begin{empheq}[box=\fbox]{align}\label{AllowedTerms}
    \text{Allowed terms:} \,\, &\mu_{A B}, \,\, \mu^{A B}, \,\, \epsilon_{A_1 ... A_{d-2}}, \,\, D^k R_{A B C D}[\mu], \,\, D^k \beta_{A}, \,\, 
    D^k\partial_{v}^p K_{A B}, \,\, D^k\partial_{r}^q \bar{K}_{A B},\nonumber\\
    & D^k \psi, \,\, D^k F_{A B}, \,\, 
    D^k\partial_{v}^p K_{A}, \,\, D^k\partial_{r}^q \bar{K}_{A}, \,\,
    D^k\partial_{v}^p \phi, \,\, D^k\partial_{r}^q \phi 
\end{empheq}
In particular, the only allowed positive boost weight terms are of the form $D^k\partial_{v}^p K_{A B}$ and $ D^k\partial_{v}^p K_{A}$ with $p\geq 0$, and $D^k\partial_{v}^p \phi$ with $p\geq 1$. This will be the crucial fact that allows us to manipulate the terms in $H$.

\subsection{Manipulating Terms Order-by-Order} \label{ManipTermsObO}

Let us return to $H$. We use the above procedures to eliminate any non-allowed terms up to $O(l^N)$. Once doing so, we can write it as a series in $l$:  
\begin{equation}\label{Qseries}
    H = H^{(0)} + \sum_{n=1}^{N-1} l^n H^{(n)} + O(l^N)
\end{equation}
By construction, the $H^{(n)}$ are quadratic in positive boost weight terms. Furthermore, $H^{(0)}$ are the terms calculated from the leading order part of the equation of motion in (\ref{Orderl^0}):
\be \label{Q0}
    H^{(0)} = K_{A B} K^{A B} + \frac{1}{2} c_1(\phi) K_A K^A + \frac{1}{2} (\partial_{v} \phi)^2
\ee

We proceed by induction order-by-order in $l$. Our inductive hypothesis is that we have manipulated the terms in $H$ up to $O(l^m)$ into the form
\begin{multline}\label{Qinduction}
    H= \partial_{v}\left[\frac{1}{\sqrt{\mu}} \partial_{v}\left(\sqrt{\mu} \sum_{n=0}^{m-1} l^n \varsigma^{(n) v} \right)\right] + \left(K_{A B}+\sum_{n=0}^{m-1} l^n X^{(n)}_{A B}\right) \left(K^{A B}+\sum_{n=0}^{m-1} l^n X^{(n) A B}\right) +\\
     \frac{1}{2} c_1(\phi) \left(K_{A}+\sum_{n=0}^{m-1} l^n X^{(n)}_{A}\right) \left(K^{A}+\sum_{n=0}^{m-1} l^n X^{(n) A}\right) + \frac{1}{2}\left( \partial_v \phi + \sum_{n=0}^{m-1} l^n X^{(n)} \right)^2 +\\
     D_{A}{\sum_{n=0}^{m-1} l^n Y^{(n) A}} + \sum_{n=m}^{N-1} l^n H^{(n)} + O(l^{N})
\end{multline}
where the $H^{(n)}$ may have gained extra terms compared to (\ref{Qseries}) but are still quadratic in positive boost weight terms.

By (\ref{Q0}), this is true for $m=1$ with $\varsigma^{(0) v} = X^{(0)}_{A B} = X^{(0)}_{A} = X^{(0)} = Y^{(0) A} = 0$. So assume it is true for some $1\leq m\leq N-1$. 

We now consider $H^{(m)}$. It is quadratic in positive boost weight quantities. However, we have reduced the set of allowed positive boost weight quantities. Therefore we can write it as a sum 
\be
    H^{(m)} = \sum_{k_1, k_2, p_1, p_2, P_1, P_2} (D^{k_1}{\partial_{v}^{p_1} P_1}) \, (D^{k_2}{\partial_{v}^{p_2} P_2}) \, Q_{k_1, k_2, p_1, p_2, P_1, P_2} 
\ee
where $P_1, P_2 \in \{ K_{A B}, K_{A}, \partial_v \phi \}$ and $Q_{k_1, k_2, p_1, p_2, P_1, P_2}$ is some linear combination of allowed terms. Note that we have dropped $A, B,...$ indices here for notational ease, and they can be contracted in any way.

We now move the $D^k_1 \partial_v^{p_1}$ derivatives off the leading positive boost weight factor in each term in the sum. The method of doing so is identical to the HKR procedure detailed in Section 3.5 of \cite{Davies:2023} but with $P_1,P_2$ in the place of factors of $K$, so we shall not repeat it here. It produces extra total derivative terms, with the end result being
\be
    H^{(m)} = \sum_{k, p, P_1, P_2} P_1 \, (D^{k}{\partial_{v}^{p} P_2}) \, Q_{k, p,P_1,P_2} + \partial_{v}\left[\frac{1}{\sqrt{\mu}} \partial_{v}\left(\sqrt{\mu} \varsigma^{(m) v} \right)\right] + D_{A}Y^{(m) A}
\ee
with $\varsigma^{(m) v}$ and $Y^{(m) A}$ quadratic in positive boost weight quantities. It also produces terms that are higher order in $l$. These are still quadratic in positive boost weight quantities so can be absorbed into $\sum_{n=m+1}^{N-1} l^n H^{(n)}$. 

We now split the sum over $P_1 \in \{ K_{A B}, K_{A}, \partial_v \phi \}$, write the remaining sums as $2X^{(m) A B}$, $c_1(\phi) X^{(m) A}$ and $X^{(m)}$, and substitute this into (\ref{Qinduction}):
\begin{multline}\label{Qinductionfinal}
    H= \partial_{v}\left[\frac{1}{\sqrt{\mu}} \partial_{v}\left(\sqrt{\mu} \sum_{n=0}^{m} l^n \varsigma^{(n) v} \right)\right] + \left(K_{A B}+\sum_{n=0}^{m-1} l^n X^{(n)}_{A B}\right) \left(K^{A B}+\sum_{n=0}^{m-1} l^n X^{(n) A B}\right) + 2 l^m K_{A B} X^{(m) A B} + \\
     \frac{1}{2} c_1(\phi) \left(K_{A}+\sum_{n=0}^{m-1} l^n X^{(n)}_{A}\right) \left(K^{A}+\sum_{n=0}^{m-1} l^n X^{(n) A}\right) + l^m c_1(\phi) K_{A} X^{(m) A} + \\
     \frac{1}{2}\left( \partial_v \phi + \sum_{n=0}^{m-1} l^n X^{(n)} \right)^2 + l^m \partial_v \phi X^{(m)} +   D_{A}{\sum_{n=0}^{m} l^n Y^{(n) A}} + \sum_{n=m+1}^{N-1} l^n H^{(n)} + O(l^{N})
\end{multline}
We now complete the three squares to bring $l^m X^{(m) A B}$, $l^m X^{(m) A}$ and $l^m X^{(m)}$ into the sums. The extra terms produced are $O(l^{m+1})$ because $X^{(0)}_{A B} = X^{(0)}_A = X^{(0)} = 0$, and are quadratic in positive boost weight quantities so can be absorbed into $\sum_{n=m+1}^{N-1} l^n H^{(n)}$. This completes the inductive step. 

This can be repeated until all terms up to $O(l^N)$ are of the correct form. Substituting this back into the definition of $H$ in (\ref{Qdef}), we can now write $E_{v v}\big|_{\NN}$ in the desired form (\ref{EvvForm}) with 
\be
    s^v_{HKR} = s^v_{BDK} + \sum_{n=0}^{N-1}l^n \varsigma^{(n) v}
\ee
This completes the generalization of the HKR entropy
\be
    S_{HKR}(v) = 4\pi \int_{C(v)}\text{d}^{d-2}x \sqrt{\mu} s^{v}_{HKR}
\ee
which satisfies $\delta^2 \dot{S}_{HKR} \geq -O(l^N)$.

\subsection{Further Modification of the Entropy}\label{FurtherModifications}

We now further modify this to generalize the entropy defined by Davies and Reall in our companion paper. Performing on $S_{HKR}$ the same steps used to get to (\ref{dotSHKR}) we have

\begin{equation} \label{HKREMS}
        \dot{S}_{HKR}(v_0)= 4 \pi \int_{C(v_0)} \text{d}^{d-2} x \sqrt{\mu(v_0)} \int_{v_0}^{\infty} \text{d}v \left[W^2 + D_{A}{Y^{A}} + O(l^{N}) \right](v)
\end{equation}
where $W^2= \left(K_{A B} + X_{A B}\right) \left(K^{A B} + X^{A B}\right) + \frac{1}{2} c_1(\phi) \left(K_{A}+ X_{A}\right) \left(K^{A}+X^{A}\right) + \frac{1}{2} \left( \partial_v \phi + X\right)^2$. We have suppressed all $x$-dependence, and switched notation to $v_0$ and $v$ to match \cite{DaviesReall:2023}. The obstruction to this integral being non-negative up to $O(l^N)$ is $D_{A}{Y^{A}}(v)$. Despite being a divergence term, it does not integrate to zero because it is evaluated at the integration variable $v$, whereas the area element $\sqrt{\mu(v_0)}$ is evaluated at $v_0$. Define
\be
    a(v_0,v)= \sqrt{\frac{\mu(v)}{\mu(v_0)}}
\ee
which measures the change in the area element from $v_0$ to $v$. 
Then, if we try to integrate $D_{A}{Y^{A}}(v)$ by parts we get\footnote{All cross sections $C(v)$ are diffeomorphic to each other, and thus we write them all as $C$ for this section.}
\be \label{DYbyparts}
\begin{split}
    \int_{C} \text{d}^{d-2} x \sqrt{\mu(v_0)} \int_{v_0}^{\infty} \text{d}v D_{A}{Y^{A}}(v) =& \int_{v_0}^{\infty} \text{d}v \int_{C}  \text{d}^{d-2} x \sqrt{\mu(v)} a^{-1}(v_0,v) D_{A}{Y^{A}}(v)\\
    =& -\int_{v_0}^{\infty} \text{d}v \int_{C}  \text{d}^{d-2} x \sqrt{\mu(v)} Y^{A}(v) D_A a^{-1}(v_0,v)\\
    =& \int_{C} \text{d}^{d-2} x \sqrt{\mu(v_0)} \int_{v_0}^{\infty} \text{d}v Y^{A}(v) D_A \log a(v_0,v)
\end{split}
\ee
Now, $Y^A(v)$ is quadratic in positive boost weight quantities and so is a sum of terms of the form $(D^{k_1}{\partial_{v}^{p_1} P_1}) $ $ (D^{k_2}{\partial_{v}^{p_2} P_2}) $ $ Q(v)$ where, as before, $P_1, P_2 \in \{ K_{A B}, K_{A}, \partial_v \phi \}$ and $Q(v)$ is some linear combination of allowed terms. Therefore this integrand closely resembles the terms we manipulated in the previous section, except with factors of $D_A \log a(v_0,v)$. We will show these terms can still be absorbed into the positive definite terms in (\ref{HKREMS}). We will do this via a similar induction over powers of $l$.

Our inductive hypothesis is that we have manipulated $\dot{S}_{HKR}(v_0)$ up to $O(l^m)$ into the form
\begin{multline} \label{Hypothesis}
    \dot{S}_{HKR}(v_0) = -\frac{d}{dv} \left(4\pi \int_{C(v)} \text{d}^{d-2}x \sqrt{\mu(v)} \sigma_m^{v}(v) \right)\\
    +4\pi \int_{C} \text{d}^{d-2} x \sqrt{\mu(v_0)} \int_{v_0}^{\infty} \text{d}v  \Big[ \left(K^{A B} + Z_{m}^{A B}\right) \left( K_{A B} + Z_{m A B} \right)+ \frac{1}{2} c_1(\phi) \left(K_{A}+ Z_{m A}\right) \left(K^{A}+Z^{A}_m\right)\\
    + \frac{1}{2}\left( \partial_v \phi + Z_m \right)^2 + R_m  + O(l^N) \Big](v_0,v)
\end{multline}
where $Z^{A B}_{m}(v_0,v)$, $Z^{A}_{m}(v_0,v)$ and $Z_{m}(v_0,v)$ are $O(l)$ and at least linear in positive boost weight quantities, and $R_m(v_0,v)$ is of the form
\be
    R_m(v_0,v) = \sum_{n=m}^{N-1} l^n \sum_{k_1, k_2, p_1, p_2, P_1, P_2} (D^{k_1}{\partial_{v}^{p_1} P_1}) \, (D^{k_2}{\partial_{v}^{p_2} P_2}) \, Q_{k_1, k_2, p_1, p_2,P_1,P_2,m, n} (v_0,v)
\ee
and, in particular, $Z_{m}^{A B}(v_0,v)$, $Z_{m}^{A}(v_0,v)$, $Z_{m}(v_0,v)$ and $Q_{k_1, k_2, p_1, p_2,P_1,P_2, m, n}(v_0,v)$ is each a linear combination of terms, where each term is a product of factors of two possible types: (i) allowed terms evaluated at $v$ and (ii) $D^q \log a(v_0,v)$ with $q\ge 1$ ($D_A$ evaluated at time $v$). If a factor of type (ii) is present then the term is bilocal, otherwise it is local. All covariant derivatives $D$ are constructed from $\mu_{A B}(v)$, and all $P_1, P_2$ terms are evaluated at $v$. 

By (\ref{HKREMS}) and (\ref{DYbyparts}), the base case $m=0$ is satisfied with $Z^{A B}_{0} = X^{A B}$, $Z^{A}_{0} = X^{A}$, $Z_{0} = X$ and $R_0= Y^A D_A \log a$. Assuming true for $m$, the obstruction to proceeding is the order $l^{m}$ terms in the sum in $R_m$, which are of the form $l^m (D^{k_1}{\partial_{v}^{p_1} P_1}) $ $ (D^{k_2}{\partial_{v}^{p_2} P_2}) $ $ Q(v_0,v)$. We aim to remove the $D^{k_1}\partial_{v}^{p_1}$ from each term and then complete the square.

We first reduce $k_1$ by 1 in each term via a spatial integration by parts:
\begin{equation} \label{spatialibp}
    \begin{split}
        &\int_{C} \text{d}^{d-2} x \sqrt{\mu(v_0)} \int_{v_0}^{\infty} \text{d}v (D^{k_1}{\partial_{v}^{p_1} P_1}) \, (D^{k_2}{\partial_{v}^{p_2} P_2}) \, Q\\
        = & - \int_{v_0}^{\infty} \text{d}v \int_{C}  \text{d}^{d-2} x \sqrt{\mu(v)} \, (D^{k_1-1}{\partial_{v}^{p_1} P_1}) \, D \left[ a^{-1} (D^{k_2}{\partial_{v}^{p_2} P_2}) Q \right]\\
        = & - \int_{C} \text{d}^{d-2} x \sqrt{\mu(v_0)} \int_{v_0}^{\infty} \text{d}v \, (D^{k_1-1}{\partial_{v}^{p_1} P_1}) \, \Big[ (D^{k_2+1}{\partial_{v}^{p_2} P_2}) \, Q \\
        & + (D^{k_2}{\partial_{v}^{p_2} P_2}) \, DQ - (D^{k_2}{\partial_{v}^{p_2} P_2}) \, Q\, D \log{a} \Big]
    \end{split}
\end{equation}
where in the last step we used $a D(a^{-1}) = -D \log a$. We repeat to bring $k_1$ to 0 in all terms, leaving us with terms of the form $l^m (\partial_{v}^{p_1} P_1) \, (D^{k}{\partial_{v}^{p_2} P_2}) \, Q$, with $Q$ still made exclusively from local allowed terms and factors of $D^q \log a(v_0,v)$.

We now aim to reduce $p_1$ to 0 by $v$-integration by parts. However, to avoid surface terms we must treat local and bilocal terms separately. 

\subsubsection{Bilocal Terms}

Bilocal terms have at least one factor of $D^q \log a(v_0,v)$. Their $v$-integration by parts follows simply:
\begin{equation} \label{vint}
\begin{split}
    \int_{v_0}^{\infty} \text{d}v (\partial_{v}^{p_1} P_1) \, (D^{k}{\partial_{v}^{p_2} P_2}) \, Q D^{q} \log a(v_0,v) =& \left[(\partial_{v}^{p_1-1} P_1) (D^{k}{\partial_{v}^{p_2} P_2}) Q D^{q} \log a(v_0,v) \right]^{\infty}_{v_0}\\
    &- \int_{v_0}^{\infty} \text{d}v (\partial_{v}^{p_1-1} P_1) \, \partial_v \left[(D^{k}{\partial_{v}^{p_2} P_2}) \, Q D^{q} \log a(v_0,v) \right]
\end{split}
\end{equation}
The boundary term vanishes at $v=\infty$ because we assume the black hole settles down to stationarity, and so positive boost weight quantities vanish. The boundary term also vanishes at $v=v_0$ because $a(v_0,v_0) \equiv 1$ and hence $D^q \log a = 0$. 
In the remaining $v$-integral, we can commute the $\partial_v$ past any $D$ derivatives using the formula 
\be
    [\partial_{v}, D_{A}]t_{B_1 ... B_n} = \sum_{i=1}^{n} \mu^{C D} ( D_{D}{ K_{A B_i} } - D_{A}{ K_{D B_i} } - D_{B_i}{ K_{A D} } ) t_{B_1... B_{i-1} C B_{i+1} ... B_n}
\ee
which will produce additional terms proportional to some $D^{k'} K$. Commuting $\partial_v$ past $D^q$ will leave $D^q \partial_v \log a$, which initially looks like a new type of bilocal term, however one can calculate that
\begin{equation} \label{dvloga}
    \partial_v \log a = \mu^{A B} K_{A B}
\end{equation}
and so this term is actually proportional to $D^q K$. Similarly in $\partial_v Q$, any $v$ derivative of $D^{q'} \log a$ can be dealt with by commuting and then using $(\ref{dvloga})$, and any non-allowed terms such as $\partial_v \beta$ or $\partial_{v r} \phi$ can be swapped out to $O(l^N)$ using the equations of motion, which will generate additional terms in $R_m$ of $O(l^{m+1})$. 

Therefore we are left with two types of terms at order $l^m$: i) terms that retain their factor of $D^q \log a$, which will be of the form $(\partial_{v}^{p_1-1} P_1) (D^{k}{\partial_{v}^{p_2} P_2}) Q D^{q} \log a(v_0,v)$ (with possibly changed $k$,$p_2$ and $Q$), and ii) terms that had $D^q \log a$ hit by $\partial_v$, which will be of the form $(D^{k'}{ K }) \,(\partial_{v}^{p_1-1} P_1) Q$ for some $k'$ and $Q$. This second type of term can potentially be local.

The $v$-integration by parts can be repeated on terms of type (i) until $p_1$ is reduced to 0, producing more terms of type (ii) along the way (which will have varying $p_1$'s). To terms of type (ii) we move the $D^{k'}$ derivatives off of $K$ via the same spatial integration by parts as in (\ref{spatialibp}). This brings them proportional to $K$, and hence, after relabelling this $K$ as $P_1$ and the old $P_1$ as $P_2$, they also effectively have $p_1$ reduced to 0. 

\subsubsection{Local Terms}

Local terms are of the form $(\partial_{v}^{p_1} P_1) \, (D^{k}{\partial_{v}^{p_2} P_2}) \, Q(v)$ with $Q(v)$ made exclusively from allowed terms evaluated at $v$. We can no longer simply do a $v$-integration by parts on this because there is no $D^q \log a$ to make the boundary term vanish at $v=v_0$. However, we can manipulate these terms in the same fashion as in the HKR procedure, namely by noting there exist unique numbers $a_j$ such that
\be
    (\partial_{v}^{p_1} P_1) \, (D^{k}{\partial_{v}^{p_2}P_2}) Q = \partial_{v}\left\{\frac{1}{\sqrt{\mu}} \partial_{v}\left[\sqrt{\mu} \sum_{j=1}^{\mathclap{p_1+p_2-1}} a_j (\partial_{v}^{p_1+p_2-1-j} P_1) \, (D^{k}{\partial_{v}^{j-1}P_2 }) Q \right]\right\}\\
    + ...
\ee
where the ellipsis denotes terms of the form
$(\partial_{v}^{\Bar{p}_1} P_1)$ $(D^{\bar{k}}{\partial_{v}^{\Bar{p}_2}P_2}) \tilde{Q}$ with $\Bar{p}_1+\Bar{p}_2 < p_1+p_2$ or $\Bar{p}_1=0$ or $\Bar{p}_2=0$. The proof follows identically to Appendix A.2 of \cite{Davies:2023} but with $P_1,P_2$ in the place of factors of $K$. The new $\tilde{Q}$ are still local, but do include terms like $\partial_{v}{ Q }$ which will involve non-allowed terms. However, these can be swapped out to $O(l^N)$ using the equations of motion, generating more $O(l^{m+1})$ in $R_m$.

We repeat this procedure on the terms in the ellipsis with $\Bar{p}_1+\Bar{p}_2 < p_1+p_2$ until eventually $\Bar{p}_1=0$ or $\Bar{p}_2=0$ for all terms. This must eventually happen because $\Bar{p}_1+\Bar{p}_2$ must decrease by at least 1 if the new $\bar{p}_1\neq 0$ and $\bar{p}_2\neq 0$, and hence $\Bar{p}_1+\Bar{p}_2$ eventually falls below 2, meaning one of $\bar{p}_1$ and $\bar{p}_2$ must be 0. Therefore we can write all local terms as a sum of terms of the form i) $\partial_{v}\left\{\frac{1}{\sqrt{\mu}} \partial_{v}\left[\sqrt{\mu} \sigma^{v}(v) \right] \right\}$ with $\sigma^v$ local and quadratic in positive boost weight quantities, ii) $P_1 \, (D^{\bar{k}}{\partial_{v}^{\bar{p}_2} P_2}) \, \tilde{Q}(v)$, and iii) $(\partial_{v}^{\bar{p}_1} P_1) \, (D^{\bar{k}} P_2 ) \, \tilde{Q}(v)$.

We have successfully reduced $k_1=p_1=0$ to 0 in terms of type (ii). For terms of type (iii), we can relabel $P_1 \leftrightarrow P_2$ and then remove the $D^{\bar{k}}$ derivatives from $P_1$ by using spatial integration by parts, as in (\ref{spatialibp}). This will introduce bilocal factors of $D^q \log a$, but this is fine: the resulting terms will all be of the desired form $P_1 \, (D^{k}{\partial_{v}^{\tilde{p}} P_2}) \, Q(v_0,v)$, i.e. they also have $k_1=p_1=0$.

Let us look at what happens to terms of type (i) when they are placed in the integral:
\be
\begin{split}
    \int_{C} \text{d}^{d-2} x \sqrt{\mu(v_0)} \int_{v_0}^{\infty} \text{d}v \,  \partial_{v}\left\{\frac{1}{\sqrt{\mu}}\partial_{v}\left[\sqrt{\mu} \sigma^{v}(v) \right] \right\} =& - \int_{C} \text{d}^{d-2}x \partial_{v}\left(\sqrt{\mu(v)} \sigma^{v}(v) \right)\Big|_{v=v_0}\\
    =& - \frac{d}{dv} \left(\int_{C} \text{d}^{d-2}x \sqrt{\mu(v)} \sigma^{v}(v) \right)\Big|_{v=v_0}
\end{split}
\ee
where in the first line we set the boundary term at $v=\infty$ to zero because we assume the black hole settles down to stationarity. These are the terms which will modify our definition of the entropy. 

\subsubsection{Completion of the Induction}

To summarise, we have rewritten all order $l^{m}$ terms in $R_m$ as
\be
    - \frac{d}{dv} \left(l^{m} \int_{C} \text{d}^{d-2}x \sqrt{\mu(v)} \sigma^{v}(v) \right)\Big|_{v=v_0}
    + l^{m} \int_{C} \text{d}^{d-2}x \sqrt{\mu(v_0)} \int_{v_0}^{\infty} \text{d}v \sum_{k, p,P_1,P_2} \, P_1 \, (D^{k}{\partial_{v}^{p} P_2}) \, Q_{k, p,P_1,P_2}(v_0,v)
\ee
where $Q_{k, p,P_1,P_2}(v_0,v)$ is a linear combination of allowed terms evaluated at $v$ and factors of $D^q \log a(v_0,v)$ with $q\geq 1$. 

Similarly to Section \ref{ManipTermsObO}, the final step in the induction is to split the sum over $P_1\in\{K_{A B}, K_A, \partial_v \phi\}$ and write the remaining sums as $2 l^{m} \tilde{Z}^{A B}$, $l^m c_1(\phi) \tilde{Z}^{A}$ and $l^m \tilde{Z}$. We then absorb them into the positive definite terms in (\ref{Hypothesis}) by completing the squares and setting $Z^{A B}_{m+1} = Z^{A B}_m + l^m \tilde{Z}^{A B}$, $Z^{A}_{m+1} = Z^{A}_m + l^m \tilde{Z}^{A}$ and $Z_{m+1} = Z_m + l^m \tilde{Z}$. The remainder terms will be $O(l^{m+1})$ (because $Z^{A B}_m$ etc. are $O(l)$), quadratic in positive boost weight (because $Z^{A B}_m$ etc. are linear), and linear combinations of local allowed terms and factors of $D^q \log a$ (because $Z^{A B}_m$ etc. are). Furthermore, $Z^{A B}_{m+1}$, $Z^{A}_{m+1}$ and $Z_{m+1}$ retain these properties. Finally we label $\sigma_{m+1}^v=\sigma_m^v+l^m\sigma^v$. Thus the induction proceeds.

We continue the induction until $m=N$, at which point $R_N$ is $O(l^N)$. Therefore, the entropy defined by
\be
    S(v) := \int_{C} \text{d}^{d-2}x \sqrt{\mu(v)} s^v(v)
\ee
with $s^v = s^{v}_{HKR} + \sigma^{v}_{N}$ satisfies a non-perturbative 2nd Law up to $O(l^N)$. 

\subsection{Gauge (Non-)Invariance of Entropy} \label{gaugeinvariance}

Through the above procedure we have constructed an entropy $S(v)$ that depends on the local geometry of the "constant time" slice $C(v)$ and satisfies a non-perturbative 2nd Law for Einstein-Maxwell-Scalar EFT. Furthermore, its entropy density $s^v$ differs from the BDK entropy density $s_{BDK}^v$ defined in Section \ref{BDKEntropy} by terms that are quadratic in perturbations around a stationary black hole. Thus the facts that the BDK entropy reduces to the Wald entropy in equilibrium and satisfies the 1st Law \cite{Biswas:2022} imply they also hold for $S(v)$. Therefore $S(v)$ satisfies many of the properties we should expect in a definition of the entropy of a black hole.

However, we should ask, is this definition of the entropy gauge invariant? There are two types of gauge in our theory: the choice of electromagnetic gauge, and our choice of coordinates.

By construction, $s^v$ only depends on Maxwell quantities through $F_{\mu \nu}$, which is invariant under a change of electromagnetic gauge. Therefore the entropy $S(v)$ is independent of electromagnetic gauge.

As for coordinate independence, our procedure was performed in affinely parameterized GNCs with $r=v=0$ on a given spacelike cross section $C$ of $\NN$ (the GNCs can be defined starting from any horizon cross-section, so the restriction $r=v=0$ is not restricting the choice of cross-section considered). However, as discussed in Section \ref{APGNCs}, such affinely parameterized GNCs are not unique: we can reparameterize the affine parameter on each horizon generator by $v' = v/a(x^A)$. This will produce a new foliation $C'(v')$ of the horizon. We should not expect $S'(v') = S(v)$ for all $v$, because $S'(v')$ and $S(v)$ measure the entropy of the different surfaces $C'(v')$ and $C(v)$. However, we should hope that $S'(0) = S(0)$ because $C'(0) = C(0) = C$. Therefore we should investigate how our entropy density $s^v$ transforms under such a gauge transformation at $r=v=0$.

By construction, $s^v$ can be split into two parts: $s^v_{BDK}$ and the modification terms that are quadratic or higher order in positive boost weight terms. A proof that $s^v_{BDK}$ is gauge invariant on $C$ is beyond the scope of this paper, and we will just assume it holds here. Why should we expect it to be gauge invariant? Well, it is the generalization of the Iyer-Wald-Wall entropy density from Einstein-Scalar EFT to Einstein-Maxwell-Scalar EFT. It is proved in \cite{Hollands:2022} that the Iyer-Wald-Wall entropy density is gauge invariant on $C$ to linear order, and can be made gauge invariant non-perturbatively by adjusting the non-unique higher order terms. We expect the proof can be extended to the Einstein-Maxwell-Scalar EFT case. However, to delve into the covariant phase space formalism of the proof would divert somewhat from the material here.

Thus, we will solely concern ourselves with the quadratic or higher order modification terms. The gauge invariance of these terms for the HKR entropy in the Einstein-Scalar case was discussed in Section 4 of \cite{Davies:2023}, which found they are gauge invariant on $C$ up to and including order $l^4$. This was done by noting that, by the HKR construction, $\varsigma^{(n) v}$ consists of terms with $n$ derivatives that are of the form $\partial_{v}^{p_1} P_1 \, (D^{k}{\partial_{v}^{p_2} P_2}) \, Q_{n, k, p_1, p_2}$ with $P_1, P_2 \in \{K_{A B}, \partial_v \phi \}$. Using that the overall boost weight is 0, we can classify the allowed terms that can appear up to 4 total derivatives. The result is that only $K_{A B}, \bar{K}_{A B}, \partial_r \bar{K}_{A B}, \partial_v \phi, \partial_r \phi, \partial_r^2 \phi, \mu^{A B}$ and $\epsilon^{A_1 ... A_{d-2}}$ can appear, all of which are gauge invariant on $r=v=0$ using the transformation rules given in Section 2.1 of \cite{Hollands:2022}. 

The same analysis follows in the Einstein-Maxwell-Scalar EFT here, with the differences being $P_1, P_2 \in \{K_{A B}, K_A, \partial_v \phi \}$ and $Q_{n, k, p_1, p_2}$ can additionally consist of allowed Maxwell terms. The result is that $K_A, \bar{K}_{A}$ and $\partial_r \bar{K}_{A}$ can appear up to and including order $l^4$, all of which are still gauge invariant on $r=v=0$. Therefore, $s^v$ is gauge invariant to the same order as in the Einstein-Scalar EFT case. As in that case, there are non-gauge invariant terms like $\beta_A, D_A \partial_v \phi$ and $D_A K_B$ that can appear at higher orders in $l$.

\subsection{Discussion of The 2nd Law for a Charged Scalar Field}

We can ask, can we generalize our proof of the 2nd Law to the EFT of gravity, electromagnetism and a \textit{charged} scalar field as defined in Section \ref{ChargedScalar}? Our starting point in the above was the BDK entropy defined in Section \ref{BDKEntropy}, which satisfies a linearized 2nd Law. However, such an entropy is only defined for a real uncharged scalar, and its generalization to a charged scalar does not exist in the literature. Proving such a generalization exists is beyond the scope of this paper as it would involve delving into phase space formalism, and therefore this section is merely a discussion. However, it seems reasonable that such a generalization would exist, in which case the following completes the generalization of the proof of the 2nd Law.

In the analysis of the real scalar field EFT, we could use positive boost weight quantities as a proxy for order of perturbation around a stationary black hole because in Section \ref{PBWQH} we proved all such quantities vanish on the horizon in equilibrium. However, in the charged scalar case things are more subtle because whilst e.g. $\partial_{v}{\phi}$ may vanish in one electromagnetic gauge, it does not in another. 

In our proof of the Generalized 0th Law for a charged scalar in Section \ref{ChargedScalar}, we were able to prove \textit{in a particular choice of gauge} that all positive boost weight quantities vanish on the horizon in equilibrium. However, that gauge was defined by the Killing vector symmetry which is no longer present in the dynamical setting of the 2nd Law, so we cannot use it directly. What we can infer however, is that positive boost weight 
quantities made from gauge invariant quantities like $F_{\mu \nu}$ vanish on the horizon in equilibrium \textit{in all gauges}. Similarly, positive boost weight components of the gauged derivatives
\be \label{gaugephiderivs}
    (\partial_{\mu_1}-i\lambda A_{\mu_1})...(\partial_{\mu_n}-i\lambda A_{\mu_n})\phi
\ee
vanish because if they vanish in one gauge then they vanish in all gauges.

We can apply these facts to a choice of gauge particularly suited to our affinely parameterized GNCs. By a suitable gauge transformation, we can always achieve \cite{Hollands:2022}
\be
    A= r \eta \text{d}v + A_A \text{d}x^A
\ee
for some function $\eta(r,v,x^A)$ regular on the horizon. $\eta$ and $A_A$ have boost weight 0. In this gauge 
\be
    \partial_{r}^p \eta \big|_{\NN} = \partial_r^p F_{r v}\big|_{\NN}, \quad \partial_r^q A_A\big|_{\NN} = \partial_{r}^{q-1}F_{r A}\big|_{\NN}, \quad \partial_v A_A\big|_{\NN} = F_{v A}\big|_{\NN}
\ee
for $p\geq 0, q\geq 1$, and hence all positive boost weight derivatives of $\eta$ and $A_A$ can be written as positive boost weight derivatives of $F_{\mu \nu}$ on the horizon. Similarly
\be
    \partial^p_v \partial^q_r \phi\big|_{\NN} = (\partial_{v}-i\lambda A_{v})^p (\partial_{r}-i\lambda A_{r})^q\phi\big|_{\NN}
\ee
and hence all positive boost weight derivatives of $\phi$ can be written as positive boost weight components of (\ref{gaugephiderivs}) (or their $\partial_A$ derivatives) on the horizon. Therefore, in this gauge all positive boost weight quantities still vanish on the horizon in equilibrium and hence can still be used as a proxy for perturbations around a stationary black hole. 

In this gauge, the leading order 2-derivative part of $E_{v v}|_{\NN}$ can be written as
\begin{equation}
    -E^{(0)}_{v v}\big|_{\NN}= \partial_{v}\left[\frac{1}{\sqrt{\mu}} \partial_{v}\left(\sqrt{\mu}\right) \right] + K_{A B} K^{A B} + \frac{1}{2} c_1(|\phi|^2) h^{A B} \partial_v A_A \partial_v A_B + |\partial_{v} \phi|^2 
\end{equation}

For the higher derivative terms, let us now assume that we can generalize the BDK entropy to the charged scalar case. I.e. we assume we can write 
\begin{equation} \label{BDKCharged}
    -E_{v v}\Big|_{\mathcal{N}} = \partial_{v}\left[\frac{1}{\sqrt{\mu}} \partial_{v}\left(\sqrt{\mu} s^{v}_{BDK}\right) + D_{A}{ s^A }\right] + ...
\end{equation}
for some real entropy current $(s^v_{BDK}, s^A)$ and where the ellipsis denotes terms that are quadratic in positive boost weight quantities. 

We can now generalize the HKR procedure as follows. We first reduce, up to $O(l^N)$, to a set of "allowed terms" given by
\begin{empheq}[box=\fbox]{align}\label{AllowedTermsCharged}
    \text{Allowed terms:} \,\, &\mu_{A B}, \,\, \mu^{A B}, \,\, \epsilon_{A_1 ... A_{d-2}}, \,\, D^k R_{A B C D}[\mu], \,\, D^k \beta_{A}, \,\, 
    D^k\partial_{v}^p K_{A B}, \,\, D^k\partial_{r}^q \bar{K}_{A B},\nonumber\\
    & D^k \eta, \,\, D^k\partial_{v}^p A_{A}, \,\, D^k\partial_{r}^q A_{A}, \,\,
    D^k\partial_{v}^p \phi, \,\, D^k\partial_{r}^q \phi, \,\, D^k\partial_{v}^p \phi^*, \,\, D^k\partial_{r}^q \phi^* 
\end{empheq}
The reduction of the metric terms follows straightforwardly in the same way as vacuum gravity by using the $E^{(0)}_{\mu\nu}=O(l)$ equations of motion. We can eliminate mixed $v$ and $r$ derivatives of $\phi$ by using $E^{(0)}=O(l)$ and evaluating $E^{(0)}$ in affinely parameterized GNCs in this gauge:
\begin{multline}
    E^{(0)}=2\partial_{r}\partial_v{\phi}+ D^{A}{D_{A}{\phi}} +K^{A}_{A} \partial_{r}{\phi}+\beta^A D_{A}{\phi} +\bar{K}^{A}_{A} \partial_{v}{\phi}\\
    -2i \lambda A^{A} D_{A}{\phi} - i \lambda \phi  D^{A}{A_{A}} - i \lambda \eta \phi - i \lambda \beta^{A} A_{A}  \phi - \lambda^2 A_{A} A^{A} \phi+...
\end{multline}
where the ellipsis denotes terms that vanish on the horizon.
The reduction of the Maxwell terms is achieved by using the equation of motion
\be
    E^{(0)}_\mu = c_1(|\phi|^2) \nabla^{\nu} F_{\mu \nu} + F_{\mu \nu} \nabla^{\nu}[c_1(|\phi|^2)] - 4 \epsilon_{\mu\nu\alpha\beta}F^{\alpha\beta} \nabla^{\nu}[c_2(|\phi|^2)] + i\lambda \left[ \phi^* \fD_\mu \phi-\phi(\fD_\mu \phi)^* \right] = O(l) 
\ee
and by substituting our choice of gauge $\psi=\eta+r\partial_r \eta$, $K_A=\partial_v A_A-r D_A \eta$, $\bar{K}_A = \partial_r A_A$ into our affinely parameterized GNC expressions for $\nabla^{\nu} F_{\mu\nu}$ given in Appendix \ref{dFAppendix}. These allow us to eliminate $v$ and $r$ derivatives of $\eta$, and mixed $v$ and $r$ derivatives of $A_A$.

In particular, the only positive boost weight allowed terms are $D^k\partial_{v}^p K_{A B}$ with $p\geq 0$, and $D^k\partial_{v}^p A_{A}$, $D^k\partial_{v}^p \phi$ and $D^k\partial_{v}^p \phi^*$ with $p\geq 1$. Therefore we can rewrite all the terms in the ellipsis in (\ref{BDKCharged}), which we label $H$, up to $O(l^N)$ as a sum of terms of the form 
\be
    (D^{k_1}{\partial_{v}^{p_1} P_1}) \, (D^{k_2}{\partial_{v}^{p_2} P_2}) \, Q
\ee
with $P_1,P_2\in \{K_{A B}, \partial_{v} A_{A}, \partial_{v} \phi, \partial_{v} \phi^*  \}$, and where $Q$ is made from allowed terms. 

Now, since $(s^v_{BDK}, s^A)$ is real, the overall sum of these terms, $H$, is real. Hence we can pair each of these terms up with its complex conjugate (or itself if it is real) and write $H$ as
 \be
    H = \sum_{k_1, k_2, p_1, p_2, P_1, P_2}\left[ (D^{k_1}{\partial_{v}^{p_1} P_1}) \, (D^{k_2}{\partial_{v}^{p_2} P_2}) \, Q_{k_1, k_2, p_1, p_2, P_1, P_2} + (D^{k_1}{\partial_{v}^{p_1} P_1^*}) \, (D^{k_2}{\partial_{v}^{p_2} P_2^*}) \, Q^*_{k_1, k_2, p_1, p_2, P_1, P_2} \right]
\ee
with $P_1,P_2\in \{K_{A B}, \partial_{v} A_{A}, \partial_{v} \phi\}$.

We now generalise our inductive hypothesis (\ref{Qinduction}) to 
\begin{multline}\label{Qinductioncharged}
    H= \partial_{v}\left[\frac{1}{\sqrt{\mu}} \partial_{v}\left(\sqrt{\mu} \sum_{n=0}^{m-1} l^n \varsigma^{(n) v} \right)\right] + \left(K_{A B}+\sum_{n=0}^{m-1} l^n X^{(n)}_{A B}\right) \left(K^{A B}+\sum_{n=0}^{m-1} l^n X^{(n) A B}\right) +\\
     \frac{1}{2} c_1(|\phi|^2) \left(\partial_v A_{A}+\sum_{n=0}^{m-1} l^n X^{(n)}_{A}\right) \left(\partial_v A^{A}+\sum_{n=0}^{m-1} l^n X^{(n) A}\right) + \left( \partial_v \phi + \sum_{n=0}^{m-1} l^n X^{(n)} \right)\left( \partial_v \phi + \sum_{n=0}^{m-1} l^n X^{(n)} \right)^* +\\
     D_{A}{\sum_{n=0}^{m-1} l^n Y^{(n) A}} + \sum_{n=m}^{N-1} l^n H^{(n)} + O(l^{N})
\end{multline}
where the $H^{(n)}$, $X^{(n)}_{A B}$ etc., are real. To proceed the induction we manipulate the terms in $H^m$ exactly as in Section \ref{ManipTermsObO}, except we always keep complex conjugates paired up and perform identical operations on them. This will ensure that when we get to the equivalent of (\ref{Qinductionfinal}) we can split the sum over $P_1\in \{K_{A B}, \partial_{v} A_{A}, \partial_{v} \phi\}$ and get
\begin{multline}\label{Qinductionchargedfinal}
    H= \partial_{v}\left[\frac{1}{\sqrt{\mu}} \partial_{v}\left(\sqrt{\mu} \sum_{n=0}^{m-1} l^n \varsigma^{(n) v} \right)\right] + \left(K_{A B}+\sum_{n=0}^{m-1} l^n X^{(n)}_{A B}\right) \left(K^{A B}+\sum_{n=0}^{m-1} l^n X^{(n) A B}\right) + 2l^m K_{A B} X^{(m) A B}\\
     +\frac{1}{2} c_1(|\phi|^2) \left(\partial_v A_{A}+\sum_{n=0}^{m-1} l^n X^{(n)}_{A}\right) \left(\partial_v A^{A}+\sum_{n=0}^{m-1} l^n X^{(n) A}\right) +l^m c_1(|\phi|^2) \partial_v A_A X^{(m) A}\\
     +\left( \partial_v \phi + \sum_{n=0}^{m-1} l^n X^{(n)} \right)\left( \partial_v \phi + \sum_{n=0}^{m-1} l^n X^{(n)} \right)^* +l^m \partial_v\phi X^{(m)*}+l^m\partial_v \phi^*X^{(m)}\\
     +D_{A}{\sum_{n=0}^{m} l^n Y^{(n) A}} + \sum_{n=m+1}^{N-1} l^n H^{(n)} + O(l^{N})
\end{multline}
and thus we can still absorb the order $l^m$ terms into the positive definite terms by completing the squares. The remainder terms are real, and thus the induction can proceed.

Generalizing the further modifications of Section \ref{FurtherModifications} would follow similarly.

Thus we can get a non-perturbative 2nd Law for a charged scalar field if we assume a BDK entropy exists in such a scenario. The procedure outlined here does not produce an entropy that is manifestly electromagnetic gauge-independent like in the real scalar field case. However, it seems reasonable this could be achieved if the hypothesized BDK entropy was gauge invariant. One could take a more careful approach to the gauge field, for example by keeping derivatives of $\phi$ in terms of gauged derivatives $(D_A -i\lambda A_A)$, $(\partial_v-i\lambda A_v)$ etc.

\section{Discussion}

This paper adds another brick in the wall of proving the Laws of Black Hole mechanics for higher derivative theories of gravity. To summarise where it leaves us, we have a 0th Law, 1st Law and 2nd Law for the EFT regime of higher derivative theories of gravity, electromagnetism and a real scalar field. The dynamical black hole entropy which is constructed along the way is independent of electromagnetic gauge for theories with any number of derivatives, and is purely geometric for theories with up to 6 derivatives (order $l^4$). It reduces to the standard factor of the area in 2-derivative GR, and reduces to the Wald entropy in equilibrium for any number of derivatives. In addition, we have shown the 0th Law continues to hold if the scalar is charged, and there is strong motivation to think the 2nd Law would hold. This suggests a more general result involving theories of gravity with any matter fields that satisfy the NEC at 2-derivative level may be provable. For example, it would be interesting to extend the result to Yang-Mills fields.

Our proofs of the 0th and 2nd Laws are perhaps not as general as we would like them to be. For the 0th Law we required our solution to be analytic in $l$ and excluded certain horizon topologies. For the 2nd Law we required our horizon to be smooth. Recent work \cite{Gadioux:2023} has considered the case of non-smooth horizons and suggested there may be additional contributions to black hole entropy motivated by quantum entanglement entropy. They also demonstrate that certain terms in the entropy current defined above can diverge when integrated over non-smooth features on the horizon. Furthermore, our definition of the entropy is dependent on our choice of GNCs above order $l^4$, which raises question about the uniqueness of black hole entropy. Therefore there is still work to be done in this area.

\section{Acknowledgements}
I thank my supervisor H.S. Reall for many invaluable comments and suggestions on this paper. I would also like to thank J. Santos and S. Bhattacharyya for helpful discussions, particularly around the 0th Law for the charged scalar field. I am supported by an STFC studentship.

\section{Appendix}

\subsection{Evaluation of $E^{(0)}_\tau$ on the Horizon}\label{EtauEvaluation}

We would like to evaluate

\be
    E^{(0)}_\tau[\Phi_J] = g^{\alpha \beta} \nabla_{\alpha}{\Big[c_1(\phi)F_{\tau\beta} - 4 c_2(\phi) F^{\gamma \delta} \epsilon_{\tau \beta \gamma \delta}\Big]}
\ee
in Killing vector GNCs on the horizon. The metric in Killing vector GNCs is given by
\begin{equation}
    g = 2 \text{d}\tau \text{d}\rho - \rho X(\rho,x^C) \text{d}\tau^2 +2 \rho \omega_{A}(\rho,x^C) \text{d}\tau \text{d}x^A + h_{A B}(\rho,x^C) \text{d}x^A \text{d}x^B
\end{equation}
On the horizon it is simply
\begin{equation}
    g\big|_{\rho=0} = 2 \text{d}\tau \text{d}\rho + h_{A B}\text{d}x^A \text{d}x^B
\end{equation}
We can calculate the Christoffel symbols on the horizon in this metric. The non-zero components are
\be
\begin{split}
    \Gamma^{\tau}_{\tau \tau} = \frac{1}{2} X, \quad\quad \Gamma^{\tau}_{\tau A} = -\frac{1}{2} \omega_A,  \quad&\quad \Gamma^{\tau}_{A B} = -\frac{1}{2} \partial_\rho h_{A B}, \quad\quad \Gamma^{\rho}_{\rho \tau} = - \frac{1}{2} X,\\
    \Gamma^{\rho}_{\rho A} = \frac{1}{2} \omega_A, \quad\quad \Gamma^{A}_{\rho\tau} = \frac{1}{2} \omega_B h^{A B} \quad&\quad \Gamma^{A}_{\rho B} = \frac{1}{2} \partial_\rho h_{B C} h^{A C} \quad\quad \Gamma^{A}_{B C} = \Gamma^{A}_{B C}[h]
\end{split}
\ee
where $\Gamma^{A}_{B C}[h]$ is the Christoffel symbol built out of the induced metric $h_{A B}$.

Now, for notational convenience, let
\be
    H_{\alpha \beta} = c_1(\phi)F_{\alpha\beta} - 4 c_2(\phi) F^{\gamma \delta} \epsilon_{\alpha \beta \gamma \delta}
\ee
Note this is antisymmetric. Then we can evaluate
\be
\begin{split}
     E^{(0)}_\tau[\Phi_J]\Big|_{\rho=0} =& g^{\alpha \beta} \nabla_{\alpha} H_{\tau \beta}\\
     =& \nabla_{\tau} H_{\tau \rho} + h^{A B} \nabla_{A} H_{\tau B}\\
     =& - \Gamma^{\mu}_{\tau \tau} H_{\mu \rho} - \Gamma^{\mu}_{\tau \rho} H_{\tau \mu} + h^{A B} \left( \partial_{A} H_{\tau B} - \Gamma^{\mu}_{A \tau} H_{\mu B}- \Gamma^{\mu}_{A B} H_{\tau \mu} \right)
\end{split}
\ee
where in the last line we used the fact that everything is independent of $\tau$. Substituting in the Christoffel symbols computed above, we get some cancellations with the result being
\be
\begin{split}
     E^{(0)}_\tau[\Phi_J]\Big|_{\rho=0} =& h^{A B} \left( \partial_{A} H_{\tau B} - \Gamma^{C}_{A B} H_{\tau C} \right)\\
     =& h^{A B} \DD_{A} H_{\tau B}\\
\end{split}
\ee
where $\DD_A$ is the covariant derivative with respect to $h_{A B}$ and only acts on $A, B, ...$ indices. Finally,
\be
\begin{split}
    H_{\tau A}\Big|_{\rho=0} = & c_1(\phi)F_{\tau A} - 4 c_2(\phi) F^{\gamma \delta} \epsilon_{\tau A \gamma \delta}\\
    = & c_1(\phi)F_{\tau A} - 8 c_2(\phi) \epsilon_{A}^{\,\,\,\,B}F_{\tau B}
\end{split}
\ee
where we used our convention $\epsilon_{A B} = \epsilon_{\rho \tau A B}$. Therefore
\be
    E^{(0)}_\tau[\Phi_J]\Big|_{\rho=0} = h^{A B} \DD_{A}\Big[c_1(\phi)F_{\tau B} - 8 c_2(\phi) \epsilon_{B}^{\,\,\,\,C}F_{\tau C}\Big]
\ee 

\subsection{0th Law for a Charged Scalar Field in the Case $\phi^{(0)}\big|_{\rho=0} \equiv 0$} \label{0thLawChargedProof}

We now deal with the case excluded in Section \ref{ChargedScalar}, in which $\phi^{(0)}$ vanishes identically on the horizon. Moreover, we assume $\phi$ vanishes on the horizon up to and including order $l^{m}$ for some $m\geq0$, i.e. $\phi^{[m]}\big|_{\rho=0}\equiv 0$.

The proof in Section \ref{ChargedScalar} breaks down for the following reasons: in the base case of our induction we can no longer extract $A^{(0)}_{\tau}\big|_{\rho=0} = 0$ from the equation $A^{(0)}_{\tau}\phi^{(0)}\big|_{\rho=0} = 0$, and similarly in the inductive step, we can no longer prove $A_{\tau}^{(k)}\big|_{\rho=0} = 0$. These were essential steps for the induction to proceed because they were used to prove the positive boost weight quantities $\partial_{A_1}...\partial_{A_n}\partial^q_v A_v$ vanish on the horizon at each order (see equation (\ref{dvAv})). Without this fact, we can't ignore the higher derivative parts of the equations of motion $E^{[k]}_{\tau A}[\Phi^{[k-1]}_J]$ and $E^{[k]}_{\tau}[\Phi^{[k-1]}_J]$ at each order in $l$. 

The solution comes from the fact that $E_{\mu \nu}$ and $E_{\mu}$ are electromagnetic gauge invariant. This means any appearance of $A_{\mu}$ is either inside a $F_{\mu \nu}$, or arises from a term of schematic form $(\fD^p \phi)^* \fD^q \phi$, in which case it will appear in the combination $\partial^a \phi^* \partial^b A_\mu \partial^c \phi$. In the former case, we can use the methods from Section \ref{PBWQH} to show it vanishes at each order in the induction. In the latter case, we will show that the vanishing of $\phi^{[m]}$ on the horizon implies positive boost weight quantities involving $\partial^a \phi^* \partial^b A_\mu \partial^c \phi$ also vanish to sufficiently high order for the induction to proceed.

From $E^{(0)}_{\tau \tau}[\Phi^{(0)}_J]|_{\rho=0}=0$ and $E^{(0)}_{\tau A}[\Phi^{(0)}_J]|_{\rho=0}=0$ we can still deduce $F^{(0)}_{\tau A}|_{\rho=0} = 0$ and $\partial_{A} X^{(0)}|_{\rho=0}=0$ respectively, and so we still have the Generalized 0th Law holding at order $l^0$. This means $\kappa^{(0)}$ is constant. We split the analysis into two cases: 1) $\kappa^{(0)}\neq 0$, and 2) $\kappa^{(0)}=0$.

\subsubsection{Case 1: $\kappa^{(0)}\neq 0$}

To proceed, we prove a lemma:

\begin{lemma} \label{philemma}
    If $\phi$ vanishes on the horizon up to and including order $l^{m}$, and $\kappa^{(0)}\neq 0$, then all derivatives of $\phi$ vanish on the horizon up to and including order $l^m$. 
\end{lemma} 

\begin{proof}
    Clearly all tangential derivatives $\partial_{\tau}^{p} \partial^q_{A_1 ... A_q} \phi$ vanish on the horizon up to and including order $l^m$. To investigate the remaining $\rho$ derivatives, we will inspect the scalar field equation of motion, $E[\Phi_J] = 0$. 
    At order $l^0$, this is $E^{(0)}[\Phi^{(0)}_J]=0$. Equation (\ref{0th order EoM 2 charged}) has the explicit form for $E^{(0)}[\Phi_J]$. We can evaluate it in Killing vector GNCs in our choice of electromagnetic gauge and find
    \be \label{E0Phi}
        E^{(0)}[\Phi_J] = (X-2 i \lambda A_{\tau} ) \partial_{\rho}\phi + \rho X \partial_{\rho}^2 \phi + \rho^2 h^{A B} \omega_{A} \omega_{B} \partial_{\rho}^2 \phi + ...
    \ee
    where the ellipsis denotes terms that are proportional to $\phi$ or its spatial derivatives $\partial^q_{A_1 ... A_q} \phi$. Therefore, plugging this and $\phi^{(0)}\big|_{\rho=0} \equiv 0$ into $E^{(0)}[\Phi^{(0)}_J]\big|_{\rho=0}=0$ gives
    \be \label{drhophi}
        (X^{(0)}-2 i \lambda A^{(0)}_{\tau} ) \partial_{\rho}\phi^{(0)}\big|_{\rho=0} = 0
    \ee
    We have that $X^{(0)}|_{\rho = 0} = 2 \kappa^{(0)}$, which is a constant we have assumed is non-zero. Therefore, $\partial_{\rho}\phi^{(0)}\big|_{\rho=0}\equiv0$. Inductively assuming $\partial^k_{\rho} \phi^{(0)}\big|_{\rho=0}\equiv 0$ for all $k\leq s$ for some $s\geq 1$, we substitute (\ref{E0Phi}) into $\partial^{s}_{\rho} E^{(0)}[\Phi^{(0)}_J]\big|_{\rho=0}=0$ to get
    \be \label{drhos1phi}
        \left[(s+1) X^{(0)}-2 i \lambda A^{(0)}_{\tau} \right] \partial^{s+1}_{\rho}\phi^{(0)}\big|_{\rho=0} = 0
    \ee
    and so $\partial^{s+1}_{\rho} \phi^{(0)}\big|_{\rho=0}\equiv 0$. Hence all derivatives of $\phi^{(0)}$ vanish on $\NN$, which proves the case $m=0$. This means any appearance of $\phi$ is at least order $l$ on the horizon. 
    
    Now assume $m\geq 1$. Inductively, let us assume all derivatives of $\phi$ vanish on the horizon up to and including order $l^n$ for some $n$, $0\leq n < m$. Therefore any appearance of $\phi$ or its derivatives in $E_{I}[\Phi_J]\big|_{\rho=0}$ is at least order $l^{n+1}$.

    The Lagrangian $\LL$ is electromagnetic gauge invariant, therefore wherever a $\phi$ or its derivatives appears, it must be multiplied by $\phi^{*}$ or its derivatives, and vice versa. This means that every term in $E[\Phi_J]$, which is the equation of motion arising from varying $\phi^*$, must be at least linear in $\phi$ or its derivatives. 

    Therefore, $E[\Phi_J]\big|_{\rho=0}$ is already $O(l^{n+1})$, and its order $l^{n+1}$ part is $E^{(0)}[g^{(0)}_{\mu \nu}, A^{(0)}_{\mu}, l^{n+1}\phi^{(n+1)}]\big|_{\rho=0}$. This is 0 by the equations of motion, and using (\ref{E0Phi}) and $\phi^{(n+1)}\big|_{\rho=0} \equiv 0$ similarly to above, we get
    \be \label{drhophin1}
        (X^{(0)}-2 i \lambda A^{(0)}_{\tau} ) \partial_{\rho}\phi^{(n+1)}\big|_{\rho=0} = 0
    \ee
    This implies $\partial_{\rho}\phi^{(n+1)}\big|_{\rho=0}\equiv0$, and similarly we can successively look at $\partial^{s}_{\rho} E^{(0)}[g^{(0)}_{\mu \nu}, A^{(0)}_{\mu}, l^{n+1}\phi^{(n+1)}]\big|_{\rho=0}=0$ to deduce $\partial^{s+1}_{\rho} \phi^{(n+1)}\big|_{\rho=0}\equiv 0$. Therefore all derivatives of $\phi$ vanish on $\NN$ up to and including $l^{n+1}$ and so the induction proceeds.
\end{proof}

This lemma implies that if $\phi$ vanishes to \textit{all} orders on the horizon, then all its derivatives vanish on the horizon. Therefore a term of the form $\partial^a \phi^* \partial^b A_\mu \partial^c \phi$ would identically vanish on the horizon, and so $A_\mu$ could only appear inside an $F_{\mu\nu}$. But in this case there would be no new positive boost weight quantities to deal with over the real scalar field case, and the equations of motion would look identical. Hence, the Generalized 0th Law would follow trivially from the real scalar field proof above.

Therefore let us assume that $\phi^{(m+1)}$ is the lowest order at which $\phi$ does not identically vanish on the horizon, i.e. that $\phi^{[m]}|_{\rho=0}\equiv 0$ and $\phi^{(m+1)}(x^A)|_{\rho=0}$ is non-zero for some $x^A$. Then by the above lemma, all derivatives of $\phi$ vanish on the horizon up to and including order $l^{m}$, and so any appearance of $\phi$ is $O(l^{m+1})$ on the horizon. But this means the problematic terms of the form $\partial^a \phi^* \partial^b A_\mu \partial^c \phi$ are already at least order $l^{2m+2}$, and so won't appear in our induction until that order!

To make this precise, we take our inductive hypothesis to be $\partial_A X^{[k-1]}\big|_{\rho=0} = 0$, $F^{[k-1]}_{\tau A}\big|_{\rho=0} = 0$ and $A^{[k-2m-3]}_{\tau}\big|_{\rho=0}=0$. This is true for the base case $k=1$ because we proved above that $\partial_A X^{(0)}\big|_{\rho=0} = 0$ and $F^{(0)}_{\tau A}\big|_{\rho=0} = 0$, and trivially $A^{(n)}_{\tau}\big|_{\rho=0}=0$ for $n<0$ by analyticity in $l$.

Assuming the hypothesis holds, we would like to study $E_{\tau A}|_{\rho=0}$ and $E_{\tau}|_{\rho=0}$ at order $l^{k}$. By gauge invariance, any appearance of $A_{\mu}$ not inside an $F_{\mu\nu}$ will come multiplied by $\partial^a \phi^* \partial^b \phi$ and so can only involve $A^{[k-2m-2]}_{\mu}$. Therefore, separating out the dependence on $A_\mu$ and $F_{\mu \nu}$, we have 
\be
    \text{At order} \,\, l^{k}, \quad E_{I}[\Phi_J]\big|_{\rho=0} = E^{[k]}_{I}[g_{\mu\nu}^{[k]}, F_{\mu \nu}^{[k]}, \phi^{[k]}, A^{[k-2m-2]}_\mu]|_{\rho=0}
\ee
for $I=(\tau A)$ or $I=\tau$. Additionally, the highest order pieces $g_{\mu\nu}^{(k)}, F_{\mu \nu}^{(k)}, \phi^{(k)}, A^{(k-2m-2)}_\mu$ can only appear in $E^{(0)}_I$ because they will already come with $l^k$:
\begin{multline}
    \text{At order} \,\, l^{k}, \quad E_{I}[\Phi_J]\big|_{\rho=0} = E^{(0)}_{I}[g_{\mu\nu}^{[k]}, F_{\mu \nu}^{[k]}, \phi^{[k]}, A^{[k-2m-2]}_\mu]|_{\rho=0} \\
    +\sum_{s=1}^k l^s E^{(s)}_{I}[g_{\mu\nu}^{[k-1]}, F_{\mu \nu}^{[k-1]}, \phi^{[k-1]}, A^{[k-2m-3]}_\mu]|_{\rho=0}
\end{multline}

The first two inductive hypotheses $\partial_A X^{[k-1]}\big|_{\rho=0} = 0$ and $F^{[k-1]}_{\tau A}\big|_{\rho=0} = 0$ imply positive boost weight quantities involving $g^{[k-1]}_{\mu \nu}, F^{[k-1]}_{\mu \nu}$ and $\phi^{[k-1]}$ vanish on the horizon by Sections \ref{PBWQH} and \ref{CompleteInduct}. Furthermore, as discussed around equations (\ref{AKilling}-\ref{dvAv}), combining them with the third hypothesis $A^{[k-2m-3]}_{\tau}\big|_{\rho=0}=0$ will imply all positive boost weight quantities involving $A^{[k-2m-3]}_\mu$ vanish on the horizon.

Therefore, for the components $I=(\tau A)$ and $I=\tau$ we see that the higher derivative parts still vanish on the horizon because they are proportional to positive boost weight components when we make the co-ordinate transformation $\rho = r(\kappa^{[k-1]}v+1)$, $\tau = \frac{1}{\kappa^{[k-1]}}\log (\kappa^{[k-1]}v+1)$, and only involve the fields $g_{\mu\nu}^{[k-1]}, F_{\mu \nu}^{[k-1]}, \phi^{[k-1]}, A^{[k-2m-3]}_\mu$. Thus we need only look at $E^{(0)}_{I}\big|_{\rho=0}$ for these components.

First up, $E_{\tau}[\Phi_J]\big|_{\rho=0}$, is
\begin{multline} \label{Etauorderk}
    \text{At order} \,\, l^{k}, \quad E_{\tau}[\Phi_J]\Big|_{\rho=0} = 2\lambda^2 l^k |\phi^{(m+1)}|^2 A^{(k-2m-2)}_\tau\\
    + l^k h^{(0) A B} \DD^{(0)}_A \left[ c_1(0) F^{(k)}_{\tau B} - 8 c_2(0) \epsilon_{\,\,B}^{(0)\,C} F^{(k)}_{\tau C} \right] = 0
\end{multline}
Integrate this against $\sqrt{h^{(0)}}$ over $C(\tau)$ to get
\be
     2\lambda^2 l^k \int_{C(\tau)} \text{d}^{d-2}x \sqrt{h^{(0)}} |\phi^{(m+1)}|^2 A^{(k-2m-2)}_\tau =0 
\ee
where the integral over the second term vanished because it was a total derivative. We have $\partial_{A} A^{(k-2m-2)}_\tau\big|_{\rho=0} = -F^{(k-2m-2)}_{\tau A}\big|_{\rho=0} = 0$ by our inductive hypothesis, so $A^{(k-2m-2)}_\tau\big|_{\rho=0}$ is a constant. Hence we can take it out of the integral to get
\be
     A^{(k-2m-2)}_\tau\big|_{\rho=0} \int_{C(\tau)} \text{d}^{d-2}x \sqrt{h^{(0)}} |\phi^{(m+1)}|^2 =0 
\ee
However $\phi^{(m+1)}$ is the lowest order piece of $\phi$ that does not identically vanish on the horizon, and so the integral is non-zero. Therefore $A^{(k-2m-2)}_\tau\big|_{\rho=0}=0$. 

Plugging this back into (\ref{Etauorderk}), let us now integrate it against $\sqrt{h^{(0)}} A^{(k)}_\tau$ over $C(\tau)$. In a similar fashion to the $\phi^{(0)}|_{\rho=0}\not\equiv 0$ case, we get 
\be
\int_{C(\tau)} \text{d}^{d-2}x \sqrt{h^{(0)}} c_1(0) h^{(0) A B} \left(\partial_A A^{(k)}_\tau \right) \left( \partial_B A^{(k)}_\tau \right) =0
\ee
and thus $F^{(k)}_{\tau A}\big|_{\rho=0} = 0$. 

Finally, we look at $E_{\tau A}[\Phi_J]\big|_{\rho=0}$ at order $l^k$. This is
\begin{multline}
    \text{At order} \,\, l^k, \quad E_{\tau A}[\Phi_J]\Big|_{\rho=0} = -\frac{1}{2} l^k \partial_{A} X^{(k)}
    - \frac{1}{2} i l^k \lambda A^{(k-2m-2)}_\tau \Big(\phi^{(m+1)*} \partial_A \phi^{(m+1)} -
    \phi^{(m+1)} \partial_A \phi^{(m+1)*}\\
    - 2i\lambda A^{(0)}_A |\phi^{(m+1)}|^2\Big) -
    \frac{1}{2} c_1(0)\left( F^{(0)}_{A B} h^{(0) B C} - F^{(0)}_{\tau \rho} \delta_{A}^{C} \right) l^k F^{(k)}_{\tau C} = 0
\end{multline}
Substituting in $A^{(k-2m-2)}_\tau\big|_{\rho=0}=0$ and $F^{(k)}_{\tau A}\big|_{\rho=0} = 0$ we get $\partial_{A} X^{(k)}\big|_{\rho=0}=0$, and thus the induction proceeds and we have proved the Generalized 0th Law.

\subsubsection{Case 2: $\kappa^{(0)}=0$}

Finally we deal with the case $\kappa^{(0)}=0$. Moreover, we assume $\kappa^{[n]}=0$ for some $n\geq 0$. The proof now gets rather technical and mostly involves chasing powers of $l$. The physical relevance of this proof is questionable as we are heavily relying on analyticity in $l$, however we include it for completeness.

In this case we cannot apply Lemma \ref{philemma}. Our aim is still to prove $A_\tau|_{\rho=0} = 0$. Suppose we have an obstruction to this, in that $A_\tau^{[N]}|_{\rho=0} = 0$ but $A_\tau^{(N+1)}|_{\rho=0} \neq 0$ for some $N$. We will try to find a contradiction. In the previous cases we used the $2\lambda^2 A_\tau |\phi|^2 $ term in $E^{(0)}_{\tau}|_{\rho=0}$ to prove $A_\tau$ vanished at each order. However it is order $l^{2m+N+3}$, therefore we first need to be able to run the induction all the way up to that order before we can hope to conclude $A_\tau^{(N+1)}|_{\rho=0} = 0$. If $N<n$ it turns out we can prove a lemma that allows us to do this. 

Before stating and proving the lemma, it is worth emphasising some logic we will use repeatedly below. Suppose we have proved $\partial_A X^{[s]}|_{\rho=0}=0$ and $F_{\tau A}^{[s]}|_{\rho=0}=0$ for some $s\geq 0$. Then we can change to affinely parameterized co-ordinates via $\rho = r(\kappa^{[s]}v+1)$, $\tau = \frac{1}{\kappa^{[s]}}\log (\kappa^{[s]}v+1)$ in which all positive boost weight quantities involving $g^{[s]}_{\mu \nu}$, $F^{[s]}_{\mu\nu}$, $\phi^{[s]}$ and $A^{[s]}_\mu$ vanish on the horizon except from $\partial_{A_1}{... \partial_{A_n}{ \partial_{v}^{q} A^{[s]}_v } }$. But earlier we calculated
\be \label{orderdvAv}
    \partial^q_v A^{[s]}_v\big|_{r=0} = \frac{(-\kappa^{[s]})^q}{(\kappa^{[s]} v+1)^{q+1}} A^{[s]}_{\tau}\big|_{\rho=0}
\ee
Therefore if $\kappa^{[n]}=0$ and $A_\tau^{[N]}|_{\rho=0} = 0$ then $\partial^q_v A^{[s]}_v\big|_{r=0}$ is at least order $l^{q(n+1)+N+1}$ Taking $\partial_A$ derivatives does not change the order on the horizon, and so we won't explicitly mention them in the analysis going forward. Furthermore, we can calculate that the non-positive boost weight quantities made from $\phi^{[s]}$, namely $\partial^a_v \partial_r^b \phi^{[s]}$ with $b\geq a$ have the following form on the horizon:
\be \label{orderdrphi}
    \partial^a_v \partial_r^b \phi^{[s]}\big|_{r=0} = \frac{b!}{(b-a)!} \left(\kappa^{[q]}\right)^a \left(\kappa^{[q]}v+1\right)^{b-a} \partial_{\rho}^b \phi^{[s]}\big|_{\rho=0}
\ee
Therefore if we also happen to know that $\partial_{\rho}^b \phi^{[N_b]}\big|_{\rho=0}=0$ then $\partial^a_v \partial_r^b \phi^{[s]}\big|_{r=0}$ is at least order $l^{a(n+1)+N_b+1}$.

Onto the lemma: 

\begin{lemma} \label{Lemma2}
    If $\phi^{[m]}|_{\rho=0}=0$, $\kappa^{[n]}=0$, $A_\tau^{[N]}|_{\rho=0} = 0$ and $A_\tau^{(N+1)}|_{\rho=0} \neq 0$ with $N<n$ then 
    \be \label{kappa0lemma}
        \partial_A X^{[2m+N+2]}\big|_{\rho=0} = 0, \quad F^{[2m+N+2]}_{\tau A}\big|_{\rho=0} = 0,\,\,\text{ and }\,\,\forall p\geq 1\,\,  \partial_\rho^p \phi^{[m-p(N+1)]}\big|_{\rho=0} = 0 
    \ee
\end{lemma}

\begin{proof}
    If $A_\tau^{(0)}|_{\rho=0} \neq 0$, i.e. $N=-1$, then the proof follows in the same way as Lemma \ref{philemma}. This is because equation (\ref{drhophi}) becomes $-2i\lambda A^{(0)}_\tau \partial_{\rho}\phi^{(0)}|_{\rho=0}=0$, from which we can still conclude $\partial_{\rho}\phi^{(0)}|_{\rho=0}=0$. Similarly in (\ref{drhos1phi}) and (\ref{drhophin1}) we still find the $\rho$ derivatives of $\phi$ vanish, and so can conclude $\partial_\rho^p \phi^{[m]}\big|_{\rho=0} = 0$ $\forall p$ as (\ref{kappa0lemma}) requires. We then trivially get $\partial_A X^{[2m+2]}\big|_{\rho=0} = 0$ and $F^{[2m+2]}_{\tau A}\big|_{\rho=0} = 0$ from the induction detailed in the rest of Case 1 (and we get our contradiction that $A_\tau^{(0)}|_{\rho=0} = 0$ if there is some $m$ such that $\phi^{(m+1)}\not\equiv 0$).

    Therefore, let $N\geq 0$. We now proceed by induction on $0\leq k <m$ with hypothesis
    \be \label{inducthypo}
        \partial_A X^{[2k+N+1]}\big|_{\rho=0} = 0, \quad F^{[2k+N+1]}_{\tau A}\big|_{\rho=0} = 0,\,\,\text{ and }\,\,\forall p\geq 1\,\,  \partial_\rho^p \phi^{[k-p(N+1)]}\big|_{\rho=0} = 0 
    \ee
    For the base case $k=0$, we note that $A_\tau^{[N]}|_{\rho=0} = 0$ and $\partial_A X^{[N]}\big|_{\rho=0} = 0$ because $\kappa^{[n]}=0$ and $N<n$. This implies all positive boost weight quantities made from $\Phi^{[N]}$ vanish on the horizon. Therefore at order $l^{N+1}$, the higher derivative pieces in $E_{\tau}|_{\rho=0}$ and $E_{\tau A}|_{\rho=0}$ vanish, and using the same method as the real scalar field case we can prove $\partial_A X^{[N+1]}\big|_{\rho=0} = 0$ and $F^{[N+1]}_{\tau A}\big|_{\rho=0} = 0$. Furthermore $\partial_\rho^p \phi^{[-p(N+1)]}\big|_{\rho=0} = 0$ $\forall p\geq 1$ trivially.

    So let us assume the hypothesis holds for $k$. For $I=(\tau A)$ and $I=\tau$ we study
    \begin{equation}\label{order2kN2}
        \text{At order} \,\, l^{2k+N+2}, \quad E_{I}[\Phi_J]\big|_{\rho=0} = E^{(0)}_{I}[\Phi^{[2k+N+2]}_J]|_{\rho=0} + l \sum_{s=1}^{2k+N+2} l^{s-1} E^{(s)}_{I}[\Phi_J^{[2k+N+1]}]|_{\rho=0}
    \end{equation}
    We again change co-ordinates via $\rho = r(\kappa^{[2k+N+1]}v+1)$, $\tau = \frac{1}{\kappa^{[2k+N+1]}}\log (\kappa^{[2k+N+1]}v+1)$, in which $I=(\tau A)$ and $I=\tau$ are proportional to positive boost weight components. As discussed above, the only positive boost weight quantity made from $\Phi_J^{[2k+N+1]}$ that doesn't vanish on the horizon by the induction hypotheses is $\partial_v^q A_v^{[2k+N+1]}$. By gauge invariance, if an $A_\mu$ appears, so must a $\partial^a \phi^* \partial^b \phi$. Therefore the only positive boost weight terms left must have the combination
    \be \label{phiAphi}
        (\partial^{a_1}_v \partial^{b_1}_r \phi^*) \, (\partial_v^{q_1} A_v) ... (\partial_{v}^{q_M} A_v)\, (\partial^{a_2}_v \partial^{b_2}_r \phi)
    \ee
    for some $M\geq 1$, with $b_1\geq a_1$, $b_2\geq a_2$ and overall boost weight $a_1+a_2-b_1-b_2 + M + \sum q_i \geq 1$. But by (\ref{orderdvAv}), (\ref{orderdrphi}) and the inductive hypothesis (which implies $N_b=k-b(N+1)$), on the horizon this is of order
    \be
    \begin{split}
        2k+2 + \left(a_1+a_2+\sum q_i \right)(n+1)+&(M-b_1-b_2)(N+1)\\
        &\geq 2k + 2 + \left(a_1+a_2-b_1-b_2 + M + \sum q_i\right)(N+1)\\
        &\geq 2k+2 + N+1
    \end{split}
    \ee
    Therefore, since the higher derivative terms in (\ref{order2kN2}) come with at least one extra power of $l$, the remaining positive boost weight quantities (\ref{phiAphi}) cannot appear until order $l^{2k+N+4}$. Hence the higher derivative terms can be safely ignored at order $l^{2k+N+2}$ and we can get $\partial_A X^{[2k+N+2]}\big|_{\rho=0} = 0$ and $F^{[2k+N+2]}_{\tau A}\big|_{\rho=0} = 0$ using the same method as the real scalar field case ($k<m$ so the $A_\tau \phi^* \phi$ terms don't appear in $E_I^{(0)}$ at this order). Furthermore, we can repeat this at order $l^{2k+N+3}$ to get $\partial_A X^{[2k+N+3]}\big|_{\rho=0} = 0$ and $F^{[2k+N+3]}_{\tau A}\big|_{\rho=0} = 0$, which lets the first two inductive hypotheses proceed.

    Turning to the third hypothesis, we start with $p=1$, i.e. we show $\partial_\rho \phi^{(k+1-(N+1))}\big|_{\rho=0} = 0$. To do this we study
    \begin{equation}\label{Eorderk1}
        \text{At order} \,\, l^{k+1}, \quad E[\Phi_J]\big|_{\rho=0} = E^{(0)}[\Phi^{[k+1]}_J]|_{\rho=0} + l \sum_{s=1}^{k+1} l^{s-1} E^{(s)}[\Phi_J^{[k]}]|_{\rho=0}
    \end{equation}
    As ever, make the co-ordinate change $\rho = r(\kappa^{[k]}v+1)$, $\tau = \frac{1}{\kappa^{[k]}}\log (\kappa^{[k]}v+1)$ in the higher derivative terms. As discussed in Lemma \ref{philemma}, every term in $E[\Phi_J]$ is at least linear in $\phi$ or its derivatives. $\phi^{[m]}|_{\rho=0}=0$ so it cannot appear undifferentiated in the above. By (\ref{orderdrphi}) and the inductive hypothesis, zero boost weight derivatives $(\partial_v \partial_r)^a \phi^{[k]}$ are at least order $a(n+1) + k -a(N+1)+1\geq k+1$ on the horizon, and so cannot appear in the higher derivative terms due to the extra factor of $l$. Furthermore, all positive boost weight derivatives of $\phi^{[k]}$ also vanish on the horizon because $\partial_A X^{[k]}\big|_{\rho=0} = 0$. This leaves negative boost weight derivatives of $\phi$, however since $E[\Phi_J]$ is overall zero boost weight, these must come multiplied by positive boost weight factors. The only $\Phi^{[k]}_J$ positive boost weight quantities that are non-vanishing on the horizon are $\partial_v^q A_v^{[k]}$, hence we are left with combinations of the form
    \be \label{phiAv}
        (\partial^{a}_v \partial^{b}_r \phi) \, (\partial_v^{q_1} A_v) ... (\partial_{v}^{q_M} A_v)
    \ee
    with $b>a$ and $M\geq 1$ and overall non-negative boost weight $a-b + M + \sum q_i \geq 0$. But on the horizon this is of order
    \be
    \begin{split}
        k+1 + \left(a+\sum q_i \right)(n+1)+&(M-b)(N+1)\\
        &\geq k + 1 + \left(a-b + M + \sum q_i\right)(N+1)\\
        &\geq k+1
    \end{split}
    \ee
    and so once again cannot appear in the higher derivative terms at order $l^{k+1}$ due to the extra factor of $l$. Therefore we can safely ignore these terms, and just look at $E^{(0)}[\Phi^{[k+1]}_J]|_{\rho=0}$, which gives 
    \begin{equation} \label{E0orderk1}
        \text{At order} \,\, l^{k+1}, \quad E[\Phi_J]\big|_{\rho=0} = -2 i \lambda l^{k+1} A^{(N+1)}_{\tau} \partial_{\rho}\phi^{(k+1-(N+1))}\big|_{\rho=0} = 0
    \end{equation}
    therefore $\partial_{\rho}\phi^{(k+1-(N+1))}\big|_{\rho=0}=0$ as desired. Now induct on the number of $\rho$ derivatives, i.e. assume $\partial_\rho^p \phi^{[k+1-p(N+1)]}\big|_{\rho=0} = 0$ for all $p\leq s$ for some $s\geq 1$. Then look at $(\nabla_\rho)^s E[\Phi_J]=0$ on the horizon at order $l^{k+1-s(N+1)}$. Change co-ordinates and note that the result has overall boost weight is $-s$. Once again, it can be shown the higher derivative terms can be ignored by calculating the order of the terms like (\ref{phiAv}) that can appear and using the inductive hypotheses. We end up with
    \begin{equation}
        \text{At order} \,\, l^{k+1-s(N+1)}, \quad (\nabla_\rho)^s E[\Phi_J]\big|_{\rho=0} = -2 i \lambda l^{k+1-s(N+1)} A^{(N+1)}_{\tau} \partial^{s+1}_{\rho}\phi^{(k+1-(s+1)(N+1))}\big|_{\rho=0} = 0
    \end{equation}
    and therefore $\partial^{s+1}_{\rho}\phi^{(k+1-(s+1)(N+1))}\big|_{\rho=0} = 0$. This completes the induction over $\rho$ derivatives, which completes the overall induction. 

    We can run the induction until $k=m$ to get $\partial_A X^{[2m+N+1]}\big|_{\rho=0} = 0$, $F^{[2m+N+1]}_{\tau A}\big|_{\rho=0} = 0$, and $\forall p\geq 1\,\,  \partial_\rho^p \phi^{[m-p(N+1)]}\big|_{\rho=0} = 0$. We cannot go a full step further because $\phi^{(m+1)}|_{\rho=0}$ is not necessarily vanishing in $E^{(0)}[\Phi^{[m+1]}_J]|_{\rho=0}$ at order $l^{m+1}$ and so equation (\ref{E0orderk1}) would be much more complicated. However we can go one order further in $E_{I}[\Phi_J]|_{\rho=0}$ for $I=\tau$ and $I=(\tau A)$ because the problem terms (\ref{phiAphi}) now don't appear until order $l^{2m+N+4}$ in the higher derivative terms, and the $A_\tau \phi^* \phi$ terms don't appear in $E_I^{(0)}$ until order $l^{2m+N+3}$. Hence we can prove $\partial_A X^{[2m+N+2]}\big|_{\rho=0} = 0$, $F^{[2m+N+2]}_{\tau A}\big|_{\rho=0} = 0$ which completes the lemma.
\end{proof}

If $\phi$ vanishes to all orders on the horizon (i.e. we can take $m\rightarrow \infty$) then we see $\partial_A X\big|_{\rho=0}$ and $F_{\tau A}\big|_{\rho=0}$ also vanish to all orders and so the Generalized 0th Law holds. Therefore let us assume that $\phi^{(m+1)}(x^A)|_{\rho=0}$ is non-zero for some $x^A$. Let us try to go one step further in the induction and look at $E_{\tau}[\Phi_J]|_{\rho=0}$ at order $l^{2m+N+3}$. The higher derivative terms still vanish because terms of the form (\ref{phiAphi}) don't appear until order $l^{2m+N+4}$ as discussed above. However $A_\tau |\phi|^2\big|_{\rho=0} = l^{2m+N+3}A^{(N+1)}_\tau |\phi^{(m+1)}|^2\big|_{\rho=0}$ and so appears in $E^{(0)}_{\tau}[\Phi_J]|_{\rho=0}$:
\begin{multline}
    \text{At order} \,\, l^{2m+N+3}, \quad E_{\tau}[\Phi_J]\Big|_{\rho=0} = 2\lambda^2 l^{2m+N+3} A^{(N+1)}_\tau |\phi^{(m+1)}|^2\\
    + l^{2m+N+3} h^{(0) A B} \DD^{(0)}_A \left[ c_1(0) F^{(2m+N+3)}_{\tau B} - 8 c_2(0) \epsilon_{\,\,B}^{(0)\,C} F^{(2m+N+3)}_{\tau C} \right] = 0
\end{multline}
However, we can now repeat the same analysis as Case 1 by integrating against $\sqrt{h^{(0)}}$ over $C(\tau)$ to get $A^{(N+1)}_\tau|_{\rho=0}=0$, which contradicts our assumption that $A^{(N+1)}_\tau|_{\rho=0}\neq 0$! 

Lemma \ref{Lemma2} only holds for $N<n$, and hence this contradiction only applies up to $N=n-1$. Therefore we can conclude one of the following must be true from the logic so far: a) $A^{[n]}_\tau|_{\rho=0} = 0$, or b) $A^{[n]}_\tau|_{\rho=0} \neq 0$ and $\phi|_{\rho=0}\equiv 0$ to all orders in which case the Generalized 0th Law holds. 

Taking forward (a), if $\kappa$ vanishes to all orders (i.e. we can take $n\rightarrow \infty$) then so does $A_\tau$, which would prove the Generalized 0th Law. Therefore, we assume $\kappa^{(n+1)}\neq 0$, in which case we can prove another lemma:

\begin{lemma} \label{Lemma3}
    If $\phi^{[m]}|_{\rho=0}=0$, $A_\tau^{[n]}|_{\rho=0} = 0$, $\kappa^{[n]}=0$, and $\kappa^{(n+1)}|_{\rho=0} \neq 0$ then 
    \be \label{kappan1lemma}
        \partial_A X^{[2m+n+2]}\big|_{\rho=0} = 0, \quad F^{[2m+n+2]}_{\tau A}\big|_{\rho=0} = 0,\,\,\text{ and }\,\,\forall p\geq 1\,\,  \partial_\rho^p \phi^{[m-p(n+1)]}\big|_{\rho=0} = 0 
    \ee
\end{lemma}

\begin{proof}
    Follows using the same steps as the proof of Lemma \ref{Lemma2} with $N=n$ except we use the non-vanishing of $\kappa^{(n+1)}$ rather than $A^{(N+1)}$ to conclude the vanishing of $\rho$-derivatives of $\phi$. E.g. (\ref{E0orderk1}) becomes
    \begin{equation}
        \text{At order} \,\, l^{k+1}, \quad E[\Phi_J]\big|_{\rho=0} =l^{k+1} (X^{(n+1)} -2 i \lambda  A^{(n+1)}_{\tau}) \partial_{\rho}\phi^{(k+1-(n+1))}\big|_{\rho=0} = 0
    \end{equation}
    $X^{(n+1)}|_{\rho=0}$ is proved to be constant in the base case of the main induction, and is non-zero because $\kappa^{(n+1)}|_{\rho=0} \neq 0$. Therefore we can conclude $\partial_{\rho}\phi^{(k+1-(n+1))}\big|_{\rho=0}=0$. Higher $\rho$-derivatives follow similarly.
\end{proof}

We see once again that if $\phi$ vanishes to all orders on the horizon then the Generalized 0th Law holds. Therefore we are left with the final case to deal with: $\phi^{[m]}|_{\rho=0}=0$, $\phi^{(m+1)}(x^A)|_{\rho=0}$ is non-zero for some $x^A$, $\kappa^{[n]}=0$, $\kappa^{(n+1)}|_{\rho=0} \neq 0$, and $A_\tau^{[n]}|_{\rho=0} = 0$. We perform our final induction on this case, which has hypothesis
\be
    \partial_A X^{[k+2m+1]}\big|_{\rho=0} = 0, \quad F^{[k+2m+1]}_{\tau A}\big|_{\rho=0} = 0, \quad A_\tau^{[k-1]}|_{\rho=0} = 0
\ee
for $k\geq n+1$. The base case $k=n+1$ follows from Lemma \ref{Lemma3}, from which we also have $\forall p\geq 1,\,\,  \partial_\rho^p \phi^{[m-p(n+1)]}\big|_{\rho=0} = 0$.

We now assume true for $k$. For $I=(\tau A)$ and $I=\tau$ we study
\begin{equation} \label{finalinductorder}
       \text{At order} \,\, l^{k+2m+2}, \quad E_{I}[\Phi_J]\big|_{\rho=0} = E^{(0)}_{I}[\Phi^{[k+2m+2]}_J]|_{\rho=0} + l \sum_{s=1}^{k+2m+2} l^{s-1} E^{(s)}_{I}[\Phi_J^{[k+2m+1]}]|_{\rho=0}
\end{equation}
Change co-ordinates via $\rho = r(\kappa^{[k+2m+1]}v+1)$, $\tau = \frac{1}{\kappa^{[k+2m+1]}}\log (\kappa^{[k+2m+1]}v+1)$, in which $I=(\tau A)$ and $I=\tau$ are proportional to positive boost weight components. As in Lemma \ref{Lemma2}, the only positive boost weight quantities made from $\Phi_J^{[k+2m+1]}$ that don't necessarily vanish on the horizon by the induction hypotheses are in the combination
\be
    (\partial^{a_1}_v \partial^{b_1}_r \phi^*) \, (\partial_v^{q_1} A_v) ... (\partial_{v}^{q_M} A_v)\, (\partial^{a_2}_v \partial^{b_2}_r \phi)
\ee
for some $M\geq 1$, with $b_1\geq a_1$, $b_2\geq a_2$ and overall boost weight $a_1+a_2-b_1-b_2 + M + \sum q_i \geq 1$. But by (\ref{orderdvAv}), (\ref{orderdrphi}) and the inductive hypotheses, on the horizon this is of order
\be
\begin{split}
    2m+2 + \left(a_1+a_2-b_1-b_2+\sum q_i \right)(n+1)+Mk &\geq 2m + 2 + n+1+M(k-n-1) \\
    &\geq 2m+2 + k
\end{split}
\ee
where in the last step we used $M\geq 1$ and $k\geq n+1$. Therefore we can once again ignore the higher derivative terms in (\ref{finalinductorder}) because they come with an additional power of $l$. 

Hence, 
\begin{multline}
    \text{At order} \,\, l^{k+2m+2}, \quad E_{\tau}[\Phi_J]\Big|_{\rho=0} = 2\lambda^2 l^{k+2m+2} A^{(k)}_\tau |\phi^{(m+1)}|^2\\
    + l^{k+2m+2} h^{(0) A B} \DD^{(0)}_A \left[ c_1(0) F^{(k+2m+2)}_{\tau B} - 8 c_2(0) \epsilon_{\,\,B}^{(0)\,C} F^{(k+2m+2)}_{\tau C} \right] = 0
\end{multline}
Identically to Case 1, we integrate this against 
$\sqrt{h^{(0)}}$ and then against $\sqrt{h^{(0)}} A^{(k)}_\tau$ to get $A_\tau^{(k)}|_{\rho=0} = 0$ and $F^{(k+2m+2)}_{\tau A}\big|_{\rho=0} = 0$ respectively.

Finally,
\begin{multline}
    \text{At order} \,\, l^{k+2m+2}, \quad E_{\tau A}[\Phi_J]\Big|_{\rho=0} = -\frac{1}{2} l^{k+2m+2} \partial_{A} X^{(k+2m+2)}
    - \frac{1}{2} i l^{k+2m+2} \lambda A^{(k)}_\tau \Big(\phi^{(m+1)*} \partial_A \phi^{(m+1)}\\
    - \phi^{(m+1)} \partial_A \phi^{(m+1)*} - 2i\lambda A^{(0)}_A |\phi^{(m+1)}|^2\Big) -
    \frac{1}{2} c_1(0)\left( F^{(0)}_{A B} h^{(0) B C} - F^{(0)}_{\tau \rho} \delta_{A}^{C} \right) l^{k+2m+2} F^{(k+2m+2)}_{\tau C} = 0
\end{multline}
into which we substitute $A_\tau^{(k)}|_{\rho=0} = 0$ and $F^{(k+2m+2)}_{\tau A}\big|_{\rho=0} = 0$ to get $\partial_{A} X^{(k+2m+2)}\big|_{\rho=0} = 0$. Hence the induction proceeds and we have proved the Generalized 0th Law.

\subsection{Full Expressions for $\nabla^{\beta} F_{\alpha \beta}$ in affinely parameterized GNCs} \label{dFAppendix}

The components of $\nabla^{\beta} F_{\alpha \beta}$ in affinely parameterized GNCs are as follows:

\be
\begin{split}
    \nabla^{\beta}F_{v \beta} = &\partial_{v}{\psi} + D^{A}{K_{A}}+K \psi+ r \big(-D^{A}{\psi} \beta_{A}+\beta^{A} \partial_{r}{K_{A}}-D^{A}{\beta^{B}} F_{A B}-\beta^{A} \beta_{A} \psi-\bar{K}^{A} \partial_{v}{\beta_{A}}+\\
    &K^{A} \bar{K} \beta_{A}-D^{A}{\beta_{A}} \psi-2K^{A} \bar{K}_{A}\,^{B} \beta_{B}\big) +{r}^{2} \big(-\alpha \partial_{r}{\psi}-\beta^{A} \beta_{A} \partial_{r}{\psi}-F^{A B} \beta_{A} \partial_{r}{\beta_{B}}+\\
    &D^{A}{\beta^{B}} \bar{K}_{A} \beta_{B}+D^{A}{\alpha} \bar{K}_{A}-D^{A}{\beta^{B}} \bar{K}_{B} \beta_{A}+\bar{K}^{A} \alpha \beta_{A}-\beta^{A} \partial_{r}{\beta_{A}} \psi-\bar{K} \alpha \psi-\bar{K} \beta^{A} \beta_{A} \psi+\\
    &2\bar{K}^{A B} \beta_{A} \beta_{B} \psi\big) +{r}^{3} \left(-\bar{K}^{A} \alpha \partial_{r}{\beta_{A}}-\bar{K}^{A} \beta^{B} \beta_{B} \partial_{r}{\beta_{A}}+\bar{K}^{A} \beta_{A} \beta^{B} \partial_{r}{\beta_{B}}+\bar{K}^{A} \beta_{A} \partial_{r}{\alpha}\right)
\end{split}
\ee

\be 
    \nabla^{\beta}F_{r \beta} = -\partial_{r}{\psi}+D^{A}{\bar{K}_{A}}+\bar{K}^{A} \beta_{A}-\bar{K} \psi+ r \left(\beta^{A} \partial_{r}{\bar{K}_{A}}-2\bar{K}^{A} \bar{K}_{A}\,^{B} \beta_{B}+\bar{K}^{A} \partial_{r}{\beta_{A}}+\bar{K} \bar{K}^{A} \beta_{A}\right)
\ee

\be \label{FullFA}
\begin{split}
    \nabla^{\beta}F_{A \beta} = &-\partial_{r}{K_{A}}-\partial_{v}{\bar{K}_{A}} + D^{B}{F_{A B}}+2K_{A B} \bar{K}^{B}+2K^{B} \bar{K}_{A B}-\beta_{A} \psi-K \bar{K}_{A}-K_{A} \bar{K}+F_{A}\,^{B} \beta_{B}+\\
    &r \big(-F^{B C} \bar{K}_{A B} \beta_{C}+D^{B}{\beta_{A}} \bar{K}_{B}+\bar{K}_{A}\,^{B} \beta_{B} \psi-\partial_{r}{\beta_{A}} \psi+\frac{1}{2}\bar{K}^{B} \beta_{A} \beta_{B}+F_{A}\,^{B} \bar{K} \beta_{B}-\\
    &D^{B}{\beta_{B}} \bar{K}_{A}+F_{A}\,^{B} \partial_{r}{\beta_{B}}-\bar{K}_{A} \beta^{B} \beta_{B}-2\bar{K}_{A} \alpha\big) +{r}^{2} \big(\bar{K}^{B} \bar{K}_{A B} \alpha+\bar{K}^{B} \bar{K}_{A B} \beta^{C} \beta_{C}+\\
    &\frac{1}{2}\bar{K}^{B} \beta_{B} \partial_{r}{\beta_{A}}-\bar{K} \bar{K}_{A} \alpha-\bar{K} \bar{K}_{A} \beta^{B} \beta_{B}-\bar{K}_{A} \beta^{B} \partial_{r}{\beta_{B}}-\bar{K}_{A} \partial_{r}{\alpha}\big)
\end{split}
\ee

These were calculated using the symbolic algebra programme Cadabra \cite{Cadabra1}\cite{Cadabra2} \cite{Cadabra3}. To get equation (\ref{EA0}), we use (\ref{dvKA}) to eliminate $\partial_{v}{\bar{K}_{A}}$ in (\ref{FullFA}).

\end{document}